\documentclass[english]{article}
\usepackage[T1]{fontenc}
\usepackage{float}
\usepackage[latin9]{inputenc}
\usepackage{amsmath,amsthm} 
\usepackage{graphicx}
\usepackage{amssymb}
\usepackage{enumitem}
\usepackage{natbib} 
\usepackage{bussproofs}
\usepackage{linguex,nicefrac,stmaryrd}
\usepackage[letterpaper]{geometry} 
\usepackage{appendix}

\usepackage{mathrsfs}
\usepackage{subcaption}
\usepackage[dvipsnames]{xcolor}

\setcitestyle{round}

\definecolor{darkgreen}{rgb}{0.0, 0.7, 0.0}

\DeclareSymbolFont{symbolsC}{U}{txsyc}{m}{n}
\DeclareMathSymbol{\boxright}{\mathrel}{symbolsC}{128}

\usepackage{thmtools,thm-restate}

\usepackage{tikz}
\usepackage{pgf} 
\usetikzlibrary{patterns,automata,arrows,shapes,snakes,topaths,trees,backgrounds,positioning,through,calc}
\usepackage{setspace}
\usepackage{multicol}
\usepackage{graphicx}

\theoremstyle{definition}
\newtheorem{theorem}{Theorem}[section]

\newtheorem{proposition}[theorem]{Proposition}
\newtheorem{corollary}[theorem]{Corollary}
\newtheorem{definition}[theorem]{Definition}

\newtheorem{lemma}[theorem]{Lemma}
\newtheorem{fact}[theorem]{Fact}
\newtheorem{example}[theorem]{Example}

\newtheorem{remark}[theorem]{Remark}

\input{fitch.sty}

\usepackage{mdframed}

\makeatletter  

\usepackage{stmaryrd}  
\usepackage{comment}
\usepackage[breaklinks]{hyperref}
\hypersetup{colorlinks=true, citecolor=blue, linkcolor=blue}  

\usepackage{cleveref}

\usepackage[draft,inline,nomargin,marginclue]{fixme}
\definecolor{cobalt}{RGB}{0,71,171}
\definecolor{brickred}{RGB}{203, 65, 84}

\FXRegisterAuthor{mm}{a,,}{\color{cobalt}MM}
\FXRegisterAuthor{wh}{awh}{\color{brickred}WH}
\fxsetup{inlineface=\sffamily}
\fxsetup{marginface=\sffamily}

\newcommand{\Tautequiv}{\;\rotatebox[origin=c]{180}{$\vDash$}\vDash}

\makeatother

\usepackage{babel} 
\onehalfspace
 
\begin{document} 

\title{The Orthologic of Epistemic Modals}
 \author{Wesley H. Holliday$^\dagger$ and Matthew Mandelkern$^\ddagger$ \\ \\$\dagger$ University of California, Berkeley {\normalsize (\href{mailto:wesholliday@berkeley.edu}{wesholliday@berkeley.edu})} \\ $\ddagger$ New York University {\normalsize (\href{mailto:mandelkern@nyu.edu}{mandelkern@nyu.edu})}}

\date{Forthcoming in \textit{Journal of Philosophical Logic}.}

\maketitle

\begin{abstract} Epistemic modals have peculiar logical features that are challenging to account for in a broadly classical framework. For instance, while a sentence  of the form $p\wedge\Diamond\neg p$  (`$p$, but it might be that not~$p$') appears to be a contradiction, $\Diamond\neg p$ does not entail $\neg p$, which would follow in classical logic. Likewise, the classical laws of distributivity and disjunctive syllogism fail for  epistemic modals. Existing attempts to account for these facts generally either under- or over-correct. Some theories predict that $ p\wedge\Diamond\neg p$, a so-called  \emph{epistemic contradiction}, is a contradiction only in an etiolated sense, under a notion of entailment that does not always allow us to replace $p\wedge\Diamond\neg p$ with a contradiction; these theories underpredict the infelicity of embedded epistemic contradictions. Other theories savage classical logic, eliminating not just rules that intuitively fail, like distributivity and disjunctive syllogism,  but also rules like non-contradiction, excluded middle, De Morgan's laws, and disjunction introduction, which intuitively remain valid for epistemic modals. In this paper, we aim for a middle ground, developing a semantics and logic for epistemic modals that makes epistemic contradictions genuine contradictions and that invalidates distributivity and disjunctive syllogism but that otherwise preserves classical laws that intuitively remain valid. We start with an \textit{algebraic semantics}, based on ortholattices instead of Boolean algebras, and then propose a \textit{possibility semantics}, based on partial possibilities related by compatibility. Both semantics yield the same consequence relation, which we axiomatize. We then show how to lift an arbitrary possible worlds model for a non-modal language to a possibility model for a language with epistemic modals. The goal throughout is to retain what is desirable about classical logic while  accounting for the non-classicality of epistemic vocabulary.\end{abstract}

{\small\hspace{.14in}\textbf{Keywords:} epistemic modals, negation, connectives, ortholattices, algebraic semantics, possibility semantics}

\section{Introduction}\label{Intro}

Exploration of epistemic modals in the last decades has shown that they do not fit easily into the framework of  classical modal logic. Following Yalcin \citeyearpar{Yalcin2007}, we can characterize the problem as follows. There is strong evidence that sentences of the form $p\wedge\lozenge\neg p$  or $\neg p\wedge \lozenge p$ (as well as variants that reverse the order) are contradictory: not only are they unassertable, but also they  embed like contradictions across the board. To get a sense for the evidence, note the infelicity of the following:

\ex. \label{butler}\a. \# The butler is the murderer, but he might not be. 
\b. \# Suppose that the butler is the murderer, but he might not be. 
\c. \#Everyone who is tall might not be tall.

In light of these data, it is natural to consider a logic that makes sentences of the form $ p\wedge\lozenge\neg p$ contradictory. The problem is that adding $ p\wedge\lozenge \neg p\vDash \bot$ as an entailment to any  modal logic that extends classical logic immediately yields the entailment $\lozenge \neg p\vDash \neg p$. But $\lozenge \neg p$ plainly does not entail $\neg p$. That the butler might not be the murderer does not entail that the butler \textit{is not} the murderer.

A profusion of responses to this problem have been developed in recent years. These responses either reject the assumption that $p\wedge\lozenge\neg p$ is truly a contradiction or  hold that it is a contradiction only under a non-classical notion of $\vDash$, one that blocks the reasoning above. For a prominent example, consider the dynamic system of \citet{GSV:1996}, wherein $p\wedge\lozenge\neg p$ is a contradiction,\footnote{Provided $p$ is free of modals or conditionals.} but conjunction and entailment are interpreted non-classically, so that we cannot conclude that $\lozenge \neg p\vDash \neg p$ (\citealt{GSV:1996}). The non-classicality of this approach, however, goes very deep: for instance, the classical laws of non-contradiction and excluded middle also fail in the dynamic system, but we know of no evidence for their failure from modal language. 
Our goal in this paper is to develop a more minimalist non-classical approach to epistemic modals that validates  $ p\wedge\lozenge\neg p\vDash \bot$ without validating $\lozenge \neg p\vDash \neg p$.  In particular, we block the inference from $ p\wedge\lozenge\neg p\vDash \bot$ to $\lozenge \neg p\vDash \neg p$ by treating negation, algebraically speaking, as an \emph{orthocomplementation} (the complement operation characteristic of ortholattices) but not necessarily a \emph{pseudocomplementation} (the complement operation of Heyting and Boolean algebras, obeying the law that  $a\wedge b=0$ implies $b\leq\neg a$). However, we aim to depart from classical logic in a minimal way, by invalidating only those classical inference patterns that have intuitive counterexamples involving epistemic modals. We note in particular that those intuitive counterexamples always involve combinations of sentences with different levels of epistemic iteration: for example, when $p$ itself is modal-free, $p\wedge\lozenge\neg p$ conjoins a sentence $p$ with no epistemic modality and a sentence $\lozenge\neg p$ with one level of epistemic modality. Hence we develop a system that is fully classical when reasoning with sentences of the same ``epistemic level.'' On this picture,  classical reasoning is dangerous only when it crosses different epistemic levels.

We start in \S~\ref{des} by drawing out intuitive desiderata for a logic of epistemic modals. We then characterize an \emph{epistemic orthologic} that meets those desiderata, first algebraically in \S~\ref{algebraic} and then using a possibility semantics in \S~\ref{PossSem}. Both semantics yield the same consequence relation. We provide a sound and complete axiomatization for this epistemic orthologic, following the call in \citealt{HollidayIcard2018} to axiomatize proposed consequence relations in natural language semantics.  In \S~\ref{Epistemicization}, we show how to lift an arbitrary possible worlds model for a non-modal language to a possibility model for a language with epistemic modals,  which yields an implementation of our possibility semantics built on the more familiar primitives of possible worlds semantics.  Finally, we compare our approach to existing approaches in \S~\ref{Comparisons} and conclude in \S~\ref{Conclusion}.

An associated online repository (\href{https://github.com/wesholliday/ortho-modals}{https://github.com/wesholliday/ortho-modals}) 
contains a Jupyter notebook with code to use our logic with the Natural Language Toolkit's interface to the Prover9 theorem prover and Mace4 model builder, as well as a Jupyter notebook with code for using the possibility semantics of \S~\ref{PossSem}. 

\section{Desiderata}\label{des}

We will begin by  bringing out the key desiderata for the logic of epistemic modals. 
In particular, our goal in this section and the next is to identify properties of an entailment relation appropriate to a language with epistemic modality.
As a methodological preliminary, let us say a bit more about the kind of entailment relation we aim to characterize. Our target consequence relation is a classical one in the sense that it aims to capture universal preservation of truth. 
This has two upshots worth noting. First, $\varphi$ entails $\psi$ iff $\varphi$ is semantically equivalent to $\varphi\wedge\psi$, where two  formulas $\varphi, \psi$ are semantically equivalent iff for any formulas $\chi$ and $\rho$, $\chi[\nicefrac{\varphi}{\rho}]$ is true iff $\chi[\nicefrac{\psi}{\rho}]$ is true (where $\chi[\nicefrac{\varphi}{\rho}]$ is the sentence that results from uniformly substituting $\varphi$ for $\rho$ in $\chi$; we will ignore potential issues about hyperintensionality here). 
Second,  if $\varphi$ entails $\psi$, then the probability of $\psi$ must be at least that of $\varphi$  (since on any natural notion of probability, the probability of semantically equivalent propositions will be equal, and the probability of a conjunct will always be at least as great as the probability of a conjunction). 

We highlight these two upshots of the classical consequence relation because they yield two ways to empirically test our target notion of entailment. 
For a brief illustration, consider the inference from $\varphi$ to $\Box \varphi$, which is valid in many logics for epistemic modals. We will invalidate this inference, and indeed, it is invalid in a logic that aims to capture inferences that universally preserve truth.  For if $\varphi$  entailed $\Box\varphi$ in this sense, then the probability of $\Box \varphi$ would always have to be at least as great as that of $\varphi$, but this is wrong. Consider a fair coin that was just flipped; suppose we do not know the outcome of the flip. Plausibly, the probability that it  landed heads is around .5, but the probability that it \emph{must} have landed heads is much lower, around 0. 
Second, since the present notion of entailment contraposes (given classical assumptions about De Morgan's laws and negation that we will not question here), if $\varphi\vDash\Box \varphi$, then $\neg\Box\varphi\vDash\neg\varphi$; but the latter is plainly implausible, since $\neg\Box\varphi$ is equivalent to $\lozenge\neg \varphi$, and clearly $\lozenge\neg \varphi$ does not entail $\neg \varphi$ (from `It might not be raining' it does not follow that it is not raining).
Finally, given that most accept that $\Box\varphi$ entails $\varphi$, if $\varphi$ also entailed $\Box\varphi$ then they would be logically equivalent. But then they would be everywhere substitutable for one another, salva veritate. However, they are clearly not: $\neg\varphi$ and $\neg\Box\varphi$ are not equivalent (compare `it's not raining' to `it's not the case that it \emph{must} be raining'). 

There are other kinds of entailment relations that we can also  characterize: for instance, prominently, we might develop a logic characterizing the inferences that preserve rational acceptance in all cases  (\citealt{Yalcin2007}, \citealt{Bledin:2014,Bledin:2015}, \citealt{Santorio:2018a}, \citealt{Norlin:2019,Norlin:2020}). These logics are well worth studying but are not our central target here. In the conclusion, we briefly discuss how to characterize a logic of acceptance  in our framework. In such a logic---unlike ours---the inference from $\varphi$ to $\Box\varphi$ might reasonably be seen to be valid. However, we do not see any sense in which our project (characterizing a logic of truth preservation) and the project of characterizing a logic of acceptance are in conflict. There is a purely terminological question about which of these deserves the name of `logic', a question whose interest we do not see; we are happy to use `logic' in the traditional way for the set of truth-preserving inferences, but if others prefer different usages, feel free to substitute your preferred terminology. There could also be a substantive question if someone thought that we should \emph{only} characterize the logic of acceptance and that the logic of truth-preservation is uninteresting or irrelevant. But such a position has little merit and (unsurprisingly) has not been defended in the literature to our knowledge. For even if we had an adequate characterization of the logic of acceptance, there would remain the further question of which inferences preserve probability and of which sentences are always substitutable salva veritate. These questions are not answered by a logic of acceptance (since we might have, e.g., that $\varphi$ entails $\Box\varphi$ on that notion, but the inference does not preserve probability, nor are $\varphi$ and $\Box\varphi$ substitutable). So everyone interested in core semantic notions needs to characterize the logic of truth preservation. It is natural to think that a logic of rational acceptance will then supervene on this logic, together with a notion of rational acceptance (and that is indeed how such a logic is defined in, e.g., \citealt{Yalcin2007}). But even if you reject this supervenience claim, you still need a logic of truth preservation, since it is not plausible that the logic of truth preservation supervenes on the logic of rational acceptance. 

We take these points to be uncontroversial, but questions about our target notion of entailment have frequently been raised in conversation and are particularly pertinent to the literature on epistemic modals (where it is often claimed that $\varphi$ entails $\Box\varphi$---which, again, we are happy to accept in one sense, but not in our target sense of entailment). 
Let us now turn to our  substantive desiderata for our target logic.

\subsection{Epistemic contradictions}\label{EpContra}

First, we give more evidence that sentences of the form  $p\wedge\lozenge\neg p$ and $\neg p\wedge \lozenge p$  are indeed contradictory, drawing on observations in \citealt{GSV:1996}, \citealt{Aloni:2000}, \citealt{Yalcin2007,YalcinDeRe}, and \citealt{Mandelkern:2018a}. Yalcin calls sentences of this form \emph{epistemic contradictions}. \citet{Mandelkern:2018a} calls epistemic contradictions and variants that reverse their order (that is, sentences of the form $\lozenge\neg p\wedge p$ or $\lozenge p\wedge \neg p$) \emph{Wittgenstein sentences}. Although this is somewhat controversial (and runs counter to the prominent dynamic approach mentioned above), we think that order does not matter to our central points,  so we intend the claims we make here to apply to all Wittgenstein sentences. However, for brevity, we often use epistemic contradictions as a stand-in for all Wittgenstein sentences. 

To start, note that Wittgenstein sentences are generally unassertable. It is hard to think of circumstances in which any of the following can be asserted: 

\ex. \a. \#It's  raining, but it might not be.
\b. \#It might be raining, but it isn't.
\c. \#The cat might be under the bed, but she is on the couch.
\d. \#I'll make lasagna, but I might not make lasagna.
\e. \#Sue isn't the winner but she might be. 

Unassertability, however, could plausibly be due to either pragmatic or semantic factors. Indeed, these sentences are superficially similar to Moore sentences (as \citet{Wittgenstein:1953} observed), and the latter are standardly explained on a pragmatic basis: you cannot assert, say, `It's raining, but I don't know that', because you cannot know this, on pain of contradiction, and you need to know what you assert.\footnote{Cf.~\citealt[\S~4.11]{Hintikka2005}. For a study of Moore sentences in the context of epistemic logic and dynamic epistemic logic, see \citealt{Holliday2010}.}

To test whether the unassertability of Wittgenstein sentences is pragmatic or semantic, we can look at how they embed in various environments, comparing their behavior with embedded Moore sentences on the one hand and embedded classical contradictions on the other, building on the structure of arguments of \citealt{Yalcin2007}. So consider the three-way comparisons below, with first the Moore sentence, then the Wittgenstein sentence, and finally the contradiction: 

\ex. \a. Suppose Sue is the winner but I don't know it. 
\b. \# Suppose Sue is the winner but she might not be.
\c. \# Suppose Sue is the winner and isn't the winner. 

\ex. \a. If Sue is the winner but I don't know it, I will chastise her for not telling me. 
\b. \# If Sue is the winner but she might not be, I will chastise her for not telling me. 
\c. \# If Sue is the winner and isn't the winner, I will chastise her for not telling me. 

\ex. \label{drekm}\a. It could be that Sue is the winner but I don't know it.
\b. \# It could be that Sue is the winner and might not be.
\c. \# It could be that Sue is the winner and isn't the winner.

\ex. \label{drek}\a. Either Sue is the winner but I don't know it, or she isn't the winner but I don't know it.
\b. \# Either Sue is the winner but might not be, or she isn't the winner but might be.
\c. \# Either Sue is the winner and isn't the winner, or she isn't the winner and is the winner.

\ex. \label{definiteec}\a. The winner is, for all I know, not the winner. 
\b. \# The winner might not be the winner.
\c. \# The winner isn't the winner.

\ex. \label{indefiniteec}\a. Someone is the winner and for all I know isn't the winner. 
\b. \# Someone is the winner and might not be the winner.
\c. \# Someone is the winner and isn't the winner. 

In all these cases, the first variant, embedding a Moore sentence, is felicitous, supporting the idea that the infelicity of Moore sentences is pragmatic. By contrast, both of the latter variants are infelicitous, suggesting that the infelicity of Wittgenstein sentences is not pragmatic, after all, but instead due to the fact that they are genuine contradictions.

This is not to say that Wittgenstein sentences and classical contradictions pattern in exactly the same ways. It is well known that classical contradictions are often judged by ordinary speakers to be interpretable (`Sue is the winner and isn't the winner' could be used to express that Sue is the winner in one sense but not in some other sense), and Wittgenstein sentences, embedded or not, can similarly be coerced (though perhaps with more difficulty) into slightly different communicative functions. Nonetheless, their behavior across the board is strikingly like that of contradictions. We conclude that they are indeed contradictions and that any residual differences between $p\wedge \neg p$ and $ p\wedge\lozenge\neg p$ should be accounted for on the basis of their underlying compositional semantics.

This presents us with our first desideratum: a logic on which Wittgenstein sentences are contradictions. That is, we want $ \varphi\wedge\lozenge \neg \varphi\vDash \bot$ (and likewise $ \neg \varphi\wedge\lozenge  \varphi\vDash \bot$,  $\lozenge \varphi\wedge \neg \varphi\vDash \bot$, and  $\lozenge \neg \varphi\wedge  \varphi\vDash \bot$). In fact, we need more than this: we need a system in which $\varphi\wedge\Diamond\neg \varphi$ can always be replaced with $\varphi\wedge\neg \varphi$ \emph{salva veritate}. 
Otherwise, we would not have an immediate account of the data above, even if  $ \varphi\wedge\lozenge\neg  \varphi\vDash \bot$. This is worth mentioning since some systems, like domain semantics (\citealt{Yalcin2007}), predict that $ \varphi\wedge\lozenge \neg \varphi\vDash \bot$ but not that $\varphi\wedge\Diamond\neg \varphi$ can always be replaced with $\varphi\wedge \neg \varphi$ \emph{salva veritate} and hence miss some of the data above (see \S~\ref{Comparisons} for further discussion). Call a logic an \emph{\textsf{E}-logic} (suggesting \textit{epistemic}) if it satisfies these first two desiderata. 

We have already seen why an \textsf{E}-logic is hard to obtain: given classical assumptions, treating Wittgenstein sentences as contradictions would make $\lozenge\neg p$ entail $\neg p$. 
So, in particular, we want a logic that is an \textsf{E}-logic but that does not yield the absurd conclusion that $\lozenge\neg p\vDash\neg p$.
 Lest this problem be treated as having to do essentially with the peculiarities of `and', we should note that the problem can equally be formulated without involving conjunction at all. If we take the evidence above to show that $p$ and $\lozenge\neg p$ are jointly inconsistent, then what we need is a logic on which $\{\lozenge\neg p,p\}\vDash\bot$. Classical reasoning would  yield $\lozenge \neg p\vDash\neg p$. So we need a logic that blocks this reasoning. 
Of course, evidence that $p$ and $\lozenge \neg p$ are jointly inconsistent is somewhat less direct than the evidence that $p\wedge\lozenge \neg p$ is itself inconsistent. After all, we cannot  embed a pair of propositions under a unary sentential operator, as we did with the corresponding conjunction. But the sentences above already yield indirect evidence that it is not just the conjunction that is inconsistent, but also its conjuncts are jointly inconsistent, since the infelicity persists when the conjuncts are distributed across the restrictor and scope of quantifiers, as in \ref{definiteec} and \ref{indefiniteec}. Similar evidence comes from pairs of attitude ascriptions, where again the infelicity persists when the conjuncts in question are distributed under distinct attitude predicates:

 \ex. \a. \# Liam believes that Sue is the winner. Liam also believes that Sue might not be the winner.
 \b. \# Suppose that Sue is the winner. Suppose, further, that Sue might not be the winner.
 \c. \# I hope Susie wins. I also hope she might not win.
 
Finally, another way to see the puzzle of epistemic contradictions, also emphasized by \citet{Yalcin2007}, is to approach it from intuitions about meaning and synonymy rather than intuitions about logic. It is very natural to think that $\lozenge p$ means something like `For all we know, $p$ is true'. This is roughly the meaning it has in the influential approach of \citealt{Kratzer1977,Kratzer:1981}, for instance. But if that is right, then $ p\wedge \lozenge\neg p$ should have the meaning of a Moore sentence and should be embeddable just like Moore sentences are. But the examples above clearly show  that this is not the case. So apparently $\lozenge p$ does not mean even roughly the same thing as `For all we know, $p$ is true'. But then what does it mean? 

\subsection{Distributivity}\label{dissubsection}

We have seen one reason to think that an adequate treatment of epistemic modals calls for revision to classical logic. A second reason, brought out in \citealt{Mandelkern:2018a}, has to do with the distributive law, according to which $\varphi\wedge (\psi\vee \chi)$ is logically equivalent to $ (\varphi\wedge \psi)\vee (\varphi\wedge \chi)$ for  any sentences $\varphi,\psi,\chi$. This is, of course, a law of classical logic. But while it is intuitively valid for modal-free fragments, it is not intuitively valid once epistemic modals are on the scene; for compare:

\ex.\a. \label{suea} Sue might be the winner and she might not be, and either she is the winner or she isn't.
\b. \label{suec} \#Sue might not be the winner and she is the winner, or else Sue might be the winner and she isn't the winner.

While \ref{suea} sounds simply like a long-winded avowal of ignorance, \ref{suec} is very strange. But if distributivity is valid, then \ref{suea} entails \ref{suec}! (In fact, given distributivity, \ref{suea} and \ref{suec} are logically \emph{equivalent} provided that $p$ entails $\Diamond p$, on which more shortly.) Such examples motivate us to look for a logic that invalidates distributivity for the modal fragment. 

The ramifications of the failure of distributivity are plausibly very wide. For instance, these failures might also help explain a puzzling observation, which is based on discussion in \citealt{Ninan:2017} and related discussion in \citealt{Aloni:2000}. Consider a fair lottery where at least one ticket will win, but not all will win. The winning ticket(s) have been drawn, but we do not know which tickets won. 
Then the following seems true:

\ex. Every ticket might not be a winning ticket.\label{everymight}

But we cannot conclude:

\ex. Some winning ticket might not be a winning ticket.\label{winmight}

The winning tickets, however, are among the tickets, and so if every ticket might not be a winning ticket, it seems that we should be able to conclude that some winning ticket might not be a winning ticket.

We can reconstruct this inference as follows. Suppose we have rigid designators for all the tickets: $t_1, \dots  , t_n$. Then \ref{everymight} can be taken to be the conjunction  $\lozenge \neg W(t_1)\wedge \dots \wedge \lozenge \neg W(t_n)$, where $W$ stands for `is a winner'. We also know the disjunction $W(t_1)\vee\dots \vee W(t_n)$. Putting these two together, we have \[\big(\lozenge \neg W(t_1)\wedge \dots \wedge\lozenge \neg W(t_n)\big)\wedge \big(W(t_1)\vee\dots\vee W(t_n)\big);\] in short, every ticket might not be a winner, but some ticket is a winner. Distributivity would allow us to infer \[\big(W(t_1)\wedge \lozenge \neg W(t_1)\big) \vee\dots \vee \big(W(t_n)\wedge \lozenge \neg W(t_n)\big),\] that is, `Some winning ticket might not be a winning ticket', as in \ref{winmight}. However, if distributivity is not valid, we cannot conclude \ref{winmight} from \ref{everymight}. 
We will set aside quantification in this paper, but this example illustrates the importance of distributive reasoning---and identifying its failures---across a variety of examples involving epistemic modals.

\subsection{Disjunctive syllogism}

A closely related point is that disjunctive syllogism intuitively fails for epistemic modals (\citealt[citing Yalcin]{KlinedinstRothschildConnectives}). Disjunctive syllogism says that $\{p\vee q, \neg q\}\vDash p$. Hence the following inference, varying an example from Klinedinst and Rothschild, is valid if disjunctive syllogism is:

\ex. \a. Either the dog is inside or it must be outside.\label{doga}
\b. It's not the case that the dog must be outside.\label{dogb} 
\c. Therefore, the dog is inside. \label{dogc}

This inference is intuitively invalid. For \ref{doga} feels a lot like the corresponding tautology:

\ex. The dog is inside or outside.

And indeed, in any \textsf{E}-logic that validates De Morgan's laws (plus double negation elimination and duality for $\lozenge $ and $\Box$), \ref{doga} is a logical truth: $p\vee\Box \neg p$ is equivalent to $\neg (\neg p\wedge\lozenge p)$, which, if the logic is an \textsf{E}-logic, is equivalent to $\neg \bot$. Again, this seems intuitive: that is, sentences of the form $p\vee\Box\neg p$ do feel a lot like the corresponding tautologies $p\vee \neg p$, as in \ref{fatima} (an intuition that our system will capture by predicting that $p\vee\Box\neg p$ is indeed a tautology):

\ex.\label{fatima}
\a. Either Fatima is home or else she must be somewhere else.
 \b. Either John must be the culprit or else it isn't him. 
 
Thus, we know \ref{doga} just on the basis of the logic of epistemic modals. And \ref{dogb} (which is equivalent, by duality, to `the dog might be inside') can be true without  the dog in fact being inside. (Given that \ref{doga} has the status of a logical truth, and given duality, this is just the point, again, that  $\lozenge p$ does not entail $p$.) Hence disjunctive syllogism is not valid for a fragment including epistemic modals.

Given very weak assumptions, disjunctive syllogism follows from distributivity. So the intuitive failure of disjunctive syllogism gives us another reason to invalidate distributivity. 

\subsection{Orthomodularity}\label{OrthomodSec}

Another law of classical logic---and indeed of the weaker system of quantum logic (see, e.g., \citealt{Chiara2002})---that we have reason to invalidate is orthomodularity, according to which if $\varphi\vDash \psi$, then $ \psi\vDash  \varphi\vee(\neg \varphi\wedge \psi) $.\footnote{That $\varphi\vDash \psi$ implies $\varphi\vee(\neg \varphi\wedge \psi)\vDash \psi$ holds by disjunction and conjunction elimination, in which case orthomodularity can be stated as: if $\varphi\vDash \psi$, then $\psi \;\rotatebox[origin=c]{180}{$\vDash$}\vDash \varphi\vee(\neg \varphi\wedge \psi)$.}  It is easy to see that this is classically valid (indeed, we have $ \psi\vDash  \varphi\vee(\neg \varphi\wedge \psi) $ in classical logic, though the assumption of $\varphi\vDash\psi$ is needed in quantum logic). But, building on observations in \citealt{Dorr:2013}, we argue that it is invalid for epistemic modals.\footnote{Orthomodularity is a weakening of  \emph{modularity}, which says that if $\varphi\vDash\psi$ then $ (\varphi\vee\chi)\wedge \psi\Tautequiv \varphi\vee (\chi\wedge\psi)$. 
Hence related counterexamples  also show that modularity fails; for instance, as we argue presently, `It's raining' entails `It might be raining'; but `It's raining or it's not raining, and it might  be raining' is not plausibly equivalent to   `It's raining, or it's not raining and it might be raining'. See Remark \ref{remarkfitch} where related issues arise.}

We assume, again, that $p\vDash \lozenge p$. By contraposition of $\vDash$ and duality, this follows from the assumption that `must' is \textit{factive}, i.e., that $\Box p\vDash p$, which is widely (though not universally) accepted; see \citealt{MustStayStrong,Fintel:2016} for a variety of arguments for factivity, and see Footnote \ref{ptomight} for direct arguments that $p\vDash\lozenge p$.\footnote{\label{ptomight}Consider first proof by cases. The following seems valid:

\ex. Either John is home or he is at work. So he might be at home or he might be at work.

Likewise, Moorean constructions like \ref{mooremight} are infelicitous, which is naturally explained if $p\vDash \lozenge p$:

\ex. \label{mooremight} \#It's raining, but I don't know whether it might be raining.

Retraction data display a similar pattern:

\ex.The butler did it$\dots$.\a. \# I'm not sure whether the butler might have done it, but what I said earlier is true.
\b. \#I only said that the butler \emph{did it}, not that he \emph{might have} done it.

The inference in question also seems to preserve truth under upward-entailing attitude predicates:

\ex. John knows that Sue is at the party. So, John knows that Sue might be at the party.

Finally, the inference seems to be supported by Modus Ponens (which is generally accepted to be valid for conditionals with epistemic antecedents, even if, as \citet{McGee:1985} and followers maintain, it is not valid for non-Boolean consequents). Suppose Mary says to Mark:

\ex. If you might be sick, please let me know.\label{mary}

Now suppose Mark is in fact sick. In light of \ref{mary}, it seems that he is therefore obligated to let Mary know that he is sick.}  Given that $p\vDash \lozenge p$, orthomodularity says that  $\Diamond p\vDash p\vee (\neg p\wedge\Diamond p)$. But, if our logic is an \textsf{E}-logic,  the right disjunct of $p\vee (\neg p\wedge\Diamond p)$ is contradictory, so the disjunction is equivalent to $p$; likewise, $(\neg p\wedge\Diamond p)\vee p$ is predicted to be equivalent to $p$. For example, consider \ref{dorrpair}:

\ex. \label{dorrpair}\a. Either it isn't the butler but it might be, or else it's the butler.\label{dorr}
\b. It's the butler. \label{dorrb}

Intuitions here are somewhat unclear, because \ref{dorr} sounds  infelicitous---as we would expect a disjunction with a contradictory disjunct to sound. Nonetheless, as \citet{Dorr:2013} observe, \ref{dorr} feels intuitively equivalent to \ref{dorrb}. However, orthomodularity predicts that `It might be the butler' entails \ref{dorr}. Hence, given the contradictoriness of $\neg p\wedge\lozenge p$, orthomodularity would let us infer $p$ from $\lozenge p$, which again is unacceptable. 

Distributivity entails orthomodularity, so the intuitive failure of the latter gives us yet another reason to reject distributivity.

\subsection{Conservativity}\label{ConsSec}

Our last desideratum is a kind of methodological conservatism. Where there is not a specific argument for the failure of a classically valid inference pattern, we should validate it, on the presumption that it is indeed valid. 
It is of course harder to argue for the validity of a schema than against it; arguing that a principle is invalid only requires a single convincing counterexample, whereas we cannot ever look at every instance of a schema by way of arguing for its validity. Still, we think that as a methodological principle it is sensible to proceed by minimally altering classical logic in light of counterexamples and examining the result. Thus, we will look for a minimal variation on classical (modal) logic that can meet the desiderata above---that is, which is an \textsf{E}-logic and invalidates distributivity and even orthomodularity for the modal fragment, but still validates these, along with all other classical laws, for the non-modal fragment. Across the modal fragment, it will retain many of the central classical principles, including non-contradiction, excluded middle, the introduction and elimination rules for conjunction and disjunction, and De Morgan's laws, all of which have been invalidated by various theories of epistemic modals.  Moreover, our logic will be fully classical over parts of the modal fragment restricted to the same ``epistemic level''---essentially, fragments built up out of sentences with uniform levels of nesting of modal operators---limiting non-classicality to the part of the modal fragment that involves combinations of  sentences across different epistemic levels. 

To be clear, we do not have an argument that the logic we present is \emph{the} minimal variation on classical logic that is an \textsf{E}-logic and invalidates the patterns that  intuitively fail for epistemic modals. There are \textsf{E}-logics which validate strictly more classical patterns than ours; and there may turn out to be a case for adopting some such extension of our logic (as we will briefly discuss in \S~\ref{EpExtSection}).\footnote{In terms defined in Section \ref{algebraic}, there are orthologics stronger than the minimal orthologic in which orthomodularity is still not derivable (see \citealt{Harding1988}), and such logics can be extended to epistemic orthologics in the sense of Definition \ref{EODef}. Though we have arguments against some of these stronger logics, we have no sweeping argument against all stronger logics. In \S~\ref{EpExtSection}, we raise the possibility of strengthening our logic with modal principles.} But as we will discuss in \S~\ref{Comparisons}, our approach preserves a lot more of classical logic than any other \textsf{E}-logic we know of in the literature. Among other things, we hope that our discussion will spur exploration of other \textsf{E}-logics extending ours.

\section{Algebraic semantics\label{algebraic}}

In this section, we begin our development of formal semantics for epistemic modals. We start with a rather abstract \textit{algebraic semantics}. A model in such a semantics associates with each formula of the formal language an element in an \textit{algebra of propositions}. This is familiar from \textit{possible world semantics} for classical propositional logic, wherein a model associates with each formula a set of possible worlds, i.e., an element of the powerset algebra $\mathcal{P}(W)$ arising from the set $W$ of worlds. The meanings of `and', `or', and `not' are given by the operations of intersection, union, and complementation, respectively, in the powerset algebra $\mathcal{P}(W)$. The  only difference between possible world semantics and algebraic semantics for classical propositional logic is that in the latter, we allow Boolean algebras other than powerset algebras. Instead of associating each formula with an element of $\mathcal{P}(W)$ for some set $W$, we associate with each formula an element of an arbitrary Boolean algebra $B$, which comes equipped with operations $\neg$, $\wedge$, and $\vee$ used to interpret `not', `and', and `or' (see \S~\ref{AlgOrtho}). The powerset algebra $\mathcal{P}(W)$ is just one concrete example of a Boolean algebra. But for the purposes of classical propositional logic, there is no loss of generality in working only with powerset algebras.  

Possible world semantics for normal modal logic can also be viewed as a concrete version of an algebraic semantics. A set $W$ together with a binary accessibility relation $R$ on $W$ gives us not only the powerset algebra $\mathcal{P}(W)$ but also a modal operator $\Diamond$ on $\mathcal{P}(W)$ defined for $A\in\mathcal{P}(W)$ by 
\[\Diamond A=\{w\in W\mid \exists v: wRv\mbox{ and }v\in A\}.\]
A more abstract algebraic semantics uses an arbitrary Boolean algebra equipped with a modal operator, called a \textit{Boolean algebra with operator} (BAO). In the context of modal logic, there \textit{is} a loss of generality in working only with powerset algebras, as not all normal modal logics can be given possible world semantics as above (see \citealt{Litak2019} and references therein). But the normal modal logics discussed in connection with natural language semantics typically can be handled by possible world semantics.

The logic we are after in this paper is \textit{non-classical}. Thus, we cannot use powerset algebras, as in possible world semantics, or  Boolean algebras more generally, as in algebraic semantics for classical logic. Instead, we will use a more general class of algebras, known as \textit{ortholattices}, in which it is possible to invalidate classical laws such as distributivity. We will add modal operators to ortholattices to give an algebraic semantics for an \textit{epistemic orthologic}. Then in \S~\ref{PossSem}, we will give a more concrete semantics, known as \textit{possibility semantics}, that stands to ortholattice semantics as possible world semantics stands to Boolean algebraic semantics. 

\subsection{Review of algebraic semantics for orthologic}\label{AlgOrtho}

In this section, we review ortholattices and their associated \textit{orthologic}. Ortholattices have long been important objects of study in lattice theory (see \citealt[\S~14]{Birkhoff1967}), examples of which arise as algebras of events in quantum mechanics (\citealt{Birkhoff1936}). We will argue that ortholattices also arise as algebras of propositions in a language with epistemic modals. Those familiar with basic lattice theory and ortholattices may skip straight to \S~\ref{AddEpistemics1}, where we add epistemic modals to the picture.

\subsubsection{Ortholattices}

First, we recall the definition of a lattice, where we think of $A$ as a set of propositions and $\vee$ and $\wedge$ as the operations of disjunction and conjunction, respectively. 

\begin{definition}\label{BoundLat} A \textit{lattice} is a tuple $L=\langle A,\vee,\wedge\rangle$ where $A$ is a nonempty set and $\vee$ and $\wedge$ are binary operations on $A$ such that the following equations hold for all $a,b,c\in A$ and $\circ\in \{\vee,\wedge\}$:
\begin{center}
\begin{minipage}{2.75in}
\begin{itemize}
\item idempotence: $a\circ a= a$;
\item commutativity: $a\circ b = b\circ a$;
\end{itemize}
\end{minipage}\begin{minipage}{2.75in}
\begin{itemize}
\item associativity: $a\circ (b\circ c)= (a\circ b)\circ c$;
\item absorption: $a\wedge (a\vee b)=a\vee (a\wedge b)=a$.
\end{itemize}
\end{minipage}
\end{center}
We define a binary relation $\leq$ on $A$, called the \textit{lattice order of $L$}, by: $a\leq b$ if and only if $a=a \wedge b$. \end{definition}

It is easy to check that $a\leq b$ is also equivalent to $a\vee b = b$ (using absorption and commutativity). Moreover, $\leq$ is a partial order (i.e., reflexive, transitive, and antisymmetric). Indeed, let us recall the order-theoretic definition of a lattice as a partially ordered set. Given a partial order  $\leqslant$  on a set $A$,
\begin{itemize}
\item an \textit{upper bound} of a subset $X\subseteq A$ is a $y\in A$ such that $x\leqslant y$ for every $x\in X$;
\item a \textit{least upper bound} of $X$ is an upper bound $y$ of $X$ such that $y\leqslant z$ for every upper bound $z$ of $X$.
\end{itemize} 
If there is a least upper bound of $X$, it is unique by the antisymmetry of $\leqslant$. The notion of a \textit{greatest lower bound} is defined dually. Then a partially ordered set is a lattice if every nonempty finite subset has both a least upper bound and a greatest lower bound. From such a partially ordered set, we obtain a lattice $\langle A,\vee,\wedge\rangle$ in the sense of Definition \ref{BoundLat} by defining $a\vee b$ to be the least upper bound of $\{a,b\}$ and $a\wedge b$ to be the greatest lower bound of $\{a,b\}$. Conversely, given a lattice in the sense of Definition \ref{BoundLat}, the partially ordered set $\langle A,\leq\rangle$ with $\leq$ defined as in Definition \ref{BoundLat} is a lattice in the order-theoretic sense. Hence we may think of lattices in terms of either the equational or order-theoretic definition. Finally, a lattice is \textit{complete} if every subset has both a least upper bound and a greatest lower bound. Every finite lattice is complete but there are infinite lattices that are not complete.

We will assume that our lattices have bounds, corresponding to the contradictory proposition $\bot$ and the trivial proposition $\top$.

\begin{definition}A \textit{bounded lattice} is a tuple $L=\langle A,\vee,0,\wedge,1\rangle$ where $\langle A,\vee,\wedge\rangle$ is a lattice and $0$ and $1$ are elements of $A$ such that for all $a\in A$, we have
\begin{itemize}
 \item boundedness: $a\vee 0 =a$ and $a\wedge 1 =a$.
\end{itemize}
\end{definition}
\noindent Order-theoretically, a lattice is bounded if its order has a minimum element and a maximum element.

Adding an operation $\neg$ for negation finally brings us to the definition of an ortholatticce.

\begin{definition}\label{OrtholatticeDef} An \textit{ortholattice} is a tuple $\langle A,\vee,0,\wedge,1,\neg\rangle $ where $\langle A,\vee,0,\wedge,1\rangle$ is a bounded lattice and $\neg$ is a unary operation on $A$, called an  \textit{orthocomplementation}, that satisfies:
\begin{enumerate}
\item\label{OrthoDef1} complementation: for all $a\in A$, $a\vee\neg a=1$ and $a\wedge\neg a=0$;
\item\label{OrthoDef2} involution: for all $a\in A$, $\neg\neg a=a$;
\item\label{OrthoDef3} order-reversal: for all $a,b\in A$, if $a\leq b$, then $\neg b\leq \neg a$.
\end{enumerate}
An equivalent definition replaces \ref{OrthoDef3} with (either one of) \textit{De Morgan's laws}:
\begin{itemize}
\item for all $a,b\in A$, $\neg (a\vee b)=\neg a\wedge\neg b$;
\item for all $a,b\in A$, $\neg (a\wedge b) = \neg a\vee\neg b$.
\end{itemize}

\end{definition}

The difference between arbitrary ortholattices and the Boolean algebras of classical logic is that ortholattices need not obey the distributive law, as we will see in Examples \ref{OrthoExample} and \ref{KeyEx}.

\begin{definition}
A \textit{Boolean algebra} is a tuple $\langle A,\vee,0,\wedge,1,\neg\rangle $ where $\langle A,\vee,0,\wedge,1\rangle$ is a bounded lattice, $\neg$ is a unary operation on $A$ satisfying complementation as in Definition \ref{OrtholatticeDef}.\ref{OrthoDef1}, and the \textit{distributive laws} hold:
\begin{itemize}
\item for all $a,b,c\in A$, $a\wedge (b\vee c)= (a\wedge b)\vee (a\wedge c)$;
\item  for all $a,b,c\in A$, $a\vee (b\wedge c)= (a\vee b)\wedge (a\vee c)$.\footnote{In fact, a lattice satisfies the first bullet point if and only if it satisfies the second, so including both is redundant.}
\end{itemize}
\end{definition}
\noindent It is straightforward to prove that every Boolean algebra is an ortholattice.

The following weakening of distributivity is important in the study of ortholattices arising in quantum mechanics (see, e.g., \citealt{Chiara2002}).

\begin{definition}An \textit{orthomodular lattice} is an ortholattice satisfying the orthomodular law:
\begin{itemize}
\item  for all $a,b\in A$, $a\vee (\neg a\wedge (a\vee b)) = a\vee b$.
\end{itemize}
Or equivalently, for all $a,c\in A$, if $a\leq c$, then  $c\leq a\vee (\neg a\wedge c)$ (equivalently, $c=a\vee (\neg a\wedge c)$).
\end{definition}

Associated with the difference between Boolean algebras and arbitrary ortholattices concerning distributivity is a difference concerning the interaction of conjunction, contradiction, and negation, given by the following standard fact.

\begin{lemma}\label{PseudoLem} In a Boolean algebra, $\neg$ is \textit{pseudocomplementation}: for all $a,b\in A$, $a\wedge b=0$ implies $b\leq\neg a$, so $\neg a$ is the greatest element with respect to $\leq$ of the set $\{b\in A\mid a\wedge b=0\}$.\footnote{This fact also holds for Heyting algebras, which provide algebraic semantics for \textit{intuitionistic logic}, where $\neg a := a\to 0$.}
\end{lemma}

Crucially, the orthocomplementation in an ortholattice is not necessarily  pseudocomplementation. In fact, the orthocomplementation being pseudocomplementation implies that the ortholattice is Boolean.

\begin{proposition}\label{PseudoToBoole} For any ortholattice $L$, the following are equivalent:
\begin{enumerate}
\item\label{PseudoToBoole1}  $L$ is a Boolean algebra;
\item\label{PseudoToBoole1.5}  $L$ is distributive;
\item\label{PseudoToBoole2}  the orthocomplementation operation in $L$ is pseudocomplementation.
\end{enumerate}
\end{proposition}
\begin{proof} The equivalence of \ref{PseudoToBoole1} and \ref{PseudoToBoole1.5} is straightforward, as noted above.

From \ref{PseudoToBoole1.5} to \ref{PseudoToBoole2}, we show that distributivity implies pseudocomplementation over ortholattices. Suppose $a\wedge b =0$, so $a\wedge b\leq \neg b$. Then since $a\wedge \neg b\leq\neg b$, we have $(a\wedge b)\vee (a\wedge\neg b)\leq \neg b$. Finally, by distributivity, $a\wedge (b\vee\neg b)\leq (a\wedge b)\vee (a\wedge\neg b)$, so $a\leq\neg b$.

From \ref{PseudoToBoole2} to \ref{PseudoToBoole1.5}, we show that pseudocomplementation implies distributivity over ortholattices. First note that pseudocomplementation implies disjunctive syllogism: $(x\vee y)\wedge \neg x\leq y$. For ${(x\vee y)\wedge \neg x\wedge \neg y\leq 0}$ by De Morgan's laws 
and complementation, so $(x\vee y)\wedge\neg x\leq y$ by pseudocomplementation and involution. Now for distributivity, we have
\begin{eqnarray*}
&& a \wedge( b\vee c)\wedge (\neg a \vee \neg b)\leq a\wedge(b\vee c)\wedge \neg b\leq a\wedge c \quad\mbox{using disjunctive syllogism twice}\\
&\Rightarrow & a\wedge(b\vee c)\wedge(\neg a\vee\neg b)\wedge\neg (a\wedge c)\leq  \neg(a\wedge c)\wedge (a\wedge c) \leq 0 \quad\mbox{} \\
&\Rightarrow& a\wedge (b\vee c) \leq \neg ((\neg a\vee\neg b)\wedge\neg (a\wedge c))\quad\mbox{by pseudocomplementation} \\
&\Rightarrow& a\wedge (b\vee c) \leq (a\wedge b)\vee (a \wedge c) \quad\mbox{by De Morgan's and involution}.\qedhere
\end{eqnarray*} 
\end{proof}

\noindent Thus, we obtain a three-way equivalence for any ortholattice $L$ between being a Boolean algebra, being distributive, and having its orthocomplementation be pseudocomplementation.

\begin{example}\label{OrthoExample} Figure \ref{OrthoEx} shows Hasse diagrams of the ortholattices $\mathbf{O}_6$ and $\mathbf{MO}_2$. Recall that in a Hasse diagram of a lattice with lattice order $\leq$, a line segment going \textit{upward} from $x$ to $y$ means that $x\leq y$ and there is no third element $z$ with $x\leq z\leq y$. Observe that $\mathbf{O}_6$ is not orthomodular and hence not distributive. For $a\leq b$ and yet $a\vee (\neg a\wedge b)= a\vee 0=a\neq b$. Also note that the orthocomplementation  is not pseudocomplementation: $\neg a\wedge b = 0$ and yet $ b\not\leq \neg\neg a=a$. $\mathbf{MO}_2$ is orthomodular\footnote{In fact, it is a \textit{modular} lattice, meaning, again, that for all $a,b,c\in A$,  if $a\leq c$, then $a\vee (b\wedge c)=(a\vee b)\wedge c$.} but it is not distributive: $a\wedge (\neg a\vee b)=a\wedge 1= a\neq  0 = 0\vee 0 = (a\wedge \neg a)\vee (a\wedge b)$. Also note that orthocomplementation is not pseudocomplementation: $a\wedge b=0$ and yet $b\not\leq \neg a$.

\begin{figure}[h]
\begin{center}
\begin{tikzpicture}[->,>=stealth',shorten >=1pt,shorten <=1pt, auto,node
distance=2cm,semithick,every loop/.style={<-,shorten <=1pt}]
\tikzstyle{every state}=[fill=gray!20,draw=none,text=black]
\node[circle,draw=black!100,fill=black!100, label=below:$\textcolor{red}{0}$,inner sep=0pt,minimum size=.175cm] (0) at (0,0) {{}};
\node[circle,draw=black!100,fill=black!100, label=left:$\textcolor{red}{a}$,inner sep=0pt,minimum size=.175cm] (a) at (-.75,.75) {{}};
\node[circle,draw=black!100,fill=black!100, label=right:$\textcolor{red}{\neg b}$,inner sep=0pt,minimum size=.175cm] (b) at (.75,.75) {{}};
\node[circle,draw=black!100,fill=black!100, label=left:$\textcolor{red}{b}$,inner sep=0pt,minimum size=.175cm] (1l) at (-.75,1.5) {{}};
\node[circle,draw=black!100,fill=black!100, label=right:$\textcolor{red}{\neg a}$,inner sep=0pt,minimum size=.175cm] (1r) at (.75,1.5) {{}};
\node[circle,draw=black!100,fill=black!100, label=above:$\textcolor{red}{1}$,inner sep=0pt,minimum size=.175cm] (new1) at (0,2.25) {{}};

\path (new1) edge[-] node {{}} (1l);
\path (new1) edge[-] node {{}} (1r);
\path (1l) edge[-] node {{}} (a);
\path (1r) edge[-] node {{}} (b);
\path (a) edge[-] node {{}} (0);
\path (b) edge[-] node {{}} (0);

\node[circle,draw=black!100,fill=black!100, label=above:$\textcolor{red}{1}$,inner sep=0pt,minimum size=.175cm] (1) at (6,2) {{}};
\node[circle,draw=black!100,fill=black!100, label=left:$\textcolor{red}{a}$,inner sep=0pt,minimum size=.175cm] (x) at (4,1) {{}};
\node[circle,draw=black!100,fill=black!100, label=right:$\textcolor{red}{\neg a}$,inner sep=0pt,minimum size=.175cm] (y) at (5,1) {{}};
\node[circle,draw=black!100,fill=black!100, label=left:$\textcolor{red}{b}$,inner sep=0pt,minimum size=.175cm] (y') at (7,1) {{}};
\node[circle,draw=black!100,fill=black!100, label=right:$\textcolor{red}{\neg b}$,inner sep=0pt,minimum size=.175cm] (z) at (8,1) {{}};
\node[circle,draw=black!100,fill=black!100, label=below:$\textcolor{red}{0}$,inner sep=0pt,minimum size=.175cm] (0) at (6,0) {{}};
\path (1) edge[-] node {{}} (y);
\path (1) edge[-] node {{}} (y');
\path (1) edge[-] node {{}} (x);
\path (1) edge[-] node {{}} (z);
\path (x) edge[-] node {{}} (0);
\path (y) edge[-] node {{}} (0);
\path (y') edge[-] node {{}} (0);
\path (z) edge[-] node {{}} (0);

\end{tikzpicture}
\end{center}
\caption{Hasse diagrams of the ortholattices $\mathbf{O}_6$ (left) and $\mathbf{MO}_2$ (right).}\label{OrthoEx}
\end{figure}
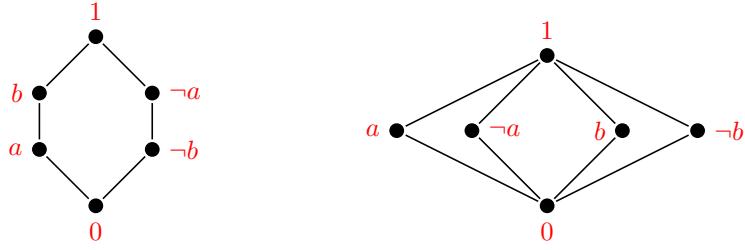
\end{example}

\subsubsection{Language and consequence}\label{LangCons1}

We can use ortholattices to interpret a basic propositional logical language.

\begin{definition}\label{LDef} Let $\mathcal{L}$ be the language generated by the grammar
\[\varphi::= \top\mid p\mid \neg\varphi\mid (\varphi\wedge\varphi)\]
where $p$ belongs to a countably infinite set $\mathsf{Prop}$ of propositional variables. 
\end{definition}

We define  $\varphi\vee\psi:=\neg(\neg\varphi\wedge\neg\psi)$, a definition justified by the fact that ortholattices satisfy De Morgan's laws and involution. Note that we use the same symbols for the connectives of $\mathcal{L}$ and the operations in ortholattices, trusting that no confusion will arise.

As usual in algebraic semantics, we interpret propositional variables as elements of an algebra and then extend the interpretation to all formulas of the language recursively.

\begin{definition} A \textit{valuation} on an ortholattice $\langle A,\vee,0,\wedge,1,\neg\rangle$ is a map $\theta: \mathsf{Prop}\to A$. Such a $\theta$ extends to $\tilde{\theta}:\mathcal{L}\to A$  by: $\tilde{\theta}(\top)=1$, $\tilde{\theta}(\neg\varphi)=\neg\tilde{\theta}(\varphi)$, and $\tilde{\theta}(\varphi\wedge\psi)=\tilde{\theta}(\varphi)\wedge \tilde{\theta}(\psi)$.
\end{definition}

Also as usual, we say that $\psi$ is a semantic consequence of $\varphi$ if the semantic value of $\varphi$ is always below the semantic value of $\psi$ in the lattice order $\leq$ (which one may now view as an entailment relation, just as the subset relation is viewed as an entailment relation between propositions in possible world semantics).

\begin{definition}\label{AlgCon0} Given a class $\mathbf{C}$ of ortholattices, we define the semantic consequence relation $\vDash_\mathbf{C}$, a binary relation on $\mathcal{L}$, as follows: $\varphi\vDash_\mathbf{C}\psi$ if for all $L\in\mathbf{C}$ and valuations $\theta$ on $L$, we have $\tilde{\theta}(\varphi)\leq \tilde{\theta}(\psi)$, where $\leq$ is the lattice order of $L$.
\end{definition}
We can similarly define a consequence relation between a set of premises on the left and a single conclusion on the right: $\Gamma\vDash_\mathbf{C}\psi$ if for all $L\in\mathbf{C}$, valuations $\theta$ on $L$, and $a\in L$, if $a$ is a lower bound of $\{\tilde{\theta}(\varphi) \mid \varphi\in\Gamma\}$, then $a\leq \tilde{\theta}(\psi)$. But for simplicity we will only consider finite sets of premises here, in which case a single premise on the left suffices given that $\mathcal{L}$ contains a conjunction interpreted as meet.

\subsubsection{Logic, soundness and completeness}\label{Log1}

Axiomatizing the semantic consequence relation of Definition \ref{AlgCon0} is straightforward.

\begin{definition}[\citealt{Goldblatt1974}]\label{OrthoLogicDef} An \textit{orthologic} is a binary relation $\vdash$ on the set $\mathcal{L}$ of formulas such that for all $\varphi,\psi,\chi\in \mathcal{L}$:
\begin{center}
\begin{tabular}{ll}
1. $\varphi\vdash\top$; & 6.  $\neg\neg \varphi\vdash\varphi$; \\
2. $\varphi\vdash\varphi$; & 7. $\varphi\wedge\neg\varphi\vdash\psi$;\\
3. $\varphi\wedge\psi\vdash\varphi$; \qquad\qquad &8. if $\varphi\vdash\psi$ and $\psi\vdash\chi$, then $\varphi\vdash\chi$;\\
4. $\varphi\wedge\psi\vdash\psi$; &9. if $\varphi\vdash\psi$ and $\varphi\vdash\chi$, then $\varphi\vdash \psi\wedge\chi$; \\
5.  $\varphi\vdash\neg\neg\varphi$; & 10. if $\varphi\vdash\psi$, then $\neg\psi\vdash\neg\varphi$.\\
\end{tabular}
\end{center}
A \textit{theorem} of the orthologic $\vdash$ is a formula $\varphi$ such that $\top\vdash\varphi$. 
As the intersection of orthologics is clearly an orthologic, there is a smallest orthologic, denoted \textsf{O} or $\vdash_\mathsf{O}$. 
\end{definition}

Note that with $\varphi\vee\psi$ defined as $\neg(\neg\varphi\wedge\neg\psi)$, we get the introduction and elimination rules for disjunction:
\begin{itemize}
\item $\varphi\vdash \varphi\vee\psi$.
\item if $\varphi\vdash\chi$ and $\psi\vdash\chi$, then $\varphi\vee\psi\vdash\chi$.
\end{itemize}
 For the first, by rule 3, we have $\neg\varphi\wedge\neg\psi\vdash \neg\varphi$, so by rule 10, we have $\neg\neg\varphi\vdash \varphi\vee\psi$, which with rules 5 and 8 yields $\varphi\vdash \varphi\vee\psi$. For the second, if $\varphi\vdash\chi$ and $\psi\vdash\chi$, then $\neg\chi\vdash \neg\varphi$ and $\neg\chi\vdash\neg\psi$ by rule 10, which implies $\neg\chi\vdash \neg\varphi\wedge\neg\psi$ by rule 9, which implies $\varphi\vee\psi \vdash \neg\neg\chi$ by rule 10 and hence $\varphi\vee\psi \vdash \chi$ by rules 6 and 8.
 
 However, we do \textit{not} have 
 pseudocomplementation or distributivity, either of which would collapse \textsf{O} to classical logic (recall Proposition \ref{PseudoToBoole}). Nor do we have
 what might be called \textit{proof by cases with side assumptions}:
 \begin{itemize}
 \item if $\xi\wedge \varphi\vdash \chi$ and $\xi\wedge\psi\vdash\chi$, then $\xi\wedge (\varphi\vee\psi)\vdash \chi$.
 \end{itemize}
This would allow the derivation of distributivity, since disjunction introduction and proof by cases with side assumptions yield

\[\xi\wedge \varphi \vdash (\xi\wedge \varphi)\vee (\xi\wedge \psi)\mbox{ and } \xi\wedge \psi\vdash (\xi\wedge \varphi)\vee (\xi\wedge \psi),\mbox{ so } \xi\wedge (\varphi\vee \psi) \vdash (\xi\wedge \varphi)\vee (\xi\wedge \psi).\] 

A completeness theorem can now be proved using standard techniques of algebraic logic.

\begin{theorem}\label{AlgComp1} The logic $\mathsf{O}$ is sound and complete with respect to the class $\mathbf{O}$ of all ortholattices according to the consequence relation of Definition \ref{AlgCon0}: for all $\varphi,\psi\in\mathcal{L}$, we have $\varphi\vdash_\mathsf{O}\psi$ if and only if $\varphi\vDash_\mathbf{O}\psi$.
\end{theorem}
\begin{proof} For soundness, it is easy to check that $\varphi\vdash_\mathsf{O}\psi$ implies $\varphi\vDash_\mathbf{O}\psi$. For completeness, recall the construction of the \textit{Lindenbaum-Tarski algebra $L$ of $\mathsf{O}$}: the underlying set of $L$ is the set of all equivalence classes of formulas of $\mathcal{L}$, where $\varphi$ and $\psi$ are equivalent if $\varphi\vdash\psi$ and $\psi\vdash\varphi$; then let $0=[\neg \top]$ and $1=[\top]$, and given equivalence classes $[\varphi]$ and $[\psi]$,  define $[\varphi]\vee[\psi]=[\varphi\vee\psi]$,  $[\varphi]\wedge[\psi]=[\varphi\wedge\psi]$, $\neg[\varphi]=[\neg\varphi]$ (that the choice of representatives does not matter is easily shown using the principles of \textsf{O}). Then where $\leq$ is the lattice order of $L$, we have $[\varphi]\leq[\psi]$ iff $\varphi\vdash_\mathsf{O}\psi$. It is easy to check that $L$ is an ortholattice. Moreover, where $\theta$ is the valuation with $\theta(p)=[p]$ for each $p\in\mathsf{Prop}$, an obvious induction shows $\tilde{\theta}(\varphi)=[\varphi]$ for each $\varphi\in\mathcal{L}$. Now if $\varphi\not\vdash_\mathsf{O}\psi$, then  $[\varphi]\not\leq [\psi]$ and hence $\tilde{\theta}(\varphi)\not\leq \tilde{\theta}(\psi)$, so $\varphi\nvDash_\mathbf{O}\psi$.\end{proof}

\begin{remark}\label{remarkfitch} A Fitch-style natural deduction system for the minimal orthologic in the signature $\top,\neg,\wedge,\vee$ can be obtained from a Fitch-style natural deduction system for classical logic in the same signature (cf.~\citealt{Fitch1952,Fitch1966}) by (i) restricting the rule of \textit{Reiteration} so that the formula occurrence to be reiterated and its reiterates must belong to the same column and the same subproofs,\footnote{In fact, Reiteration can be dropped entirely provided we understand a (sub)proof with a single formula $\varphi$ as a (sub)proof from the assumption $\varphi$ to conclusion $\varphi$. See \citealt{Holliday2023} for a rigorous formulation.}   (ii) keeping the introduction and elimination rules for $\wedge$ and $\vee$ the same,\footnote{We have in mind formulations of the introduction and elimination rules for $\wedge$ and $\vee$, as in \citealt{Fitch1952,Fitch1966} (but unlike some introductory logic texts), that do not themselves build in the power to draw previous formulas into subproofs; the Reiteration rule must always be used for this purpose.} and (iii) allowing $\neg$ Introduction  and Reductio ad Absurdum to apply when there is a contradiction between a formula derived in a subproof and a formula previously derived in the column in which the subproof immediately occurs. Such a proof system is shown in Figure \ref{FitchRules}. The motivation for restricting Reiteration can easily be seen when we consider epistemic modal propositions,  to be added in the next section, as in the following example: 
\[\begin{nd}
\hypo [1] {1} {\Diamond p} 
\hypo [2] {2} {(p\vee\neg p)} 
\open
\hypo [3] {3} {p}
\have [4] {4} {p\vee (\neg p\wedge\Diamond p)} \oi{3}
\close
\open
\hypo [5] {5} {\neg p}
\have [6] {6} {\Diamond p} \r{1}
\have [7] {7} {\neg p \wedge \Diamond p} \ai{5,6}
\have [8] {8} {p\vee (\neg p \wedge \Diamond p)} \oi{7}
\close
\have [9] {9} {p\vee (\neg p\wedge\Diamond p)} \oe{2,3-4, 5-8}
\end{nd}\]
\noindent The application of Reiteration on line 6 is obviously suspect: we should not be allowed to reiterate the assumption that \textit{might $p$} into a subproof where we have just supposed \textit{not} $p$! Moreover, if $q$ is a genuine propositional variable, standing in for any proposition (cf.~\S~\ref{DistinguishSec}), including an epistemic proposition such as \textit{might $p$}, then we cannot accept the analogous proof with $q$ in place of $\Diamond p$. However, if we restrict the Reiteration rule as suggested above, then all is well: where $\varphi\vdash_{\mathsf{FitchO}}\psi$ means that there is a Fitch-style proof of the conclusion $\psi$ from the assumption $\varphi$, Theorem \ref{AlgComp1} holds with $\vdash_{\mathsf{FitchO}}$ in place of $\vdash_\mathsf{O}$. For soundness, an easy induction shows that if there is a subproof or proof beginning with $\varphi$ and ending with $\psi$, then---thanks to the restriction on Reiteration---$\psi$ is in fact a semantic consequence of $\varphi$, i.e., $\varphi\vDash_\mathbf{O}\psi$. For completeness, it is easy to check that the principles of Definition \ref{OrthoLogicDef}, as well as the equivalence of $\varphi\vee\psi$ and $\neg(\neg\varphi\wedge\neg\psi)$, hold for $\vdash_{\mathsf{FitchO}}$, in which case the same style of proof as for Theorem \ref{AlgComp1} applies.\end{remark}

\begin{figure}
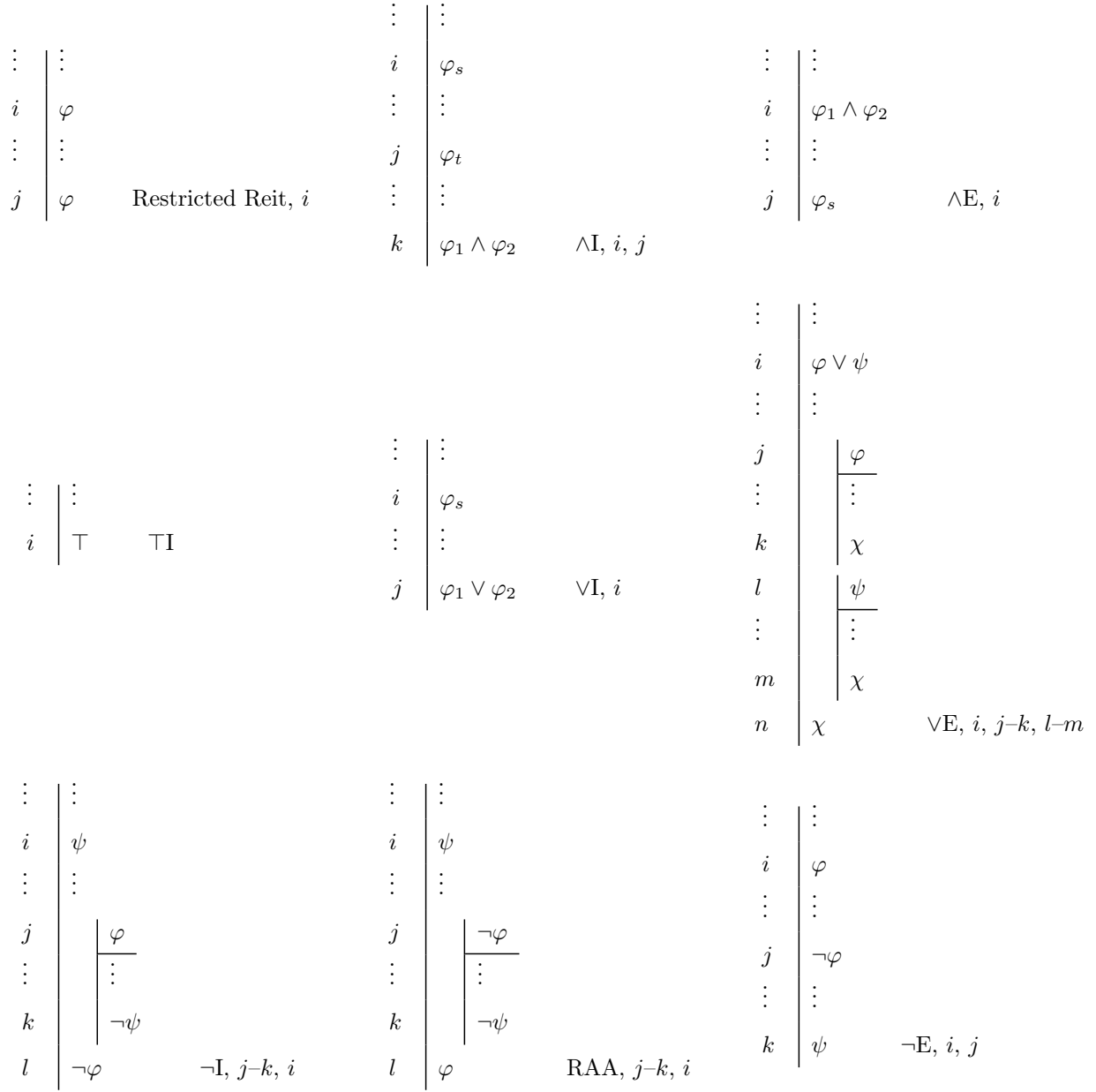

\begin{center}
\begin{minipage}{2.25in}
\[\begin{nd}
\have [\vdots] {1} {\vdots} 
\have [i] {2} {\varphi}
\have [\vdots] {3} {\vdots}
\have [j] {5} {\varphi}  \lr{2}
\end{nd}\]\end{minipage}\begin{minipage}{2.25in}
\[\begin{nd}
\have [\vdots] {4} {\vdots}
\have [i] {5}  {\varphi_s}
\have [\vdots] {8}   {\vdots}
\have [j] {9} {\varphi_t}
\have [\vdots] {10}   {\vdots}
\have [k] {11}   {\varphi_1\wedge\varphi_2}\ai{5,9}
\end{nd}\quad\;\;\;\]
\end{minipage}\begin{minipage}{2.25in}
\[\begin{nd}
\have [\vdots] {4} {\vdots}
\have [i] {5}  {\varphi_1\wedge\varphi_2}
\have [\vdots] {8}   {\vdots}
\have [j] {9} {\varphi_s} \ae{5}
\end{nd}\qquad\quad\,\]
\end{minipage}

\begin{minipage}{2.25in}
\[\begin{nd}
\have [\vdots] {1} {\vdots} 
\have [i] {2} {\top} \topi{}
\end{nd}\qquad\qquad\quad\,\]
\end{minipage}\begin{minipage}{2.25in}
\[\begin{nd}
\have [\vdots] {1} {\vdots} 
\have [i] {2} {\varphi_s}
\have [\vdots] {3} {\vdots} 
\have [j] {4} {\varphi_1\vee\varphi_2}\oi{2}
\end{nd}\qquad\;\;\;\]
\end{minipage}\begin{minipage}{2.25in}
\[\begin{nd}
\have [\vdots] {} {\vdots} 
\have [i] {0} {\varphi\vee\psi}
\have [\vdots] {1} {\vdots} 
\open
\hypo [j] {2} {\varphi}
\have [\vdots] {3} {\vdots}
\have [k] {5} {\chi}
\close
\open
\hypo [l] {7}  {\psi}
\have [\vdots] {8} {\vdots}
\have [m] {9} {\chi}
\close
\have [n] {n} {\chi} \oe{0,2-5,7-9}
\end{nd}\]
\end{minipage}

\begin{minipage}{2.25in}
\[\begin{nd}
\have [\vdots] {0} {\vdots}
\have [i] {3}   {\psi}
\have [\vdots] {}  {\vdots}
\open
\hypo [j] {1} {\varphi}
\have [\vdots] {4}   {\vdots}
\have [k] {5}   {\neg\psi}
\close
\have [l] {6} {\neg\varphi} \ni{1-5,3}
\end{nd}\,\]
\end{minipage}\begin{minipage}{2.25in}
\[\begin{nd}
\have [\vdots] {0} {\vdots}
\have [i] {3}   {\psi}
\have [\vdots] {4}   {\vdots}
\open
\hypo[j] {1} {\neg \varphi}
\have [\vdots] {}  {\vdots}
\have [k] {5}   {\neg\psi}
\close
\have [l]{6} {\varphi} \RAA{1-5,3}
\end{nd}
\]
\end{minipage}\begin{minipage}{2.25in}
\[\begin{nd}
\have [\vdots] {4} {\vdots}
\have [i] {5}  {\varphi}
\have [\vdots] {6} {\vdots}
\have [j] {7}  {\neg\varphi}
\have [\vdots] {8} {\vdots}
\have [k] {9}{\psi} \ne{5, 7}
\end{nd}\qquad\qquad\,\]
\end{minipage}
\end{center}
\caption{Rules of a Fitch-style proof system for the minimal orthologic, where $s,t\in \{1,2\}$.}\label{FitchRules}
\end{figure}

\subsection{Adding epistemic modality}\label{AddEpistemics1}

In this section, we extend the algebraic semantics of \S~\ref{AlgOrtho} to interpret the modals `must' and `might'.

\subsubsection{Modal and epistemic ortholattices}

We begin by extending ortholattices with a unary operation $\Box$. First, we define a baseline notion of a \textit{modal} ortholattice and then add additional conditions for epistemic ortholattices.

\begin{definition} A \textit{modal ortholattice} is a tuple $\langle A,\vee,0,\wedge,1,\neg,\Box\rangle$ where $\langle A,\vee,0,\wedge,1,\neg\rangle $ is an ortholattice and $\Box$ is a unary operation on $A$ satisfying:
\begin{itemize}
\item $\Box (a\wedge b)=\Box a\wedge\Box b$ for all $a,b\in A$, and $\Box 1=1$.
\end{itemize}
For $a\in A$, we define $\Diamond a=\neg\Box\neg a$.  \end{definition}

The following is easy to check.
\begin{lemma}\label{DiamondLem} In any modal ortholattice, we have:
\begin{itemize}
\item $\Diamond (a\vee b)=\Diamond a\vee\Diamond b$ for all $a,b\in A$, and $\Diamond 0=0$.
\end{itemize}
\end{lemma}

\noindent Indeed, we could have defined modal ortholattices as algebras $\langle A,\vee,0,\wedge,1,\neg,\Diamond\rangle$ where $\langle A,\vee,0,\wedge,1,\neg\rangle $ is an ortholattice and $\Diamond$ is a unary operation on $A$ satisfying the conditions of Lemma \ref{DiamondLem}. But later it will turn out to be more convenient to have $\Box$ as our primitive.

Now we can consider additional constraints on the $\Box$ operation, just as modal logic considers additional axioms on a $\Box$ modality (see, e.g., \citealt[p.~116]{Chagrov1997}). 
In particular, in order to view $\Box$ and $\Diamond$  as `must' and `might', we adopt two further constraints. First, corresponding to the factivity of `must', we have the following constraint.

\begin{definition} A \textit{$\mathsf{T}$ modal ortholattice} is a modal ortholattice also satisfying
 \begin{itemize}
 \item $\Box a\leq a$ for all $a\in A$.
 \end{itemize}
 \end{definition}
 Next comes the crucial constraint corresponding to Wittgenstein sentences being contradictions.
 \begin{definition}
 
An \textit{epistemic ortholattice} is a $\mathsf{T}$ modal ortholattice also satisfying 
 \begin{itemize}
 \item \textsf{Wittgenstein's Law}: $\neg a\wedge\Diamond a =0$ for all $a\in A$.
 \end{itemize}
\end{definition}
\noindent By the involution property of $\neg$ in an ortholattice, \textsf{Wittgenstein's Law} is equivalent to $a\wedge\Diamond \neg a =0$; and by the commutativity of $\wedge$ in a lattice, it is also equivalent to $\Diamond a\wedge \neg a =0$ (and $\Diamond \neg a\wedge  a =0$), in contrast to the consistency of $\Diamond a\wedge\neg a$ (and $\Diamond \neg a\wedge a$) in some dynamic systems (e.g.,~in \citealt{GSV:1996}; see \citealt{Benthem:1996}).

Some philosophers of language have argued for \textit{iteration principles} for `must' and `might', leading to the following additional constraints.

\begin{definition} An \textit{$\mathsf{S5}$ modal ortholattice} is a $\mathsf{T}$ modal ortholattice also satisfying:
\begin{itemize}
\item $\Box a\leq \Box\Box a$ for all $a\in A$;
\item $\Diamond a\leq \Box\Diamond a$ for all $a\in A$.
\end{itemize}
A \textit{$\mathsf{S5}$ epistemic ortholattice} is an $\mathsf{S5}$ modal ortholattice that is also an epistemic ortholattice.\end{definition}

\begin{figure}[h]
\begin{center}
\tikzset{every loop/.style={min distance=10mm,looseness=10}}
\begin{tikzpicture}[->,>=stealth',shorten >=1pt,shorten <=1pt, auto,node
distance=2.5cm,semithick]
\tikzstyle{every state}=[fill=gray!20,draw=none,text=black]
\node[circle,draw=black!100, fill=black!100, label=below:$\textcolor{red}{0}$,inner sep=0pt,minimum size=.175cm] (0) at (0,0) {{}};
\node[circle,draw=black!100, fill=black!100, label=right:$\textcolor{red}{\neg d}$,inner sep=0pt,minimum size=.175cm] (d) at (0,2) {{}};

\node[circle,draw=black!100, fill=black!100, label=below:$\textcolor{red}{a}$,inner sep=0pt,minimum size=.175cm] (a) at (-2.25,2) {{}};
\node[circle,draw=black!100, fill=black!100, label=below:$\textcolor{red}{\neg c}$,inner sep=0pt,minimum size=.175cm] (Nc) at (2.25,2) {{}};

\node[circle,draw=black!100, fill=black!100, label=left:$\textcolor{red}{b}$,inner sep=0pt,minimum size=.175cm] (b) at (-2.25,4) {{}};
\node[circle,draw=black!100, fill=black!100, label=right:$\textcolor{red}{\neg b}$,inner sep=0pt,minimum size=.175cm] (Nb) at (2.25,4) {{}};
\node[circle,draw=black!100, fill=black!100,label=above:$\textcolor{red}{c}$,inner sep=0pt,minimum size=.175cm] (c) at (-2.25,6) {{}};

\node[circle,draw=black!100, fill=black!100, label=above:$\textcolor{red}{\neg a}$,inner sep=0pt,minimum size=.175cm] (Na) at (2.25,6) {{}};

\node[circle,draw=black!100, fill=black!100, label=right:$\textcolor{red}{d}\;\;$,inner sep=0pt,minimum size=.175cm] (Nd) at (0,6) {{}};
\node[circle,draw=black!100, fill=black!100,label=above:$\textcolor{red}{1}$,inner sep=0pt,minimum size=.175cm] (1) at (0,8) {{}};

\path (d) edge[-] node {{}} (c);
\path (d) edge[-] node {{}} (Na);
\path (d) edge[-] node {{}} (0);

\path (1) edge[-] node {{}} (Nd);
\path (Nd) edge[-] node {{}} (a);
\path (Nd) edge[-] node {{}} (Nc);
\path (1) edge[-] node {{}} (c);
\path (1) edge[-] node {{}} (Na);
\path (a) edge[-] node {{}} (b);
\path (b) edge[-] node {{}} (c);
\path (Nc) edge[-] node {{}} (Nb);
\path (Nb) edge[-] node {{}} (Na);
\path (a) edge[-] node {{}} (0);
\path (Nc) edge[-] node {{}} (0);

\path (Na) edge[loop right,blue] node {{}} (Na);
\path (d) edge[loop left,blue] node {{}} (d);
\path (Nd) edge[loop left,blue] node {{}} (Nd);
\path (1) edge[loop left,blue] node {{}} (1);
\path (c) edge[loop left,blue] node {{}} (c);
\path (a) edge[loop left,blue] node {{}} (a);
\path (Nc) edge[loop right,blue] node {{}} (Nc);
\path (0) edge[loop left,blue] node {{}} (0);

\path (b) edge[->,blue,bend right] node {{}} (a);
\path (Nb) edge[->,blue,bend left] node {{}} (Nc);
\end{tikzpicture}\qquad\quad\begin{tikzpicture}[->,>=stealth',shorten >=1pt,shorten <=1pt, auto,node
distance=2.5cm,semithick]
\tikzstyle{every state}=[fill=gray!20,draw=none,text=black]
\node[circle,draw=black!100, fill=black!100, label=below:$\textcolor{red}{\bot}$,inner sep=0pt,minimum size=.175cm] (0) at (0,0) {{}};
\node[circle,draw=black!100, fill=black!100, label=right:$\textcolor{red}{\Diamond p \wedge \Diamond \neg p}$,inner sep=0pt,minimum size=.175cm] (d) at (0,2) {{}};

\node[circle,draw=black!100, fill=black!100, label=below:$\textcolor{red}{\Box p}$,inner sep=0pt,minimum size=.175cm] (a) at (-2.25,2) {{}};
\node[circle,draw=black!100, fill=black!100, label=below:$\;\;\;\;\textcolor{red}{\Box\neg p}$,inner sep=0pt,minimum size=.175cm] (Nc) at (2.25,2) {{}};

\node[circle,draw=black!100, fill=black!100, label=left:$\textcolor{red}{p}$,inner sep=0pt,minimum size=.175cm] (b) at (-2.25,4) {{}};
\node[circle,draw=black!100, fill=black!100, label=right:$\textcolor{red}{\neg p}$,inner sep=0pt,minimum size=.175cm] (Nb) at (2.25,4) {{}};
\node[circle,draw=black!100, fill=black!100,label=above:$\textcolor{red}{\Diamond p}$,inner sep=0pt,minimum size=.175cm] (c) at (-2.25,6) {{}};

\node[circle,draw=black!100, fill=black!100, label=above:$\;\;\;\textcolor{red}{\Diamond\neg p}$,inner sep=0pt,minimum size=.175cm] (Na) at (2.25,6) {{}};

\node[circle,draw=black!100, fill=black!100, label=right:$\textcolor{red}{\Box p \vee\Box\neg p}\;\;$,inner sep=0pt,minimum size=.175cm] (Nd) at (0,6) {{}};
\node[circle,draw=black!100, fill=black!100,label=above:$\textcolor{red}{\top}$,inner sep=0pt,minimum size=.175cm] (1) at (0,8) {{}};

\path (d) edge[-] node {{}} (c);
\path (d) edge[-] node {{}} (Na);
\path (d) edge[-] node {{}} (0);

\path (1) edge[-] node {{}} (Nd);
\path (Nd) edge[-] node {{}} (a);
\path (Nd) edge[-] node {{}} (Nc);
\path (1) edge[-] node {{}} (c);
\path (1) edge[-] node {{}} (Na);
\path (a) edge[-] node {{}} (b);
\path (b) edge[-] node {{}} (c);
\path (Nc) edge[-] node {{}} (Nb);
\path (Nb) edge[-] node {{}} (Na);
\path (a) edge[-] node {{}} (0);
\path (Nc) edge[-] node {{}} (0);

\path (Na) edge[loop right,blue] node {{}} (Na);
\path (d) edge[loop left,blue] node {{}} (d);
\path (Nd) edge[loop left,blue] node {{}} (Nd);
\path (1) edge[loop left,blue] node {{}} (1);
\path (c) edge[loop left,blue] node {{}} (c);
\path (a) edge[loop left,blue] node {{}} (a);
\path (Nc) edge[loop right,blue] node {{}} (Nc);
\path (0) edge[loop left,blue] node {{}} (0);

\path (b) edge[->,blue,bend right] node {{}} (a);
\path (Nb) edge[->,blue,bend left] node {{}} (Nc);
\end{tikzpicture}
\end{center}\caption{Hasse diagram of an $\mathsf{S5}$ epistemic ortholattice with two labelings}\label{Fig1}
\end{figure}
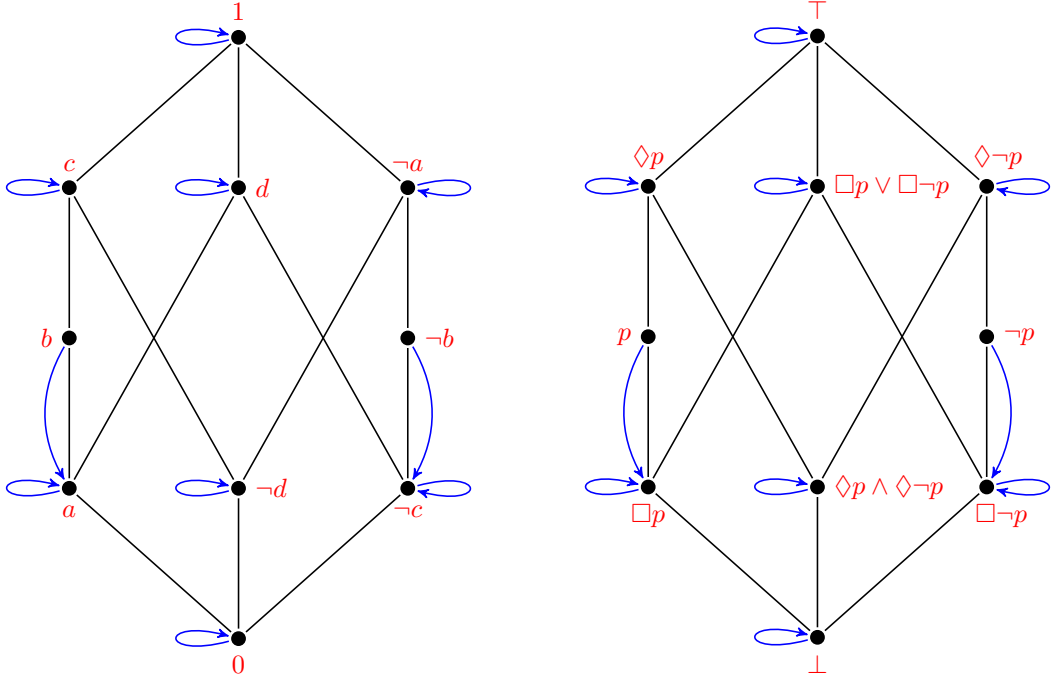

\begin{example}\label{KeyEx} Figure \ref{Fig1} displays the Hasse diagram of an $\mathsf{S5}$ epistemic ortholattice $L$, labeled in two ways. Recall that a line segment going \textit{upward} from $x$ to $y$ means that $x\leq y$ and there is no third element $z$ with $x\leq z\leq y$. The $\Box$ operation is depicted by the blue arrows. Note the failure of distributivity:
\[(p\vee\neg p)\wedge (\Diamond p\wedge\Diamond \neg p) = 1\wedge (\Diamond p\wedge\Diamond\neg p)= \Diamond p\wedge\Diamond \neg p\neq 0\]
and yet \[(p\wedge\Diamond \neg p) \vee (\neg p\wedge\Diamond p)= 0\vee 0 = 0.\]
That distributivity fails follows from the failure of orthomodularity. Recall this is the condition that $a\leq b$ implies $a\vee (\neg a\wedge b) =b$. Yet in Figure \ref{Fig1} we have $p\leq\Diamond p$ and  $p\vee (\neg p\wedge\Diamond p)=p\vee 0=p\neq\Diamond p$. 

Next, observe that \[ p\wedge\Diamond \neg p = 0\mbox{ and yet }\Diamond \neg p\not\leq \neg p.\] This shows that the orthocomplementation $\neg$ is not \textit{pseudocomplementation} (recall Lemma \ref{PseudoLem}). 

Also observe that \[ (p\vee\Box\neg p) \wedge \lozenge p  = \top\wedge\lozenge p=\lozenge p\not\leq  p.\] This shows that disjunctive syllogism fails.

Finally, note that the \textit{subortholattice} of $L$ generated by $p$ is the four-element \textit{Boolean} algebra:
\begin{center}
\begin{tikzpicture}[->,>=stealth',shorten >=1pt,shorten <=1pt, auto,node
distance=2.5cm,semithick]
\tikzstyle{every state}=[fill=gray!20,draw=none,text=black]

\node[circle,draw=black!100, fill=black!100,label=below:$\textcolor{red}{\bot}$,inner sep=0pt,minimum size=.175cm] (botL) at (0,0) {{}};
\node[circle,draw=black!100,fill=black!100, label=left:$\textcolor{red}{p}$,inner sep=0pt,minimum size=.175cm] (AL) at (-1,1) {{}};
\node[circle,draw=black!100,fill=black!100, label=right:$\textcolor{red}{\neg p}$,inner sep=0pt,minimum size=.175cm] (CL) at (1,1) {{}};
\node[circle,draw=black!100,fill=black!100, label=above:$\textcolor{red}{\top}$,inner sep=0pt,minimum size=.175cm] (B'L) at (0,2) {{}};

\path (CL) edge[-] node {{}} (botL);
\path (AL) edge[-] node {{}} (botL);
\path (B'L) edge[-] node {{}} (AL);
\path (B'L) edge[-] node {{}} (CL);
\end{tikzpicture}
\end{center}
Thus, if we think of $p$ as a non-modal proposition such as `it is raining' and generate further propositions using disjunction, conjunction, and negation, the result is Boolean. This corresponds to the fact that the failures of distributivity, orthomodularity, pseudocomplementation, and disjunctive syllogism discussed above essentially involve epistemic modals. In \S~\ref{DistinguishSec}, we return to this observation and show how to recover full classical reasoning for a Boolean fragment of our language interpreted in Boolean subalgebras of our lattices.\end{example}

\subsubsection{Language and consequence}

We can use epistemic ortholattices to interpret the following propositional modal language.

\begin{definition}\label{ELDef} Let $\mathcal{EL}$ be the language generated by the grammar
\[\varphi::= \top\mid p\mid \neg\varphi\mid (\varphi\wedge\varphi)\mid\Box \varphi\]
where $p$ belongs to a countably infinite set $\mathsf{Prop}$ of propositional variables. 
\end{definition}

\noindent We define  $\bot:=\neg \top $, $\varphi\vee\psi:=\neg(\neg\varphi\wedge\neg\psi)$, and $\Diamond\varphi:=\neg\Box\neg\varphi$. Again we use the same symbols for the connectives of $\mathcal{EL}$ and the operations in epistemic ortholattices, trusting that no confusion will arise.

The algebraic semantics follows the same approach as in \S~\ref{LangCons1}.

\begin{definition} A \textit{valuation} on a modal ortholattice $\langle A,\vee,0,\wedge,1,\neg,\Box\rangle$ is a map $\theta: \mathsf{Prop}\to A$. Such a $\theta$ extends to $\tilde{\theta}:\mathcal{EL}\to A$  by: $\tilde{\theta}(\top)=1$, $\tilde{\theta}(\neg\varphi)=\neg\tilde{\theta}(\varphi)$, $\tilde{\theta}(\varphi\wedge\psi)=\tilde{\theta}(\varphi)\wedge \tilde{\theta}(\psi)$ , and $\tilde{\theta}(\Box\varphi)=\Box\tilde{\theta}(\varphi)$.
\end{definition}

\begin{definition}\label{AlgCon} Given a class $\mathbf{C}$ of modal ortholattices, define the semantic consequence relation $\vDash_\mathbf{C}$, a binary relation on $\mathcal{EL}$, as follows: $\varphi\vDash_\mathbf{C}\psi$ if for every $L\in\mathbf{C}$ and valuation $\theta$ on $L$, we have $\tilde{\theta}(\varphi)\leq \tilde{\theta}(\psi)$, where $\leq$ is the lattice order of $L$.
\end{definition}

\subsubsection{Logic, soundness and completeness}

Building on \S~\ref{Log1}, axiomatizing the semantic consequence relation of Definition \ref{AlgCon} is straightforward.

\begin{definition}\label{EODef} An \textit{epistemic orthologic} is a binary relation $\vdash$ on the set $\mathcal{EL}$ of formulas satisfying for all $\varphi,\psi\in\mathcal{EL}$ conditions 1-10 of Definition \ref{OrthoLogicDef} plus:
\begin{center}
\begin{tabular}{ll}
11. if $\varphi\vdash\psi$, then $\Box\varphi\vdash\Box\psi$;\qquad\qquad  & 14. $\Box \varphi\vdash\varphi$; \\
12. $\Box\varphi\wedge\Box\psi\vdash \Box(\varphi\wedge\psi)$; & 15. $\neg \varphi\wedge\Diamond \varphi\vdash \bot$ (Wittgenstein's Law). \\
13. $\varphi\vdash\Box\top$;
\end{tabular}
\end{center}
As the intersection of epistemic orthologics is clearly an epistemic orthologic, there is a smallest epistemic orthologic, denoted \textsf{EO} or $\vdash_{\mathsf{EO}}$.\footnote{\label{EOFitch}A Fitch-style proof system for $\mathsf{EO}$ can be obtained from the Fitch-style proof system for $\mathsf{O}$ in Figure \ref{FitchRules} by adding rules of $\Box$ introduction, $\Box$ elimination, and Strict Reiteration as in \citealt{Fitch1966}, plus an additional rule of Epistemic Contradiction that is like the $\neg$E rule in Figure \ref{FitchRules} but with $\neg\varphi$ replaced by $\neg\Box\varphi$.}
\end{definition}

Wittgenstein's Law yields the following noteworthy property of epistemic orthologics.

\begin{lemma}\label{BoxBotLem} For any epistemic orthologic $\vdash$ and $\varphi\in\mathcal{EL}$, if $\Box\varphi\vdash\bot$, then $\varphi\vdash\bot$.
\end{lemma}
\begin{proof} If $\Box\varphi\vdash\bot$, then $\top\vdash\neg\Box\varphi$, in which case $\varphi\vdash\varphi\wedge\neg\Box\varphi\vdash\neg\neg\varphi\wedge\Diamond\neg\varphi\vdash\bot$, so $\varphi\vdash\bot$. 
\end{proof}
\noindent Conversely, if one assumes the principle in Lemma \ref{BoxBotLem} together with principles 11 and 14 in Definition \ref{EODef} on top of orthologic, then Wittgenstein's Law is derivable (see the proof of Fact \ref{GeneralizedWittLaw}).

As in \S~\ref{Log1}, a completeness theorem can be proved using standard techniques of algebraic logic.

\begin{theorem}\label{AlgComp2} The logic $\textsf{EO}$ is sound and complete with respect to the class $\mathbf{EO}$ of epistemic ortholattices according to the consequence relation of Definition \ref{AlgCon}: for all $\varphi,\psi\in\mathcal{EL}$, we have $\varphi\vdash_\mathsf{EO}\psi$ if and only if $\varphi\vDash_\mathbf{EO}\psi$.
\end{theorem}
\begin{proof} By the same strategy as in the proof of Theorem \ref{AlgComp1}, checking that the Lindenbaum-Tarski algebra of $\textsf{EO}$ is an epistemic ortholattice.
\end{proof}

\begin{remark}Soundness and completeness with respect to \textsf{S5} epistemic ortholattices is also straightforward by adding the rules $\Box\varphi\vdash\Box\Box\varphi$ (known as $\mathsf{4}$) and $\Diamond \varphi\vdash\Box\Diamond\varphi$ (known as $\mathsf{5}$) to Definition \ref{EODef}. However, neither of these rules, nor the rule $\varphi\vdash\Box\Diamond\varphi$ (known as $\mathsf{B}$), is valid with respect to epistemic ortholattices in general. Moreover, the evidence about the status of these inferences for epistemic modality is mixed. 

On the one hand, \citet{Moss2014}  argues that patterns from nested epistemic modals tell against the collapse principles that these axioms together entail. For instance, the sentences in \ref{stackedmight} do not sound obviously equivalent but instead are intuitively increasingly strong:\footnote{A potential confound in these cases comes from the phenomenon of modal concord, where different modals do not contribute independent modal meanings; see \citealt{Wijnbergen:2020} for an overview. Moss's claim is that these cases are not cases of modal concord; judgments are, to be sure, somewhat tenuous.}

\ex.\label{stackedmight}\a. John might possibly win.
 \b.  John might win. 
\c. John certainly might win. 

Judgments here are somewhat difficult to ascertain, however, given that it is very difficult to directly stack epistemic modals in English. 

Pulling in the other direction, conjunctions that instantiate violations of \textsf{4}, \textsf{5}, and \textsf{B}---that is, sentences with the form $\Box p\wedge\neg \Box\Box p$, $\lozenge p\wedge\neg\Box\lozenge p$, and $ p\wedge\neg \Box\lozenge p$, respectively---do sound inconsistent, as in \ref{antis5} (using duality to improve readability):

\ex.\label{antis5} 
\a. \# John must be the winner, but maybe he might not be  the winner.
\b.\# John might be the winner, but it might be that he must not be the winner. 
\c. \# John is the winner, but it might be that he must not be the winner. 

This extends to embedded environments, suggesting this is not simply a Moorean phenomenon but should be accounted for logically.

In fact, the account so far has a nice way to make sense of both of these kinds of evidence. On the one hand, \textsf{4}, \textsf{5}, and \textsf{B} are all invalid according to our consequence relation for $\mathbf{EO}$. On the other hand, the conjunctions above that instantiate their violations are all inconsistent. Reasoning from that inconsistency to the validity of \textsf{4}, \textsf{5}, and \textsf{B}, although classically valid, is blocked in our system since negation is not pseudocomplementation. This lets us account both for the fact that nested modals appear not to collapse and for the fact that conjunctions that witness the failures of  \textsf{4}, \textsf{5}, and \textsf{B} appear inconsistent. 

One might think that without $\textsf{4}$, which is equivalent to $\Diamond\Diamond \varphi\vdash\Diamond\varphi$, the other principles of our logic could not account for the fact that $p\wedge\Diamond\Diamond\neg p$ should be inconsistent. In fact, we do not need $\textsf{4}$ to account for this. For the following, let $\Diamond^0\varphi=\varphi$ and $\Diamond^{n+1}\varphi=\Diamond\Diamond^n\varphi$.
\end{remark}

\begin{fact}\label{GeneralizedWittLaw} For any epistemic orthologic $\vdash$, $n\in\mathbb{N}$, and $\varphi\in\mathcal{EL}$, we have $\varphi\wedge\Diamond^n \neg\varphi\vdash\bot$.
\end{fact}
\begin{proof} By induction on $n$. Assume $\psi\wedge\Diamond^n\neg\psi\vdash\bot$ for all $\psi\in\mathcal{EL}$. Then for any $\varphi\in\mathcal{EL}$, \[\Box (\varphi\wedge\Diamond^{n+1}\neg\varphi)\vdash \Box\varphi\wedge \Box\Diamond^{n+1}\neg\varphi  \vdash \Box\varphi\wedge \Diamond^{n+1}\neg\varphi  \vdash \Box\varphi\wedge \Diamond^{n}\neg\Box\varphi \vdash \bot,\]
using the inductive hypothesis for the last step. Thus, $\varphi\wedge\Diamond^{n+1}\neg\varphi\vdash\bot$ by Lemma \ref{BoxBotLem}.\end{proof}

\subsubsection{Distinguishing Boolean propositions}\label{DistinguishSec}

We now make precise the idea from \S~\ref{ConsSec} that classical principles should hold for the ``non-modal fragment.'' If we understand the basic elements $p,q,r \in \mathsf{Prop}$ of our inductively defined formal language as \textit{propositional variables}, standing in for arbitrary propositions including epistemic modal propositions, then we do not want classical principles like $p\wedge (q\vee r)\vdash  (p\wedge q)\vee (p\wedge r)$, since accepting such a principle for $p,q,r$ means accepting it for all propositions (cf.~\citealt[pp.~147-8]{Burgess2003}). However, we can add to the basic elements of our formal language a set $\mathsf{Bool}$ whose elements we think of as variables only for \textit{non-epistemic}, or \emph{Boolean}, propositions.\footnote{The distinction between viewing atomic formulas of the propositional modal language as genuine propositional variables vs.~non-epistemic variables has a precedent in the literature on dynamic epistemic logic (see, e.g., \citealt{HHI2012,HHI2013}).} Typographically,  $p,q,r,\dots$ are elements of $\mathsf{Prop}$, whereas $\mathtt{p},\mathtt{q},\mathtt{r},\dots$ are elements of $\mathsf{Bool}$.

\begin{definition}\label{ELPlus} Let $\mathcal{EL}^+$ be the language generated by the grammar
\[\varphi::= \top\mid p\mid \mathtt{p}\mid  \neg\varphi\mid (\varphi\wedge\varphi)\mid\Box \varphi\]
where $p$ belongs to a countably infinite set $\mathsf{Prop}$ and $\mathtt{p}$ belongs to a countably infinite set $\mathsf{Bool}$. A formula of $\mathcal{EL}^+$ is said to be \textit{Boolean} if all its propositional variables are from $\mathsf{Bool}$ and it does not contain $\Box$.\end{definition}

In line with our idea that the only failures of classicality come from epistemic modals, we interpret the variables of $\mathsf{Bool}$ in a Boolean subalgebra of our ambient ortholattice of propositions.

\begin{definition} A \textit{modal ortho-Boolean lattice} is a tuple $\langle A,B,\vee,0,\wedge,1,\neg,\Box\rangle$ where $\langle A,\vee,0,\wedge,1,\neg,\Box\rangle$ is a modal ortholattice and $\langle B,\vee_{\mid B},0,\wedge_{\mid B},1,\neg_{\mid B}\rangle $ is a Boolean algebra where $B\subseteq A$ and $\vee_{\mid B}$, $\wedge_{\mid B}$, and $\neg_{\mid B}$ are the restrictions of $\vee$, $\wedge$, and $\neg$, respectively, to $B$. \end{definition}

\noindent Note that every modal ortholattice can be expanded to a modal ortho-Boolean lattice by taking $B=\{0,1\}$, so modal ortho-Boolean lattices may be viewed as a generalization of modal ortholattices.

So far there is no required connection between the distinguished Boolean subalgebra of a modal ortho-Boolean lattice and the ambient modal ortholattice: any Boolean subalgebra can be distinguished. However, there appear to be intuitively valid principles relating Boolean---but not arbitrary---propositions and epistemic modals, and capturing such principles requires a certain coherence between the Boolean subalgebra $B$ and the ambient epistemic ortholattice. Consider, for example, the following.

\begin{definition}\label{LevelwiseBooldeanDef} Given a modal ortho-Boolean lattice $L=\langle A,B,\vee,0,\wedge,1,\neg,\Box\rangle$, define:
\begin{itemize}
\item  $B_0=B$;
\item  $B_{n+1} $ is the subortholattice of $\langle A,\vee,0,\wedge,1,\neg\rangle$ generated by $\{\Box b\mid b\in B_n\}$.
\end{itemize}
Then $L$ is \textit{level-wise Boolean} if $B_n$ is Boolean for each $n\in\mathbb{N}$.
\end{definition}

\noindent The motivation for this condition is straightforward: no natural language counterexample to a classical inference that we have found is such that all propositions come from the same level $B_n$. For example, the counterexample to pseudocomplementation, going from $p\wedge\Diamond\neg p=0$ to $\Diamond \neg p\leq\neg p$, involves $p,\neg p\in B_n$ and $\Diamond\neg p\in B_{n+1}$. Similar points apply to our counterexamples involving distributivity, disjunctive syllogism, and orthomodularity. This observation is, to our knowledge, novel, and it suggests a picture on which classical reasoning \emph{across different epistemic levels} can be dangerous, but classical reasoning within a given epistemic level is safe. We can model that picture with level-wise Boolean epistemic ortholattices:

\begin{definition}\label{OrthoBoole}An \textit{epistemic ortho-Boolean lattice} is a modal ortho-Boolean lattice $\langle A,B,\vee,0,\wedge,1,\neg,\Box\rangle$ that is level-wise Boolean and such that $\langle A,\vee,0,\wedge,1,\neg,\Box\rangle$ is an epistemic ortholattice.\end{definition}

\noindent For checking the level-wise Boolean condition, note that if $B_{n+1}=B_n$, then $B_k=B_n$ for all $k\geq n$; and if $A$ is finite, then we are bound to reach such a fixed point $B_n$.

\begin{example}\label{LevelwiseBooleEx} Consider the epistemic ortholattice in Figure \ref{Fig1} with $B=\{\bot,p,\neg p,\top\}$, highlighted on the left of Figure \ref{StratFig}. Then $B$ forms a subortholattice and a four-element Boolean algebra. Moreover, $B_1=\{\bot,\Box p, \Diamond\neg p, \Diamond p, \Box \neg p, \Box p\vee\Box\neg p,\Diamond p\wedge\Diamond\neg p, \top\}$ forms an eight-element Boolean algebra, highlighted on the right of Figure \ref{StratFig}; and $B_2=B_1$, so  $B_n=B_1$ for all $n\geq 1$. Thus, by equipping the epistemic ortholattice in Figure \ref{Fig1} with $B$, we obtain an epistemic ortho-Boolean lattice.
\end{example}

\begin{figure}[h]
\begin{center}
\begin{tikzpicture}[scale=.9,->,>=stealth',shorten >=1pt,shorten <=1pt, auto,node
distance=2.5cm,semithick]
\tikzstyle{every state}=[fill=gray!20,draw=none,text=black]
\node[circle,draw=black!100, fill=yellow!100, label=right:$\textcolor{red}{\bot}$,inner sep=0pt,minimum size=.175cm] (0) at (0,0) {{}};
\node[circle,draw=black!100, fill=black!100, label=right:{\small$\textcolor{red}{\Diamond p \wedge \Diamond \neg p}$},inner sep=0pt,minimum size=.175cm] (d) at (0,2) {{}};

\node[circle,draw=black!100, fill=black!100, label=below:$\textcolor{red}{\Box p}$,inner sep=0pt,minimum size=.175cm] (a) at (-2.25,2) {{}};
\node[circle,draw=black!100, fill=black!100, label=below:$\;\;\;\;\textcolor{red}{\Box\neg p}$,inner sep=0pt,minimum size=.175cm] (Nc) at (2.25,2) {{}};

\node[circle,draw=black!100, fill=yellow!100, label=left:$\textcolor{red}{p}$,inner sep=0pt,minimum size=.175cm] (b) at (-2.25,4) {{}};
\node[circle,draw=black!100, fill=yellow!100, label=right:$\textcolor{red}{\neg p}$,inner sep=0pt,minimum size=.175cm] (Nb) at (2.25,4) {{}};
\node[circle,draw=black!100, fill=black!100,label=above:$\textcolor{red}{\Diamond p}$,inner sep=0pt,minimum size=.175cm] (c) at (-2.25,6) {{}};

\node[circle,draw=black!100, fill=black!100, label=above:$\;\;\;\textcolor{red}{\Diamond\neg p}$,inner sep=0pt,minimum size=.175cm] (Na) at (2.25,6) {{}};

\node[circle,draw=black!100, fill=black!100, label=right:{\small$\textcolor{red}{\Box p \vee\Box\neg p}\;\;$},inner sep=0pt,minimum size=.175cm] (Nd) at (0,6) {{}};
\node[circle,draw=black!100, fill=yellow!100,label=right:$\textcolor{red}{\top}$,inner sep=0pt,minimum size=.175cm] (1) at (0,8) {{}};

\path (d) edge[-] node {{}} (c);
\path (d) edge[-] node {{}} (Na);
\path (d) edge[-] node {{}} (0);

\path (1) edge[-] node {{}} (Nd);
\path (Nd) edge[-] node {{}} (a);
\path (Nd) edge[-] node {{}} (Nc);
\path (1) edge[-] node {{}} (c);
\path (1) edge[-] node {{}} (Na);
\path (a) edge[-] node {{}} (b);
\path (b) edge[-] node {{}} (c);
\path (Nc) edge[-] node {{}} (Nb);
\path (Nb) edge[-] node {{}} (Na);
\path (a) edge[-] node {{}} (0);
\path (Nc) edge[-] node {{}} (0);

\path (Na) edge[loop right,blue] node {{}} (Na);
\path (d) edge[loop left,blue] node {{}} (d);
\path (Nd) edge[loop left,blue] node {{}} (Nd);
\path (1) edge[loop left,blue] node {{}} (1);
\path (c) edge[loop left,blue] node {{}} (c);
\path (a) edge[loop left,blue] node {{}} (a);
\path (Nc) edge[loop right,blue] node {{}} (Nc);
\path (0) edge[loop left,blue] node {{}} (0);

\path (b) edge[->,blue,bend right] node {{}} (a);
\path (Nb) edge[->,blue,bend left] node {{}} (Nc);
\end{tikzpicture}\qquad\qquad\qquad\begin{tikzpicture}[scale=.9,->,>=stealth',shorten >=1pt,shorten <=1pt, auto,node
distance=2.5cm,semithick]
\tikzstyle{every state}=[fill=gray!20,draw=none,text=black]
\node[circle,draw=black!100, fill=yellow!100, label=right:$\textcolor{red}{\bot}$,inner sep=0pt,minimum size=.175cm] (0) at (0,0) {{}};
\node[circle,draw=black!100, fill=yellow!100, label=right:{\small$\textcolor{red}{\Diamond p \wedge \Diamond \neg p}$},inner sep=0pt,minimum size=.175cm] (d) at (0,2) {{}};

\node[circle,draw=black!100, fill=yellow!100, label=below:$\textcolor{red}{\Box p}$,inner sep=0pt,minimum size=.175cm] (a) at (-2.25,2) {{}};
\node[circle,draw=black!100, fill=yellow!100, label=below:$\;\;\;\;\textcolor{red}{\Box\neg p}$,inner sep=0pt,minimum size=.175cm] (Nc) at (2.25,2) {{}};

\node[circle,draw=black!100, fill=black!100, label=left:$\textcolor{red}{p}$,inner sep=0pt,minimum size=.175cm] (b) at (-2.25,4) {{}};
\node[circle,draw=black!100, fill=black!100, label=right:$\textcolor{red}{\neg p}$,inner sep=0pt,minimum size=.175cm] (Nb) at (2.25,4) {{}};
\node[circle,draw=black!100, fill=yellow!100,label=above:$\textcolor{red}{\Diamond p}$,inner sep=0pt,minimum size=.175cm] (c) at (-2.25,6) {{}};

\node[circle,draw=black!100, fill=yellow!100, label=above:$\;\;\;\textcolor{red}{\Diamond\neg p}$,inner sep=0pt,minimum size=.175cm] (Na) at (2.25,6) {{}};

\node[circle,draw=black!100, fill=yellow!100, label=right:{\small$\textcolor{red}{\Box p \vee\Box\neg p}\;\;$},inner sep=0pt,minimum size=.175cm] (Nd) at (0,6) {{}};
\node[circle,draw=black!100, fill=yellow!100,label=right:$\textcolor{red}{\top}$,inner sep=0pt,minimum size=.175cm] (1) at (0,8) {{}};

\path (d) edge[-] node {{}} (c);
\path (d) edge[-] node {{}} (Na);
\path (d) edge[-] node {{}} (0);

\path (1) edge[-] node {{}} (Nd);
\path (Nd) edge[-] node {{}} (a);
\path (Nd) edge[-] node {{}} (Nc);
\path (1) edge[-] node {{}} (c);
\path (1) edge[-] node {{}} (Na);
\path (a) edge[-] node {{}} (b);
\path (b) edge[-] node {{}} (c);
\path (Nc) edge[-] node {{}} (Nb);
\path (Nb) edge[-] node {{}} (Na);
\path (a) edge[-] node {{}} (0);
\path (Nc) edge[-] node {{}} (0);

\path (Na) edge[loop right,blue] node {{}} (Na);
\path (d) edge[loop left,blue] node {{}} (d);
\path (Nd) edge[loop left,blue] node {{}} (Nd);
\path (1) edge[loop left,blue] node {{}} (1);
\path (c) edge[loop left,blue] node {{}} (c);
\path (a) edge[loop left,blue] node {{}} (a);
\path (Nc) edge[loop right,blue] node {{}} (Nc);
\path (0) edge[loop left,blue] node {{}} (0);

\path (b) edge[->,blue,bend right] node {{}} (a);
\path (Nb) edge[->,blue,bend left] node {{}} (Nc);
\end{tikzpicture}
\end{center}
\caption{Highlighting in yellow the Boolean subalgebras $B_0$ (left) and $B_1$ (right) of the epistemic ortholattice from Figure \ref{Fig1}.} \label{StratFig}
\end{figure}
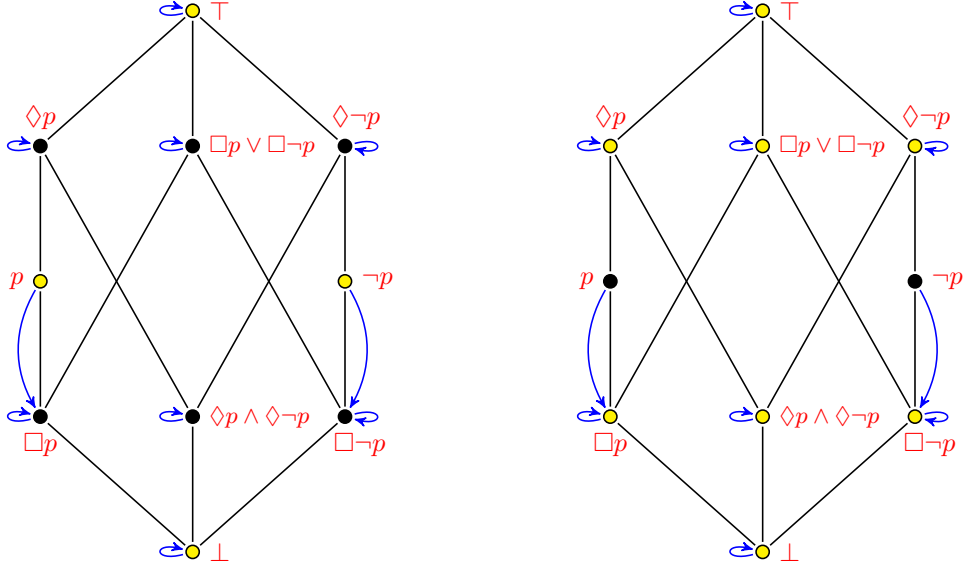

Let us now turn to the semantics of $\mathcal{EL}^+$.

\begin{definition} A valuation $\theta$ on a modal ortho-Boolean lattice $\langle A,B,\vee,0,\wedge,1,\neg,\Box\rangle$ is a map $\theta:\mathsf{Prop}\cup\mathsf{Bool}\to A$ such that for all $\mathtt{p}\in\mathsf{Bool}$, $\theta(\mathtt{p})\in B$. Such a valuation extends to $\tilde{\theta}: \mathcal{EL}^+\to A$ by: $\tilde{\theta}(\top)=1$, $\tilde{\theta}(\neg\varphi)=\neg\tilde{\theta}(\varphi)$, $\tilde{\theta}(\varphi\wedge\psi)=\tilde{\theta}(\varphi)\wedge \tilde{\theta}(\psi)$, and $\tilde{\theta}(\Box\varphi)=\Box\tilde{\theta}(\varphi)$.
\end{definition}

Corresponding to the different subortholattices $B_n$, we have a hierarchy of sublanguages $\mathcal{B}_n$.

\begin{definition}$\,$
\begin{itemize}
\item Let $\mathcal{B}_0$ be the set of Boolean formulas as in Definition \ref{ELPlus}.
\item Let $\mathcal{B}_{n+1}$ be the smallest set of formulas that includes $\{\Box\varphi \mid \varphi\in \mathcal{B}_{n}\}$ and is closed under $\neg$ and $\wedge$.
\end{itemize}
\end{definition}
An obvious induction on the structure of formulas shows that formulas at each level are interpreted in the corresponding $B_n$.

\begin{lemma}\label{InterpretInBn} Let $\theta$ be a valuation on a modal ortho-Boolean lattice $\langle A,B,\vee,0,\wedge,1,\neg,\Box\rangle$. Then for any $n\in\mathbb{N}$ and $\beta\in \mathcal{B}_n$, we have $\tilde{\theta}(\beta)\in B_n$.
\end{lemma}

Our consequence relation for $\mathcal{EL}^+$ is defined just like our consequence relation for $\mathcal{EL}$ but quantifying over modal ortho-Boolean lattices instead of modal ortholattices.

\begin{definition}\label{AlgCon+} Given a class $\mathbf{C}$ of modal ortho-Boolean lattices, define the semantic consequence relation $\vDash_\mathbf{C}^+$, a binary relation on $\mathcal{EL}^+$, as follows: $\varphi\vDash_\mathbf{C}^+\psi$ if for every $L\in\mathbf{C}$ and valuation $\theta$ on $L$, we have $\tilde{\theta}(\varphi)\leq \tilde{\theta}(\psi)$, where $\leq$ is the lattice order of $L$. 
\end{definition}

Turning to the logic, we now explicitly include the principle of distributivity for formulas at a given epistemic level.\footnote{A Fitch-style proof system for $\mathsf{EO}^+$ can be obtained from the Fitch-style proof system for $\mathsf{EO}$ sketched in Footnote \ref{EOFitch} by adding a rule of Level-wise Reiteration, which is like Fitch's rule of Reiteration except that it only allows reiteration of formulas from $\mathcal{B}_n$ into subproofs whose opening assumptions are from $\mathcal{B}_n$.} 

\begin{definition}\label{EOplus} Let $\mathsf{EO}^+$ be defined just like $\mathsf{EO}$ in Definition \ref{EODef} for arbitrary formulas $\varphi,\psi,\chi\in\mathcal{EL}^+$ but in addition we have for all $n\in \mathbb{N}$ and $\alpha,\beta,\gamma\in\mathcal{B}_n$:
\begin{itemize}
\item[16.] $\alpha\wedge (\beta\vee\gamma)\vdash (\alpha\wedge \beta)\vee (\alpha\wedge\gamma)$.
\end{itemize}
\end{definition}

Adding distributivity for formulas at a given epistemic level allows us to derive all principles of classical logic for formulas at that level. For example, we can derive the psueodcomplementation principle that $\alpha\wedge\beta\vdash\bot$ implies $\alpha\vdash\neg\beta$ for $\alpha,\beta\in\mathcal{B}_n$ as in the proof of Proposition \ref{PseudoToBoole}.

We can also derive principles of classical normal modal logic for formulas at a given epistemic level.

\begin{example}\label{DiamondBoxDiamond} We can derive that $\Diamond \delta\wedge \Box\gamma\vdash \Diamond (\delta\wedge\gamma)$ for $\delta,\gamma\in \mathcal{B}_n$ as follows:  \[\Diamond \delta\wedge\Box \gamma\wedge \Box\neg (\delta\wedge \gamma)\vdash\Diamond \delta\wedge \Box(\gamma\wedge \neg (\delta\wedge \gamma))\vdash \Diamond \delta\wedge \Box \neg\delta\vdash \bot,\] so by pseudocomplementation for formulas in $\mathcal{B}_{n+1}$, we have  $\Diamond \delta\wedge \Box\gamma\vdash \Diamond (\delta\wedge\gamma)$. Note that we would not want this principle for formulas of different epistemic levels, since $\Diamond p\wedge \Box\Diamond\neg p$ should not entail $\Diamond (p\wedge\Diamond\neg p)$.
\end{example}

More generally, we have the following completeness theorem.

\begin{theorem}\label{EO+comp1}  The logic $\textsf{EO}^+$ (summarized in Figure \ref{EO+Fig}) is sound and complete with respect to the class $\mathbf{EO}^+$ of all epistemic ortho-Boolean lattices according to the consequence relation of Definition \ref{AlgCon+}: for all $\varphi,\psi\in\mathcal{EL}^+$, we have $\varphi\vdash_{\mathsf{EO}^+}\psi$ if and only if $\varphi\vDash_\mathbf{EO}^+\psi$.
\end{theorem}

\begin{proof} Soundness is again straightforward. For completeness, as in the proofs of Theorems \ref{AlgComp1} and \ref{AlgComp2}, we consider the Lindenbaum-Tarski algebra $L$ of $\textsf{EO}^+$. Let $B=\{[\beta]\mid \beta\mbox{ Boolean}\}$, and note that $B$ forms a subortholattice of $L$. As a consequence of the distributivity rule of $\mathsf{EO}^+$, each $B_n$ generated from $B$ is a Boolean algebra under the restricted operations of $L$. Hence we have an epistemic ortho-Boolean lattice. Finally, let $\theta$ be the valuation on  $L$ with $\theta(p)=[p]$ for all $p\in\mathsf{Prop}$ and $\theta(\mathtt{p})=[\mathtt{p}]$ for all $\mathtt{p}\in\mathsf{Bool}$. The rest of the proof is the same as for Theorems \ref{AlgComp1} and \ref{AlgComp2}.\end{proof}

\noindent As before, soundness and completeness with respect to \textsf{S5} epistemic ortho-Boolean lattices is also straightforward by adding the rules $\Box\varphi\vdash\Box\Box\varphi$ and $\Diamond \varphi\vdash\Box\Diamond\varphi$ to Definition \ref{EOplus}.

 $\mathsf{EO}^+$ is thus our proposed base logic for epistemic modals. $\mathsf{EO}^+$ clearly satisfies all the desiderata of \S~\ref{des} except perhaps for conservativity. When it comes to conservativity,  $\mathsf{EO}^+$ preserves much of classical modal logic, but we will see in \S  ~\ref{EpExtSection} that there are plausible \textsf{E}-logics that preserve even more of classical modal logic. We will raise the question whether $\mathsf{EO}^+$ is ``conservative enough'' or whether we want to also commit to those further principles or to other principles of classical modal logic.

\begin{figure}[h]
\begin{center}
\begin{tabular}{|ll|}
\hline &\\
1. $\varphi\vdash\top$; & 6.  $\neg\neg \varphi\vdash\varphi$; \\
2. $\varphi\vdash\varphi$; & 7. $\varphi\wedge\neg\varphi\vdash\psi$;\\
3. $\varphi\wedge\psi\vdash\varphi$; \qquad\qquad &8. if $\varphi\vdash\psi$ and $\psi\vdash\chi$, then $\varphi\vdash\chi$;\\
4. $\varphi\wedge\psi\vdash\psi$; &9. if $\varphi\vdash\psi$ and $\varphi\vdash\chi$, then $\varphi\vdash \psi\wedge\chi$; \\
5.  $\varphi\vdash\neg\neg\varphi$; & 10. if $\varphi\vdash\psi$, then $\neg\psi\vdash\neg\varphi$.\\ & \\

11. if $\varphi\vdash\psi$, then $\Box\varphi\vdash\Box\psi$; & 14. $\Box \varphi\vdash\varphi$; \\
12. $\Box\varphi\wedge\Box\psi\vdash \Box(\varphi\wedge\psi)$; & 15. $\neg \varphi\wedge\Diamond \varphi\vdash \bot$;\\
13. $\varphi\vdash\Box\top $; & \\ & \\

16. $\alpha\wedge (\beta\vee\gamma)\vdash (\alpha\wedge \beta)\vee (\alpha\wedge\gamma)$ & \\
\quad\;\, for $\alpha,\beta,\gamma\in\mathcal{B}_n$. & \\  & \\
\hline
\end{tabular}
\end{center}
\caption{Principles of the epistemic orthologic $\mathsf{EO}^+$}\label{EO+Fig}
\end{figure}

\section{Possibility semantics}\label{PossSem}

While the algebraic semantics of Section~\ref{algebraic} shows perspicuously exactly how we are departing from the Boolean algebraic semantics underlying classical logic,
 some may be unsatisfied with the way in which the algebraic approach simply builds into our semantics---in the form of equations or inequalities on lattices---precisely the logical principles we want to get out. Thus, in this section, we turn to a more concrete and perhaps more intuitive and illuminating \textit{possibility semantics} for epistemic orthologic. We begin with possibility semantics for non-modal orthologic, reviewing the approach of \citealt{Goldblatt1974}.\footnote{A similar approach appears in \citealt{Dishkant1972}; for comparison of the two approaches, see \citealt[Remark 2.12]{Holliday2021}.} Although Goldblatt does not classify his semantics under the banner of ``possibility semantics,'' this classification is justified by the close connection, explained in \citealt{Holliday2021}, between his semantics for orthologic using a relation of \textit{compatibility} (or what he calls \textit{proximity}) and possibility semantics for classical logic using a relation of \textit{refinement}.\footnote{In fact, Goldblatt focuses on an equivalent semantics for orthologic using the complement of compatibility, namely \textit{incompatibility} or  \textit{orthogonality}, but whether one works with compatibility or incompatibility is just a matter of preference.} Indeed, classical possibility semantics as developed in \citealt{Humberstone1981} and \citealt{Holliday2014,HollidayForthA,HollidayForthB} can be recast in terms of compatibility instead of refinement (see Remark \ref{BooleanCase} and \citealt{Holliday2022}). Here we assume no previous exposure to possibility semantics of any kind. All we assume is familiarity with possible world semantics, of which possibility semantics is a generalization.

\subsection{Review of possibility semantics for orthologic}\label{ReviewPoss}

Possibility semantics for orthologic replaces the set $W$ of worlds from classical possible world semantics with a set $S$ of possibilities, endowed with a \textit{compatibility relation} $\between$ between possibilities. 

\begin{definition}\label{CompFrame} A (\textit{symmetric}) \textit{compatibility frame} is a pair $\mathcal{F}=\langle S,\between\rangle$ where $S$ is a nonempty set and $\between$ is a reflexive and symmetric binary relation on $S$. 
\end{definition}

A set $W$ of possible worlds may be regarded as a compatibility frame in which each $w\in W$ is compatible only with itself: $v\between w$ implies $v=w$. Think of $x\between y$ as meaning that \textit{$x$ does not settle as true anything that $y$ settles as false}. If $w$ and $v$ are complete possible worlds, then---assuming distinct worlds must differ in some respect---$w$ must settle something as true that $y$ settles as false. However, unlike possible worlds, distinct \textit{possibilities} may be compatible with each other---a sign of their \textit{partiality}. For example, the possibility that it is raining in Beijing is compatible with the possibility that it is sunny in Malibu. 

The above notion of compatibility is weaker than what might be called \textit{compossibility}: $x$ and $y$ are compossible if there is some possibility $z$ that settles as true everything that $x$ settles as true and everything that $y$ settles as true. As we shall see below, in the above sense of compatibility, a possibility $x$ settling as true that \textit{it might be raining in Beijing} can be \textit{compatible} with a possibility $y$ settling as true that \textit{it is not raining in Beijing}, since settling as true \textit{that it might be raining in Beijing} does not entail settling as false \textit{that it is not raining in Beijing}; however, such an $x$ and $y$ cannot be \textit{compossible}, since there can be no possibility $z$ that settles as true that \textit{it might be raining in Beijing but it isn't}.

\begin{remark} \citealt{Holliday2021,Holliday2022} considers more general compatibility frames in which $\between$ is only assumed to be reflexive,\footnote{In this case, Definition \ref{RegProp} below is modified to use `$\forall z\between^{-1}y$' instead of `$\forall z\between y$'.} as such frames can be used to represent arbitrary complete lattices; but here we use the term `compatibility frame' for only the symmetric frames, which will give rise to complete ortholattices.\end{remark}

We now turn to the question of what counts as a \textit{proposition}. In basic possible world semantics for classical logic, every set of worlds is (or corresponds to) a proposition; as a result, the collection of propositions ordered by inclusion forms a Boolean algebra. In semantics for \textit{non-classical} logics, not every set of worlds or possibilities can count as a proposition---for then we would still be stuck with a Boolean algebra of propositions. Thus, for example, in possible world semantics for intuitionistic logic (\citealt{Dummett1959}, \citealt{Grzegorczyk1964}, \citealt{Kripke1965}), only sets of worlds that are upward closed with respect to an information order count as legitimate propositions.  In possibility semantics for orthologic, only sets of possibilities that behave sensibly with respect to the compatibility relation count as legitimate propositions.

\begin{definition}\label{RegProp} Given a compatibility frame $\langle S,\between\rangle$, a set $A\subseteq S$ is said to be \textit{$\between$-regular} if for all $x\in S$,
\[x\not\in A \Rightarrow \exists y\between x\;\forall z\between y\;\, z\not\in A.\]
\end{definition}
\noindent The idea is that if $x$ does \textit{not} make a proposition $A$ true, then there should be a possibility $y$ compatible with $x$ that makes $A$ \textit{false}, so that all possibilities $z$ compatible with $y$ do not make $A$ true.\footnote{\label{ClosureOperator}Another way to view the definition of $\between$-regular sets is via the function $c_\between:\mathcal{P}(S)\to \mathcal{P}(S)$ defined as follows: for $A\subseteq S$, we have $c_\between(A)=\{x\in S\mid \forall y\between x\;\exists z\between y:z\in A\}$. This $c_\between$ is a \textit{closure operator} (i.e., inflationary, idempotent, and monotone operation) on $\mathcal{P}(S)$. The \textit{fixpoints} of $c_\between$, i.e., those $A$ such that $A=c_\between(A)$, are exactly the $\between$-regular sets as defined above.} See Figure~\ref{RegDefFig}. If $x$ already makes $A$ false, so no possibility compatible with $x$ makes $A$ true, then we can simply take $y=x$. The interesting case is where $x$ neither makes $A$ true nor makes $A$ false. Thus, regularity can be expressed in the slogan: \[\mbox{Indeterminacy Implies Compatibility with Falsity.}\]
Thus, if $A$ is indeterminate at $x$, then $x$ is compatible with a $y$ that makes $A$ false.

Let us emphasize that regularity is a constraint on what counts as a \emph{proposition} in possibility semantics, not a constraint on what counts as a \emph{possibility}. Any pair of a nonempty set $S$ and compatibility relation on that set is a frame for possibility semantics, just as any pair of a set and accessibility relation is a frame for classical possible world semantics. But unlike in classical possible world semantics, not every subset of $S$ counts as a proposition. 
Restricting what subsets of a state space count as propositions has several precedents: again, from possible world semantics for intuitionistic logic; from standard approaches to probability, where a  $\sigma$-algebra over the state space---which need not be the full powerset algebra---is specified; as well as possible-world approaches to conditionals, where a set of admissible conditional antecedents is specified (as in \citealt{Stalnaker:1970}, where the conditional operator is defined only on sentences in the language, not on arbitrary propositions). Even in classical possible world semantics, what counts as a propositions is implicitly limited to those classes of possible worlds that are also sets (\citealt{Kaplan:1995}).

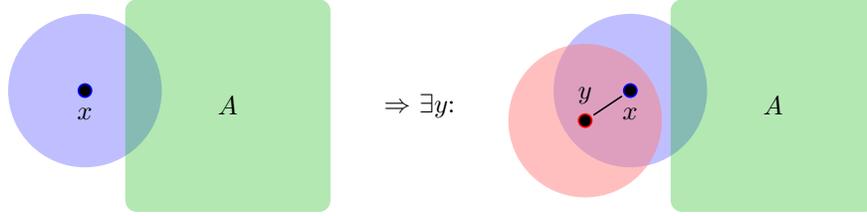
\begin{figure}[H]
\begin{center}
\tikzset{every loop/.style={min distance=10mm,looseness=10}}
\begin{tikzpicture}[->,>=stealth,shorten >=1pt,shorten <=1pt, auto,node
distance=2.5cm,semithick]
\tikzstyle{every state}=[fill=gray!20,draw=none,text=black]

\filldraw[color=blue!50, opacity = 0.5, very thick](-1.25,-.6) circle (1);

\node[circle,draw=blue!100, fill=black!100, label=below:$x$,inner sep=0pt,minimum size=.175cm] at (-1.25,-.6) {{}};

\path[-, draw=darkgreen, opacity=0.3, fill = darkgreen, thick, rounded corners] (2, .5) -- (2, -2.2) --  (.7, -2.2) -- (-.7, -2.2) -- (-.7, .6)  -- (2, .6) -- (2, .5) ; 

\node at (.65,-.8) {{$A$}};

\node at (3.2,-.8) {{$\Rightarrow $ $\exists y$:}};

\filldraw[color=blue!50, opacity = 0.5, very thick](6,-.6) circle (1);

\path[-, draw=darkgreen, opacity=0.3, fill = darkgreen, thick, rounded corners] (9.25, .5) -- (9.25, -2.2) --  (7.95, -2.2) -- (6.55, -2.2) -- (6.55, .6)  -- (9.25, .6) -- (9.25, .5) ; 

\node at (7.9,-.8) {{$A$}};

\filldraw[color=red!50, opacity = 0.5, very thick](5.4,-1) circle (1);

\node[circle,draw=blue!100, fill=black!100, label=below:$x$,inner sep=0pt,minimum size=.175cm] (a) at (6,-.6) {{}};

\node[circle,draw=red!100, fill=black!100, label=above:$y$,inner sep=0pt,minimum size=.175cm] (b) at (5.4,-1) {{}};
\path (a) edge[-] node {{}} (b);

\end{tikzpicture}
\end{center}
\caption{A depiction of the $\between$-regularity of $A$: if $x$ is not in $A$, then $x$ is compatible with a $y$ that is not compatible with any possibility in $A$. The blue disk represents the set of possibilities compatible with $x$; the red disk represents the set of possibilities compatible with $y$; and the green region represents $A$.}\label{RegDefFig}
\end{figure}

From the compatibility relation, it is useful to define a relation of \textit{refinement}. This can be done in two equivalent ways: $y$ refines $x$ if every proposition settled true by $x$ is also settled true by $y$, or equivalently, if every possibility compatible with $y$ is compatible with $x$.

\begin{lemma}\label{RefinementDef} For any compatibility frame $\langle S,\between\rangle$, the following are equivalent for any $x,y\in S$:
\begin{enumerate}
\item\label{RefinementDef1} for all $\between$-regular sets $A\subseteq S$, if $x\in A$, then $y\in A$;
\item\label{RefinementDef2} for all $z\in S$, if $z\between y$ then $z\between x$.
\end{enumerate}
When these conditions hold, we write $y\sqsubseteq x$.
\end{lemma}
\begin{proof} From \ref{RefinementDef1} to \ref{RefinementDef2}, suppose $z\between y$ but not $z\between x$. Let $A=\{v\in S\mid \mbox{not }z\between v\}$. We claim that $A$ is $\between$-regular. We must show that if $u\not\in A$, then $\exists u'\between u$ $\forall u''\between u'$  $u''\not\in A$, i.e., that if $z\between u$, then $\exists u'\between u$ $\forall u''\between u'$  $z\between u''$. Indeed, let $u'=z$ and use the symmetry of $\between$. Thus, $A$ is $\between$-regular, and $x\in A$ but $y\not\in A$.

From \ref{RefinementDef2} to \ref{RefinementDef1}, let $A$ be a $\between$-regular set. Suppose $x\in A$ and condition \ref{RefinementDef2} holds. We claim that for all $y'\between y$ there is a $y''\between y'$ with $y''\in A$, so $y\in A$ by the $\between$-regularity of $A$. Suppose $y'\between y$. Then since condition \ref{RefinementDef2} holds, we have $y'\between x$ and hence $x\between y'$ by symmetry of $\between$. Thus, we may take $y''=x$.
\end{proof}

For any possibility $x$, the set of refinements of $x$ is a $\between$-regular set and hence may be regarded as a proposition---the proposition \textit{that possibility $x$ obtains}.
\begin{lemma}\label{Downx} Given a compatibility frame $\mathcal{F}=\langle S,\between\rangle$ and $x\in S$, the set $\mathord{\downarrow}x=\{y\in S\mid y\sqsubseteq x\}$ is $\between$-regular.
\end{lemma}
\begin{proof} Suppose $y\not\in \mathord{\downarrow}x$. Hence there is a $y'\between y$ such that \textit{not} $y'\between x$. Now consider any $y''\between y'$. Since $y'\between y''$ but \textit{not} $y'\between x$, we have that $y''\not\sqsubseteq x$. Thus, we have shown that if $y\not\in \mathord{\downarrow}x$, then $\exists y'\between y$ $\forall y''\between y'$ $y''\not\in  \mathord{\downarrow}x$. Hence  $\mathord{\downarrow}x$ is $\between$-regular.\end{proof}

We can also use the notion of refinement to define \textit{possible worlds}: a world is a possibility that refines every possibility with which it is compatible.

\begin{definition}\label{WorldDef} Given a compatibility frame $\mathcal{F}=\langle S,\between\rangle$, a \textit{world in $\mathcal{F}$} is a $w\in S$ such that for all $x\in S$, $w\between x$ implies $w\sqsubseteq x$.
\end{definition}

Now in line with the idea that only $\between$-regular sets of possibilities are propositions, a \textit{model} should interpret the propositional variables of our formal language as $\between$-regular sets.

\begin{definition}\label{CompModDef}
A \textit{compatibility model} is a pair $\mathcal{M}=\langle \mathcal{F},V\rangle$ where $\mathcal{F}=\langle S,\between\rangle$ is a compatibility frame and $V$ is a function assigning to each $p\in\mathsf{Prop}$ a $\between$-regular set $V(p)\subseteq S$.  We say that $\mathcal{M}$ is \textit{based on} $\mathcal{F}$.
\end{definition}

We now observe how any compatibility frame gives rise to an ortholattice. Compare the following result, due to Birkhoff \citeyearpar[\S\S~32-4]{Birkhoff1940}, to the way in which for any set $W$, the family of all subsets ordered by the subset relation $\subseteq$ forms a complete Boolean algebra. 

\begin{proposition}\label{FrameToOrtho} Given any compatibility frame $\mathcal{F}=\langle S,\between\rangle$, the $\between$-regular sets ordered by the subset relation $\subseteq$ form a complete lattice, which becomes an ortholattice with the operation $\neg$ defined by 
\begin{equation}\neg A=\{x\in S\mid \forall y\between x\;\, y\not\in A\}.\label{NegEq}\end{equation}
For the lattice operations, we have $A\wedge B=A\cap B$, $A\vee B= \neg(\neg A\cap \neg B)$, $1=S$, and $0=\varnothing$.  We denote the ortholattice of $\between$-regular subsets of $\mathcal{F}$ by~$O(\mathcal{F})$.
\end{proposition}

\noindent Thus, $\wedge$ is intersection as in possible world semantics, but note how $\neg$ and $\vee$ are no longer interpreted as in possible world semantics. In possible world semantics, $\neg A=\{x\in W\mid x\not\in A\}$, in contrast to (\ref{NegEq}), and $\vee$ is union, in contrast to what we get when we unpack the definition $\vee$ as $A\vee B= \neg(\neg A\cap \neg B)$:
\begin{equation}A\vee B=\{x\in S\mid \forall y\between x\,\exists z\between y: z\in A\cup B\},\label{OrEq}\end{equation}
i.e., $x$ makes $A\vee B$ true just in case every possibility compatible with $x$ is in turn compatible with a possibility that makes one of $A$ or $B$ true.

\begin{remark}\label{BooleanCase} As observed in \citealt{Holliday2021} (Example 3.16), the ortholattice arising from a compatibility frame as in Proposition \ref{FrameToOrtho} is a Boolean algebra if for all $x,y\in S$, if $x\between y$, then  $x$ and $y$ have a common refinement, i.e., there is a $z$ such that $z\sqsubseteq x$ and $z\sqsubseteq y$. Indeed, this condition is not only sufficient but also necessary for the ortholattice to be Boolean: for if $x$ and $y$ have no common refinement, so $\mathord{\downarrow}x\cap\mathord{\downarrow}y=\varnothing$, then in a Boolean algebra we must have $\mathord{\downarrow}x \subseteq\neg \mathord{\downarrow}y$, contradicting $x\between y$. Thus, classicality corresponds to the condition (to use terminology introduced after Definition \ref{CompFrame}) that \textit{compatibility} implies \textit{compossibility}. 

Our departure from classicality can now be seen as follows: we want there to be distinct but \textit{compatible} possibilities that settle $A$ and $\Diamond\neg A$ as true, respectively, as neither $A$ nor $\Diamond\neg A$ should entail the negation of the other; yet such possibilities should not be \textit{compossible}---they should have no common refinement, since no single possibility should settle both $A$ and $\Diamond\neg A$ as true.\end{remark}

\begin{remark} It is noteworthy that two compatibility frames $(S,\between)$ and $(S,\between')$ can have the same derived refinement relation but give rise to non-isomorphic ortholattices; an example is provided in the second of the  notebooks cited in \S~\ref{Intro}. Thus, although the refinement relation has all the information one needs in a classical setting (wherein compatibility can be defined from a primitive partial order of refinement by: $x\between y$ if $x$ and $y$ have a common refinement), it does not in our non-classical setting here.
\end{remark}

\begin{example}\label{CompEx} Figure \ref{Fig2} shows a simple compatibility frame with five possibilities (above) and its refinement relation (below) defined from compatibility as in Lemma \ref{RefinementDef}.

\begin{figure}[h]
\begin{center}
\tikzset{every loop/.style={min distance=10mm,looseness=10}}
\begin{tikzpicture}[->,>=stealth',shorten >=1pt,shorten <=1pt, auto,node
distance=2.5cm,semithick]
\tikzstyle{every state}=[fill=gray!20,draw=none,text=black]

\node[circle,draw=black!100, fill=black!100, label=below:$x_1$,inner sep=0pt,minimum size=.175cm] (1) at (0,0) {{}};
\node[circle,draw=black!100, fill=black!100, label=below:$x_2$,inner sep=0pt,minimum size=.175cm] (2) at (1.5,0) {{}};
\node[circle,draw=black!100, fill=black!100, label=below:$x_3$,inner sep=0pt,minimum size=.175cm] (3) at (3,0) {{}};
\node[circle,draw=black!100, fill=black!100, label=below:$x_4$,inner sep=0pt,minimum size=.175cm] (4) at (4.5,0) {{}};
\node[circle,draw=black!100, fill=black!100, label=below:$x_5$,inner sep=0pt,minimum size=.175cm] (5) at (6,0) {{}};

\path (1) edge[-] node {{}} (2);
\path (2) edge[-] node {{}} (3);
\path (3) edge[-] node {{}} (4);
\path (4) edge[-] node {{}} (5);

\end{tikzpicture} \vspace{.5in}

\tikzset{every loop/.style={min distance=10mm,looseness=10}}
\begin{tikzpicture}[->,>=stealth',shorten >=1pt,shorten <=1pt, auto,node
distance=2.5cm,semithick]
\tikzstyle{every state}=[fill=gray!20,draw=none,text=black]

\node[circle,draw=black!100, fill=black!100, label=below:$x_1$,inner sep=0pt,minimum size=.175cm] (1) at (0,0) {{}};
\node[circle,draw=black!100, fill=black!100, label=below:$x_2$,inner sep=0pt,minimum size=.175cm] (2) at (1.5,0) {{}};
\node[circle,draw=black!100, fill=black!100, label=below:$x_3$,inner sep=0pt,minimum size=.175cm] (3) at (3,0) {{}};
\node[circle,draw=black!100, fill=black!100, label=below:$x_4$,inner sep=0pt,minimum size=.175cm] (4) at (4.5,0) {{}};
\node[circle,draw=black!100, fill=black!100, label=below:$x_5$,inner sep=0pt,minimum size=.175cm] (5) at (6,0) {{}};

\path (1) edge[<-,dashed] node {{}} (2);
\path (4) edge[->,dashed] node {{}} (5);

\end{tikzpicture} 

\end{center}\caption{A compatibility frame (above), where compatible possibilities are linked by an edge in the graph, and its refinement relation (below), where a dashed arrow from $y$ to $z$ indicates that $z$ is a refinement of $y$, i.e., $z\sqsubseteq y$. All reflexive compatibility loops and reflexive refinement loops are omitted.}\label{Fig2}
\end{figure}
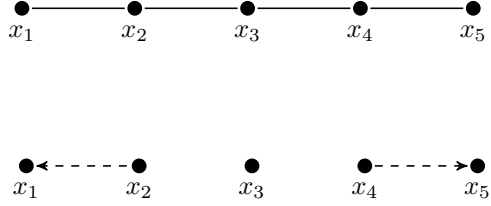

 Figure \ref{RegFig} shows the ten $\between$-regular subsets of the compatibility frame from Figure \ref{Fig2}, each highlighted in green. For example, $\{x_1,x_2,x_3\}$ is $\between$-regular according to Definition \ref{RegProp} because each possibility outside the set is compatible with $x_5$, which is not compatible with anything in the set; but $\{x_2,x_3,x_4\}$ is not $\between$-regular, because $x_1$ is outside the set and yet everything compatible with $x_1$ is compatible with something in the set. Note that ordering  the $\between$-regular subsets in Figure \ref{RegFig} by $\subseteq$ and defining $\neg$ as in Proposition \ref{FrameToOrtho} yields an ortholattice $O(\mathcal{F})$ isomorphic to the ortholattice in Figure \ref{Fig1} (reproduced on the right of Figure \ref{JIfig} below)!

\begin{figure}[h]
\begin{center}

\begin{tikzpicture}[->,>=stealth',shorten >=1pt,shorten <=1pt, auto,node
distance=2.5cm,semithick,scale=.5]
\tikzstyle{every state}=[fill=gray!20,draw=none,text=black]

\node[circle,draw=black!100, fill=green!100, label=below:$x_1$,inner sep=0pt,minimum size=.175cm] (1) at (0,0) {{}};
\node[circle,draw=black!100, fill=green!100, label=below:$x_2$,inner sep=0pt,minimum size=.175cm] (2) at (1.5,0) {{}};
\node[circle,draw=black!100, fill=green!100, label=below:$x_3$,inner sep=0pt,minimum size=.175cm] (3) at (3,0) {{}};
\node[circle,draw=black!100, fill=green!100, label=below:$x_4$,inner sep=0pt,minimum size=.175cm] (4) at (4.5,0) {{}};
\node[circle,draw=black!100, fill=green!100, label=below:$x_5$,inner sep=0pt,minimum size=.175cm] (5) at (6,0) {{}};

\path (1) edge[-] node {{}} (2);
\path (2) edge[-] node {{}} (3);
\path (3) edge[-] node {{}} (4);
\path (4) edge[-] node {{}} (5);

\end{tikzpicture}\vspace{.1in}

\begin{tikzpicture}[->,>=stealth',shorten >=1pt,shorten <=1pt, auto,node
distance=2.5cm,semithick,scale=.5]
\tikzstyle{every state}=[fill=gray!20,draw=none,text=black]

\node[circle,draw=black!100, fill=green!100, label=below:$x_1$,inner sep=0pt,minimum size=.175cm] (1) at (0,0) {{}};
\node[circle,draw=black!100, fill=green!100, label=below:$x_2$,inner sep=0pt,minimum size=.175cm] (2) at (1.5,0) {{}};
\node[circle,draw=black!100, fill=green!100, label=below:$x_3$,inner sep=0pt,minimum size=.175cm] (3) at (3,0) {{}};
\node[circle,draw=black!100, fill=black!100, label=below:$x_4$,inner sep=0pt,minimum size=.175cm] (4) at (4.5,0) {{}};
\node[circle,draw=black!100, fill=black!100, label=below:$x_5$,inner sep=0pt,minimum size=.175cm] (5) at (6,0) {{}};

\path (1) edge[-] node {{}} (2);
\path (2) edge[-] node {{}} (3);
\path (3) edge[-] node {{}} (4);
\path (4) edge[-] node {{}} (5);

\end{tikzpicture}\qquad\qquad\qquad\qquad\qquad \begin{tikzpicture}[->,>=stealth',shorten >=1pt,shorten <=1pt, auto,node
distance=2.5cm,semithick,scale=.5]
\tikzstyle{every state}=[fill=gray!20,draw=none,text=black]

\node[circle,draw=black!100, fill=black!100, label=below:$x_1$,inner sep=0pt,minimum size=.175cm] (1) at (0,0) {{}};
\node[circle,draw=black!100, fill=black!100, label=below:$x_2$,inner sep=0pt,minimum size=.175cm] (2) at (1.5,0) {{}};
\node[circle,draw=black!100, fill=green!100, label=below:$x_3$,inner sep=0pt,minimum size=.175cm] (3) at (3,0) {{}};
\node[circle,draw=black!100, fill=green!100, label=below:$x_4$,inner sep=0pt,minimum size=.175cm] (4) at (4.5,0) {{}};
\node[circle,draw=black!100, fill=green!100, label=below:$x_5$,inner sep=0pt,minimum size=.175cm] (5) at (6,0) {{}};

\path (1) edge[-] node {{}} (2);
\path (2) edge[-] node {{}} (3);
\path (3) edge[-] node {{}} (4);
\path (4) edge[-] node {{}} (5);

\end{tikzpicture} \vspace{.1in}

\begin{tikzpicture}[->,>=stealth',shorten >=1pt,shorten <=1pt, auto,node
distance=2.5cm,semithick,scale=.5]
\tikzstyle{every state}=[fill=gray!20,draw=none,text=black]

\node[circle,draw=black!100, fill=green!100, label=below:$x_1$,inner sep=0pt,minimum size=.175cm] (1) at (0,0) {{}};
\node[circle,draw=black!100, fill=black!100, label=below:$x_2$,inner sep=0pt,minimum size=.175cm] (2) at (1.5,0) {{}};
\node[circle,draw=black!100, fill=black!100, label=below:$x_3$,inner sep=0pt,minimum size=.175cm] (3) at (3,0) {{}};
\node[circle,draw=black!100, fill=black!100, label=below:$x_4$,inner sep=0pt,minimum size=.175cm] (4) at (4.5,0) {{}};
\node[circle,draw=black!100, fill=green!100, label=below:$x_5$,inner sep=0pt,minimum size=.175cm] (5) at (6,0) {{}};

\path (1) edge[-] node {{}} (2);
\path (2) edge[-] node {{}} (3);
\path (3) edge[-] node {{}} (4);
\path (4) edge[-] node {{}} (5);

\end{tikzpicture}

\vspace{.1in}

\begin{tikzpicture}[->,>=stealth',shorten >=1pt,shorten <=1pt, auto,node
distance=2.5cm,semithick,scale=.5]
\tikzstyle{every state}=[fill=gray!20,draw=none,text=black]

\node[circle,draw=black!100, fill=green!100, label=below:$x_1$,inner sep=0pt,minimum size=.175cm] (1) at (0,0) {{}};
\node[circle,draw=black!100, fill=green!100, label=below:$x_2$,inner sep=0pt,minimum size=.175cm] (2) at (1.5,0) {{}};
\node[circle,draw=black!100, fill=black!100, label=below:$x_3$,inner sep=0pt,minimum size=.175cm] (3) at (3,0) {{}};
\node[circle,draw=black!100, fill=black!100, label=below:$x_4$,inner sep=0pt,minimum size=.175cm] (4) at (4.5,0) {{}};
\node[circle,draw=black!100, fill=black!100, label=below:$x_5$,inner sep=0pt,minimum size=.175cm] (5) at (6,0) {{}};

\path (1) edge[-] node {{}} (2);
\path (2) edge[-] node {{}} (3);
\path (3) edge[-] node {{}} (4);
\path (4) edge[-] node {{}} (5);

\end{tikzpicture}\qquad\qquad\qquad\qquad\qquad  \begin{tikzpicture}[->,>=stealth',shorten >=1pt,shorten <=1pt, auto,node
distance=2.5cm,semithick,scale=.5]
\tikzstyle{every state}=[fill=gray!20,draw=none,text=black]

\node[circle,draw=black!100, fill=black!100, label=below:$x_1$,inner sep=0pt,minimum size=.175cm] (1) at (0,0) {{}};
\node[circle,draw=black!100, fill=black!100, label=below:$x_2$,inner sep=0pt,minimum size=.175cm] (2) at (1.5,0) {{}};
\node[circle,draw=black!100, fill=black!100, label=below:$x_3$,inner sep=0pt,minimum size=.175cm] (3) at (3,0) {{}};
\node[circle,draw=black!100, fill=green!100, label=below:$x_4$,inner sep=0pt,minimum size=.175cm] (4) at (4.5,0) {{}};
\node[circle,draw=black!100, fill=green!100, label=below:$x_5$,inner sep=0pt,minimum size=.175cm] (5) at (6,0) {{}};

\path (1) edge[-] node {{}} (2);
\path (2) edge[-] node {{}} (3);
\path (3) edge[-] node {{}} (4);
\path (4) edge[-] node {{}} (5);

\end{tikzpicture}\vspace{.1in}

 \begin{tikzpicture}[->,>=stealth',shorten >=1pt,shorten <=1pt, auto,node
distance=2.5cm,semithick,scale=.5]
\tikzstyle{every state}=[fill=gray!20,draw=none,text=black]

\node[circle,draw=black!100, fill=black!100, label=below:$x_1$,inner sep=0pt,minimum size=.175cm] (1) at (0,0) {{}};
\node[circle,draw=black!100, fill=black!100, label=below:$x_2$,inner sep=0pt,minimum size=.175cm] (2) at (1.5,0) {{}};
\node[circle,draw=black!100, fill=green!100, label=below:$x_3$,inner sep=0pt,minimum size=.175cm] (3) at (3,0) {{}};
\node[circle,draw=black!100, fill=black!100, label=below:$x_4$,inner sep=0pt,minimum size=.175cm] (4) at (4.5,0) {{}};
\node[circle,draw=black!100, fill=black!100, label=below:$x_5$,inner sep=0pt,minimum size=.175cm] (5) at (6,0) {{}};

\path (1) edge[-] node {{}} (2);
\path (2) edge[-] node {{}} (3);
\path (3) edge[-] node {{}} (4);
\path (4) edge[-] node {{}} (5);

\end{tikzpicture}

\vspace{.1in}

\begin{tikzpicture}[->,>=stealth',shorten >=1pt,shorten <=1pt, auto,node
distance=2.5cm,semithick,scale=.5]
\tikzstyle{every state}=[fill=gray!20,draw=none,text=black]

\node[circle,draw=black!100, fill=green!100, label=below:$x_1$,inner sep=0pt,minimum size=.175cm] (1) at (0,0) {{}};
\node[circle,draw=black!100, fill=black!100, label=below:$x_2$,inner sep=0pt,minimum size=.175cm] (2) at (1.5,0) {{}};
\node[circle,draw=black!100, fill=black!100, label=below:$x_3$,inner sep=0pt,minimum size=.175cm] (3) at (3,0) {{}};
\node[circle,draw=black!100, fill=black!100, label=below:$x_4$,inner sep=0pt,minimum size=.175cm] (4) at (4.5,0) {{}};
\node[circle,draw=black!100, fill=black!100, label=below:$x_5$,inner sep=0pt,minimum size=.175cm] (5) at (6,0) {{}};

\path (1) edge[-] node {{}} (2);
\path (2) edge[-] node {{}} (3);
\path (3) edge[-] node {{}} (4);
\path (4) edge[-] node {{}} (5);

\end{tikzpicture}\qquad\qquad\qquad\qquad\qquad \begin{tikzpicture}[->,>=stealth',shorten >=1pt,shorten <=1pt, auto,node
distance=2.5cm,semithick,scale=.5]
\tikzstyle{every state}=[fill=gray!20,draw=none,text=black]

\node[circle,draw=black!100, fill=black!100, label=below:$x_1$,inner sep=0pt,minimum size=.175cm] (1) at (0,0) {{}};
\node[circle,draw=black!100, fill=black!100, label=below:$x_2$,inner sep=0pt,minimum size=.175cm] (2) at (1.5,0) {{}};
\node[circle,draw=black!100, fill=black!100, label=below:$x_3$,inner sep=0pt,minimum size=.175cm] (3) at (3,0) {{}};
\node[circle,draw=black!100, fill=black!100, label=below:$x_4$,inner sep=0pt,minimum size=.175cm] (4) at (4.5,0) {{}};
\node[circle,draw=black!100, fill=green!100, label=below:$x_5$,inner sep=0pt,minimum size=.175cm] (5) at (6,0) {{}};

\path (1) edge[-] node {{}} (2);
\path (2) edge[-] node {{}} (3);
\path (3) edge[-] node {{}} (4);
\path (4) edge[-] node {{}} (5);

\end{tikzpicture}\vspace{.1in}

\begin{tikzpicture}[->,>=stealth',shorten >=1pt,shorten <=1pt, auto,node
distance=2.5cm,semithick,scale=.5]
\tikzstyle{every state}=[fill=gray!20,draw=none,text=black]

\node[circle,draw=black!100, fill=black!100, label=below:$x_1$,inner sep=0pt,minimum size=.175cm] (1) at (0,0) {{}};
\node[circle,draw=black!100, fill=black!100, label=below:$x_2$,inner sep=0pt,minimum size=.175cm] (2) at (1.5,0) {{}};
\node[circle,draw=black!100, fill=black!100, label=below:$x_3$,inner sep=0pt,minimum size=.175cm] (3) at (3,0) {{}};
\node[circle,draw=black!100, fill=black!100, label=below:$x_4$,inner sep=0pt,minimum size=.175cm] (4) at (4.5,0) {{}};
\node[circle,draw=black!100, fill=black!100, label=below:$x_5$,inner sep=0pt,minimum size=.175cm] (5) at (6,0) {{}};

\path (1) edge[-] node {{}} (2);
\path (2) edge[-] node {{}} (3);
\path (3) edge[-] node {{}} (4);
\path (4) edge[-] node {{}} (5);

\end{tikzpicture}

\end{center}
\caption{The ten $\between$-regular subsets of the compatibility frame from Figure \ref{Fig2}, each highlighted in green}\label{RegFig}
\end{figure}
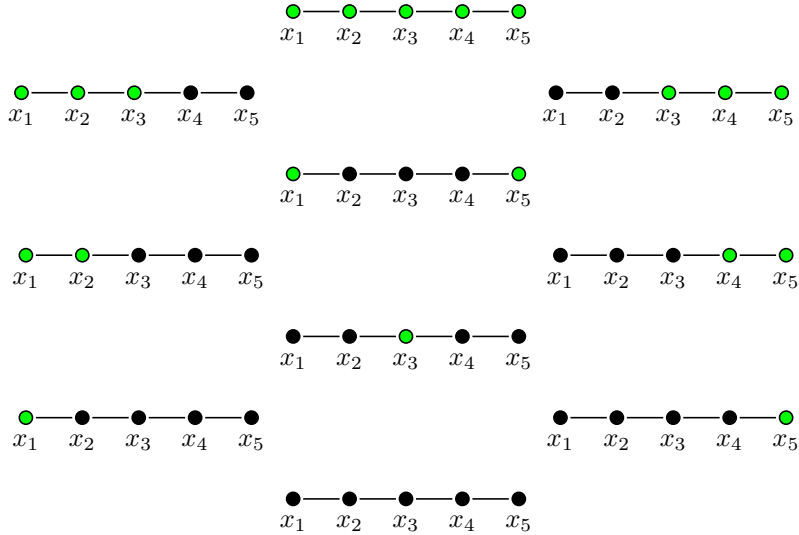

\end{example}

\begin{example}\label{TwoOrthoRep}
The two ortholattices in Figure \ref{OrthoEx} can be realized as the ortholattices of $\between$-regular subsets of the two compatibility frames in Figure \ref{CompFig}. Verifying this claim is a good check for understanding.
\end{example}

\begin{figure}[h]
\begin{center}
\tikzset{every loop/.style={min distance=10mm,looseness=10}}
\begin{tikzpicture}[->,>=stealth',shorten >=1pt,shorten <=1pt, auto,node
distance=2.5cm,semithick]
\tikzstyle{every state}=[fill=gray!20,draw=none,text=black]

\node[circle,draw=black!100, fill=black!100, label=below:$x_1$,inner sep=0pt,minimum size=.175cm] (1) at (0,0) {{}};
\node[circle,draw=black!100, fill=black!100, label=below:$x_2$,inner sep=0pt,minimum size=.175cm] (2) at (1.5,0) {{}};
\node[circle,draw=black!100, fill=black!100, label=below:$x_3$,inner sep=0pt,minimum size=.175cm] (3) at (3,0) {{}};
\node[circle,draw=black!100, fill=black!100, label=below:$x_4$,inner sep=0pt,minimum size=.175cm] (4) at (4.5,0) {{}};

\path (1) edge[-] node {{}} (2);
\path (2) edge[-] node {{}} (3);
\path (3) edge[-] node {{}} (4);

\end{tikzpicture} \qquad\qquad \begin{tikzpicture}[->,>=stealth',shorten >=1pt,shorten <=1pt, auto,node
distance=2.5cm,semithick]
\tikzstyle{every state}=[fill=gray!20,draw=none,text=black]

\node[circle,draw=black!100, fill=black!100, label=above:$x_1$,inner sep=0pt,minimum size=.175cm] (1) at (0,0) {{}};
\node[circle,draw=black!100, fill=black!100, label=above:$x_2$,inner sep=0pt,minimum size=.175cm] (2) at (1.5,0) {{}};
\node[circle,draw=black!100, fill=black!100, label=below:$x_3$,inner sep=0pt,minimum size=.175cm] (3) at (1.5,-1.5) {{}};
\node[circle,draw=black!100, fill=black!100, label=below:$x_4$,inner sep=0pt,minimum size=.175cm] (4) at (0,-1.5) {{}};

\path (1) edge[-] node {{}} (2);
\path (2) edge[-] node {{}} (3);
\path (3) edge[-] node {{}} (4);
\path (1) edge[-] node {{}} (4);
\end{tikzpicture} 

\end{center}\caption{Compatibility frames that give rise to the ortholattices in Figure \ref{OrthoEx} (reproduced on the left and middle of Figure \ref{JIfig})}\label{CompFig}
\end{figure}
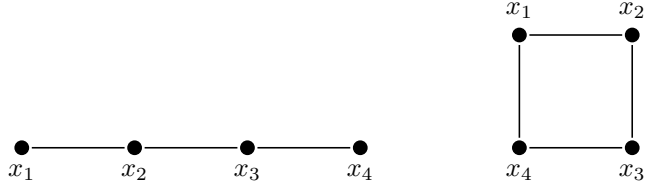

From Examples \ref{CompEx} and \ref{TwoOrthoRep}, one might notice the key to representing any finite ortholattice using a compatibility frame: the possibilities in the frame correspond to the  \textit{join-irreducible} elements of the ortholattices, i.e., those nonzero elements $a$ of the ortholattice that cannot be obtained as a join of elements distinct from $a$---intuitively, those noncontradictory propositions that cannot be expressed as a disjunction of distinct propositions; then two possibilities $a$ and $b$ are compatible if $a\not\leq \neg b$ in the ortholattice. Figure \ref{JIfig} highlights in orange the join-irreducible elements of the three ortholattices  considered so far, elements which one can match up with the possibilities in the corresponding compatibility frames in Figures \ref{CompFig} and \ref{Fig2}.

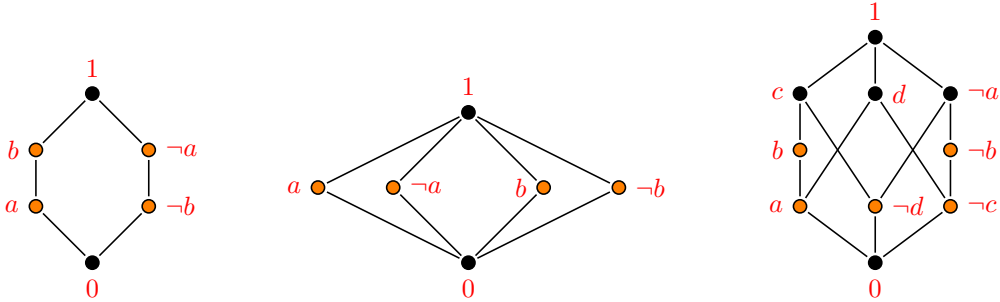
\begin{figure}[h]
\begin{center}
\tikzset{every loop/.style={min distance=10mm,looseness=10}}
\begin{tikzpicture}[->,>=stealth',shorten >=1pt,shorten <=1pt, auto,node
distance=2cm,semithick,every loop/.style={<-,shorten <=1pt}]
\tikzstyle{every state}=[fill=gray!20,draw=none,text=black]
\node[circle,draw=black!100,fill=black!100, label=below:$\textcolor{red}{0}$,inner sep=0pt,minimum size=.175cm] (0) at (0,0) {{}};
\node[circle,draw=black!100,fill=orange!100, label=left:$\textcolor{red}{a}$,inner sep=0pt,minimum size=.175cm] (a) at (-.75,.75) {{}};
\node[circle,draw=black!100,fill=orange!100, label=right:$\textcolor{red}{\neg b}$,inner sep=0pt,minimum size=.175cm] (b) at (.75,.75) {{}};
\node[circle,draw=black!100,fill=orange!100, label=left:$\textcolor{red}{b}$,inner sep=0pt,minimum size=.175cm] (1l) at (-.75,1.5) {{}};
\node[circle,draw=black!100,fill=orange!100, label=right:$\textcolor{red}{\neg a}$,inner sep=0pt,minimum size=.175cm] (1r) at (.75,1.5) {{}};
\node[circle,draw=black!100,fill=black!100, label=above:$\textcolor{red}{1}$,inner sep=0pt,minimum size=.175cm] (new1) at (0,2.25) {{}};

\path (new1) edge[-] node {{}} (1l);
\path (new1) edge[-] node {{}} (1r);
\path (1l) edge[-] node {{}} (a);
\path (1r) edge[-] node {{}} (b);
\path (a) edge[-] node {{}} (0);
\path (b) edge[-] node {{}} (0);

\node[circle,draw=black!100,fill=black!100, label=above:$\textcolor{red}{1}$,inner sep=0pt,minimum size=.175cm] (1) at (5,2) {{}};
\node[circle,draw=black!100,fill=orange!100, label=left:$\textcolor{red}{a}$,inner sep=0pt,minimum size=.175cm] (x) at (3,1) {{}};
\node[circle,draw=black!100,fill=orange!100, label=right:$\textcolor{red}{\neg a}$,inner sep=0pt,minimum size=.175cm] (y) at (4,1) {{}};
\node[circle,draw=black!100,fill=orange!100, label=left:$\textcolor{red}{b}$,inner sep=0pt,minimum size=.175cm] (y') at (6,1) {{}};
\node[circle,draw=black!100,fill=orange!100, label=right:$\textcolor{red}{\neg b}$,inner sep=0pt,minimum size=.175cm] (z) at (7,1) {{}};
\node[circle,draw=black!100,fill=black!100, label=below:$\textcolor{red}{0}$,inner sep=0pt,minimum size=.175cm] (0) at (5,0) {{}};
\path (1) edge[-] node {{}} (y);
\path (1) edge[-] node {{}} (y');
\path (1) edge[-] node {{}} (x);
\path (1) edge[-] node {{}} (z);
\path (x) edge[-] node {{}} (0);
\path (y) edge[-] node {{}} (0);
\path (y') edge[-] node {{}} (0);
\path (z) edge[-] node {{}} (0);

\end{tikzpicture}\qquad\quad\;\begin{tikzpicture}[->,>=stealth',shorten >=1pt,shorten <=1pt, auto,node
distance=2.5cm,semithick]
\tikzstyle{every state}=[fill=gray!20,draw=none,text=black]
\node[circle,draw=black!100, fill=black!100, label=below:$\textcolor{red}{0}$,inner sep=0pt,minimum size=.175cm] (0) at (0,0) {{}};
\node[circle,draw=black!100, fill=orange!100, label=right:$\textcolor{red}{\neg d}$,inner sep=0pt,minimum size=.175cm] (d) at (0,.75) {{}};

\node[circle,draw=black!100, fill=orange!100, label=left:$\textcolor{red}{a}$,inner sep=0pt,minimum size=.175cm] (a) at (-1,.75) {{}};
\node[circle,draw=black!100, fill=orange!100, label=right:$\textcolor{red}{\neg c}$,inner sep=0pt,minimum size=.175cm] (Nc) at (1,.75) {{}};

\node[circle,draw=black!100, fill=orange!100, label=left:$\textcolor{red}{b}$,inner sep=0pt,minimum size=.175cm] (b) at (-1,1.5) {{}};
\node[circle,draw=black!100, fill=orange!100, label=right:$\textcolor{red}{\neg b}$,inner sep=0pt,minimum size=.175cm] (Nb) at (1,1.5) {{}};
\node[circle,draw=black!100, fill=black!100,label=left:$\textcolor{red}{c}$,inner sep=0pt,minimum size=.175cm] (c) at (-1,2.25) {{}};

\node[circle,draw=black!100, fill=black!100, label=right:$\textcolor{red}{\neg a}$,inner sep=0pt,minimum size=.175cm] (Na) at (1,2.25) {{}};

\node[circle,draw=black!100, fill=black!100, label=right:$\textcolor{red}{d}\;\;$,inner sep=0pt,minimum size=.175cm] (Nd) at (0,2.25) {{}};
\node[circle,draw=black!100, fill=black!100,label=above:$\textcolor{red}{1}$,inner sep=0pt,minimum size=.175cm] (1) at (0,3) {{}};

\path (d) edge[-] node {{}} (c);
\path (d) edge[-] node {{}} (Na);
\path (d) edge[-] node {{}} (0);

\path (1) edge[-] node {{}} (Nd);
\path (Nd) edge[-] node {{}} (a);
\path (Nd) edge[-] node {{}} (Nc);
\path (1) edge[-] node {{}} (c);
\path (1) edge[-] node {{}} (Na);
\path (a) edge[-] node {{}} (b);
\path (b) edge[-] node {{}} (c);
\path (Nc) edge[-] node {{}} (Nb);
\path (Nb) edge[-] node {{}} (Na);
\path (a) edge[-] node {{}} (0);
\path (Nc) edge[-] node {{}} (0);

\end{tikzpicture}
\end{center}\caption{Hasse diagrams of ortholattices from Figures \ref{OrthoEx} and \ref{Fig1} with join-irreducible elements in orange}\label{JIfig}
\end{figure}

The fact that finite ortholattices can be represented using compatibility frames as described above is a consequence of the following more general representation theorem for all complete ortholattices.\footnote{For the representation of arbitrary (including incomplete) ortholattices using compatibility (or incompatibility) frames equipped with a topology, see \citealt{Goldblatt1975} and \citealt{McDonald2021}, and for associated categorical dualities, see \citealt{Bimbo2007}, \citealt{Dmitrieva2021}, and \citealt{McDonald2021}.} A set $V$ of elements of a lattice $L$ is said to be \textit{join-dense in $L$} if every element of $L$ can be obtained as the (possibly infinite) join of some elements of $V$ (e.g., the set of all elements of $L$ is trivially join-dense in $L$).
 
\begin{theorem}\label{MacLaren} Let $L=\langle A,\vee,0,\wedge,1,\neg\rangle $ be an ortholattice and $V$ any join-dense subset of $L$. Then where $\mathcal{F}=\langle V\setminus\{0\},\between\rangle$ is the compatibility frame based on $V\setminus\{0\}$ with $\between$ defined by $a\between b$ if $a\not\leq \neg b$, we have that $L$ embeds into $O(\mathcal{F})$ via the map $a\mapsto \{b\in V\setminus\{0\}\mid b\leq a\}$, which is an isomorphism if $L$ is complete.
\end{theorem}
\noindent For a proof, see, e.g., \citealt[Theorems 2.3 and 2.5]{MacLaren1964}. Compare this representation of complete ortholattices to Tarski's \citeyearpar{Tarski1935} representation of \textit{complete and atomic Boolean algebras}, which lies at the foundation of possible world semantics: any such Boolean algebra is isomorphic to the powerset of its set of atoms (think possible worlds).\footnote{Recall that an \textit{atom} in a bounded lattice is a nonzero element $a$ such that $b\leq a$ implies $b=0$ or $b=a$. A lattice $L$ is \textit{atomic} if for every nonzero element $b$ in $L$, there is an atom $a\leq b$.} Now in a \textit{finite} ortholattice, the set of join-irreducible elements is join-dense in the lattice, so we obtain the following corollary of Theorem \ref{MacLaren}, which explains Examples \ref{CompEx} and \ref{TwoOrthoRep}.

\begin{corollary}\label{FiniteRep} Let $L=\langle A,\vee,0,\wedge,1,\neg\rangle $ be a finite ortholattice. Then where $\mathcal{F}=\langle J,\between\rangle$ is the compatibility frame based on the set $J$ of join-irreducible elements of $L$ with $\between$ defined by $a\between b$ if $a\not\leq \neg b$, we have that $O(\mathcal{F})$ is isomorphic to $L$.
\end{corollary}

Let us now turn to formal semantics for the language $\mathcal{L}$ (recall Definition \ref{LDef}). Proposition \ref{FrameToOrtho} leads us to the following.

\begin{definition}\label{GoldblattSemantics} Given a compatibility model $\mathcal{M}= \langle S,\between,V\rangle$, $x\in S$, and $\varphi\in\mathcal{L}$, we define $\mathcal{M},x\Vdash \varphi$ as follows:
\begin{enumerate}
\item $\mathcal{M},x\Vdash\top$;
\item $\mathcal{M},x\Vdash p$ iff $x\in V(p)$;
\item $\mathcal{M},x\Vdash \varphi\wedge\psi$ iff $\mathcal{M},x\Vdash \varphi$ and $\mathcal{M},x\Vdash \psi$;
\item $\mathcal{M},x\Vdash \neg\varphi$ iff for all $y\between x$, $\mathcal{M},y\nVdash \varphi$.
\end{enumerate}
We define $\llbracket \varphi\rrbracket^\mathcal{M}=\{x\in S\mid \mathcal{M},x\Vdash\varphi\}$.
\end{definition}

\noindent Then given our definition of $\varphi\vee\psi$ as $\varphi\vee\psi:=\neg(\neg\varphi\wedge\neg\psi)$, as in (\ref{OrEq}) we have:
\begin{itemize}
\item $\mathcal{M},x\Vdash \varphi\vee\psi$ iff for all $y\between x$ there is a $z\between y$ such that $\mathcal{M},z\Vdash \varphi$ or $\mathcal{M},z\Vdash\psi$.
\end{itemize}

An easy induction shows that the set of possibilities that make a formula true is indeed a proposition. 

\begin{lemma} For any compatibility model $\mathcal{M}$ and $\varphi\in\mathcal{L}$, $\llbracket \varphi\rrbracket^\mathcal{M}$ is a $\between$-regular set.\end{lemma}

\begin{example} Consider a compatibility model based on the frame in Figure \ref{Fig2} with $V(p)=\{x_1,x_2\}$, so $V(p)$ is a $\between$-regular set, as shown in Figure \ref{CompModFig}. Then observe that $\llbracket\neg p\rrbracket^\mathcal{M}=\{x_4,x_5\}$, another $\between$-regular set.
\begin{figure}[h]
\begin{center}
\tikzset{every loop/.style={min distance=10mm,looseness=10}}
\begin{tikzpicture}[->,>=stealth',shorten >=1pt,shorten <=1pt, auto,node
distance=2.5cm,semithick]
\tikzstyle{every state}=[fill=gray!20,draw=none,text=black]

\node[circle,draw=black!100, fill=green!100, label=below:$x_1$,inner sep=0pt,minimum size=.175cm] (1) at (0,0) {{}};
\node at (0,-.75) {{$p$}};

\node[circle,draw=black!100, fill=green!100, label=below:$x_2$,inner sep=0pt,minimum size=.175cm] (2) at (1.5,0) {{}};
\node at (1.5,-.75) {{$p$}};

\node[circle,draw=black!100, fill=black!100, label=below:$x_3$,inner sep=0pt,minimum size=.175cm] (3) at (3,0) {{}};
\node[circle,draw=black!100, fill=black!100, label=below:$x_4$,inner sep=0pt,minimum size=.175cm] (4) at (4.5,0) {{}};

\node[circle,draw=black!100, fill=black!100, label=below:$x_5$,inner sep=0pt,minimum size=.175cm] (5) at (6,0) {{}};

\path (1) edge[-] node {{}} (2);
\path (2) edge[-] node {{}} (3);
\path (3) edge[-] node {{}} (4);
\path (4) edge[-] node {{}} (5);

\end{tikzpicture} 

\end{center}\caption{A compatibility model based on the frame from Figure \ref{Fig2} with $V(p)$ highlighted in green}\label{CompModFig}
\end{figure}
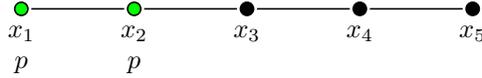
\end{example}

We define semantic consequence as usual in terms of truth preservation.

\begin{definition}\label{PossCons} Given a class $\mathbf{C}$ of compatibility frames, define the semantic consequence relation $\vDash_\mathbf{C}$, a binary relation on $\mathcal{L}$, as follows: $\varphi\vDash_\mathbf{C}\psi$ if for every $\mathcal{F}\in\mathbf{C}$, model $\mathcal{M}$ based on $\mathcal{F}$, and possibility $x$ in $\mathcal{M}$, if $\mathcal{M},x\Vdash\varphi$, then $\mathcal{M},x\Vdash\psi$.
\end{definition}
\noindent We can of course also define a consequence relation between a set of premises on the left and a single conclusion on the right: $\Gamma\vDash_\mathbf{C}\psi$ if for every $\mathcal{F}\in\mathbf{C}$, model $\mathcal{M}$ based on $\mathcal{F}$, and possibility $x$ in $\mathcal{M}$, if $\mathcal{M},x\Vdash \varphi$ for all $\varphi\in\Gamma$, then $\mathcal{M},x\Vdash\psi$. But for simplicity we only consider finite sets of premises here, in which case a single premise on the left suffices given that $\mathcal{L}$ contains a conjunction interpreted as intersection.

Goldblatt \citeyearpar{Goldblatt1974} proved the following completeness theorem, showing that the possibility semantics above is indeed a semantics for the orthologic $\mathsf{O}$ of Definition \ref{OrthoLogicDef}.

\begin{theorem}\label{Ocomp} The logic $\mathsf{O}$ is sound and complete with respect to the class $\mathbf{CF}$ of all compatibility frames according to the consequence relation in Definition \ref{PossCons}: for all $\varphi,\psi\in\mathcal{L}$, we have $\varphi\vdash_\mathsf{O}\psi$ if and only if $\varphi\vDash_\mathbf{CF}\psi$.
\end{theorem}

Theorem \ref{Ocomp} is easily proved from the embedding result in Theorem \ref{MacLaren}. In \S~\ref{AddEp} we will give a proof of a modal version of this completeness theorem that will show another way to prove Theorem \ref{Ocomp}.

\subsection{Adding epistemic modality}\label{AddEp}

In this section, we extend the compatibility frames of \S~\ref{ReviewPoss} to interpret modalities $\Box$ and $\Diamond$. We do so in the style of \textit{relational} possibility semantics (\citealt{Humberstone1981}, \citealt{HollidayForthA,HollidayForthB,Holliday2022}); as we note below, we could equally work with a  \textit{functional possibility semantics} (\citealt{Holliday2014}) instead. 

As usual, given a binary relation $R$ on a set $S$ of possibilities, we define a $\Box$ operation on propositions by
\begin{equation}\Box A = \{x\in S\mid R(x)\subseteq A\},\label{BoxEq1}\end{equation}
where $R(x)=\{y\in S\mid xRy\}$.  Given our definition of $\Diamond A$ as $\neg\Box\neg A$, we have 
\begin{equation}\Diamond A=\{x\in S\mid \forall x' \between x\; \exists y'\in R(x')\;\exists y''\between y': y''\in A\}.\label{DiamondEq}\end{equation}
Finally, we posit one condition on the relation between \textit{epistemic accessibility} and what \textit{might} be the case:
\begin{itemize}
\item $R$-\textsf{regularity}: if $x$ can epistemically access a possibility compatible with $y$, then $x$ is compatible with a possibility according to which $y$ \textit{might obtain}. 
\end{itemize}
Formally, the antecedent means that $xRy'\between y$, and the consequent means that $x$ is compatible with a possibility $x'\in \Diamond \mathord{\downarrow}y$ (recall Lemma \ref{Downx}).

We can now define the basic frames for our modal semantics.

\begin{definition}\label{DCompFrame}  A \textit{modal compatibility frame} is a triple $\mathcal{F}=\langle S,\between,R\rangle$ where $\langle S,\between\rangle$ is a compatibility frame and $R$ is a binary relation on $S$ satisfying
\begin{itemize}
\item $R$-\textsf{regularity}: if $xRy'\between y$, then $\exists x'\between x$: $x'\in\Diamond \mathord{\downarrow}y$. 
\end{itemize}
\noindent A \textit{modal compatibility model} is a pair $\mathcal{M}=\langle \mathcal{F},V\rangle$ where $\mathcal{F}=\langle S,\between, R\rangle$ is a modal compatibility frame and $\langle S,\between, V\rangle$ is a compatibility model as in Definition \ref{CompModDef}. We say that $\mathcal{M}$ is \textit{based on $\mathcal{F}$}.
\end{definition}

\begin{lemma}\label{RregEquiv} $R$-\textsf{regularity} can be written equivalently without $\Diamond$ as follows: 
\begin{itemize}
\item $R$-\textsf{regularity}: if $xRy'\between y$, then $\exists x'\between x$ $\forall x''\between x'$ $\exists y''$: $x''Ry''\between y$. 
\end{itemize}
\end{lemma}
\begin{proof} By (\ref{DiamondEq}), $x'\in \Diamond \mathord{\downarrow}y$ in $R$-\textsf{regularity} is equivalent to: $\forall x''\between x'$ $\exists y''\in R(x'')$ $\exists y'''\between y''$: $y'''\in \mathord{\downarrow}y$. But $\exists y'''\between y''$: $y'''\in \mathord{\downarrow}y$ is equivalent to $y''\between y$; from left-to-right, if $y'''\sqsubseteq y$, then $y''\between y'''$ implies $y''\between y$, and from right-to-left, take $y'''=y$. Thus, $x'\in \Diamond \mathord{\downarrow}y$ is equivalent to the condition on $x'$ in the lemma. \end{proof} 

We will begin by proving using $R$-\textsf{regularity} that if $A$ is a proposition, so is $\Box A$.

\begin{proposition}\label{BoxReg} For any modal compatibility frame $\mathcal{F}=\langle S,\between,R\rangle$ and $\between$-regular set $A\subseteq S$, we have that $\Box A$ is $\between$-regular.
\end{proposition}
\begin{proof} We must show that 
\[ x\not\in \Box A \Rightarrow  \exists x'\between x\,\forall x''\between x'\;\, x''\not\in \Box A.\]
Suppose $x\not\in \Box A$, so for some $y\in R(x)$, we have $y\not\in A$. Then since $A$ is $\between$-regular, there is a $z\between y$ such that $z\in \neg A$. Since $xRy\between z$, by $R$-\textsf{regularity} as in Lemma \ref{RregEquiv}, $\exists x'\between x\,\forall x''\between x'$ $\exists y''$: $x''R y''\between z$. Since $y''\between z$ implies $y''\not\in A$ given $z\in\neg A$, it follows that $\exists x'\between x\,\forall x''\between x'$ $\exists y''\in R(x'')$: $y''\not\in A$ and hence $x''\not\in\Box A$.\end{proof}

The underlying compatibility frame of a modal compatibility frame gives rise to an ortholattice just as in Proposition \ref{FrameToOrtho}, and the $R$ relation gives us the modal operation of a modal ortholattice.

\begin{proposition}\label{EpistemicFrameToEpistemicOrtho} For any modal compatibility frame $\mathcal{F}=\langle S,\between,R\rangle$, the $\between$-regular sets ordered by inclusion form a complete lattice, which becomes an ortholattice with $\neg$ defined as in (\ref{NegEq}) and then a modal ortholattice with $\Box$ as defined in Definition (\ref{BoxEq1}). We denote this modal ortholattice by~$O(\mathcal{F})$.
\end{proposition}
\begin{proof} Given Proposition \ref{FrameToOrtho}, we need only check the new modal part, which is completely standard:
\begin{eqnarray*}
\Box(A\cap B) &=& \{x\in S\mid R(x)\subseteq A\cap B\} \\
&=&\{x\in S\mid R(x)\subseteq A \} \cap \{x\in S\mid R(x)\subseteq B\} \\
&=&\Box A\cap \Box B,
\end{eqnarray*}
and $\Box 1 =\Box S=\{x\in S\mid R(x)\subseteq S\}=S =1$.\end{proof}

Now to interpret $R(x)$ as an \textit{information state}, we impose two additional conditions on $R$.  The first is the familiar reflexivity condition on epistemic accessibility.

\begin{definition}\label{TComp} A \textit{$\mathsf{T}$ compatibility frame} is a modal compatibility frame  satisfying:
\begin{itemize}
\item \textsf{Reflexivity}: for all $x\in S$,  $xRx$. 
\end{itemize}
\end{definition}
\noindent As usual, it follows from this condition that $\Box A$ entails $A$.\footnote{In fact, the following weaker condition suffices: for all $x\in S$, there is a $y\in R(x)$ with $x\sqsubseteq y$.}

\begin{proposition} For any $\mathsf{T}$ compatibility frame $\mathcal{F}$ and $\between$-regular set $A$, we have $\Box A\subseteq A$. Hence $O(\mathcal{F})$ is a $\mathsf{T}$ modal ortholattice.
\end{proposition}
\begin{proof} If $x\in \Box A$, then $R(x)\subseteq A$, which with $xRx$ implies $x\in A$.\end{proof}

The second condition essentially says that \textit{there is a possibility where everything settled true by $x$ is known}. 

\begin{definition}\label{Knowability} An \textit{epistemic compatibility frame} is a \textsf{T} compatibility frame also satisfying:
\begin{itemize}
\item \textsf{Knowability}: for all $x\in S$, there is a $y\in S$ such that for all $z\in R(y)$, $z\sqsubseteq x$.
\end{itemize}
\end{definition}
\noindent Recall from Lemma \ref{RefinementDef} that $z\sqsubseteq x$ implies that every proposition true at $x$ is also true at $z$.

\textsf{Knowability} can naturally be seen as the semantic constraint corresponding to the proof theoretic principle, noted in Lemma \ref{BoxBotLem}, that for any $\varphi\in\mathcal{EL}$, if $\Box\varphi\vdash_\mathsf{EO}\bot$, then $\varphi\vdash_\mathsf{EO}\bot$, or contrapositively, if $\varphi\nvdash_\mathsf{EO}\bot$, then  $\Box\varphi\nvdash_\mathsf{EO}\bot$.  In the presence of our background logic, that principle is equivalent to Wittgenstein's Law. Correspondingly, \textsf{Knowability} ensures that $\neg A$ and $\Diamond A$ are inconsistent in any epistemic compatibility frame.

\begin{proposition}\label{EpCont} For any epistemic compatibility frame $\mathcal{F}$ and $\between$-regular set $A$, we have $\neg A\cap \Diamond A=\varnothing$. Hence $O(\mathcal{F})$ is an epistemic ortholattice.
\end{proposition}
\begin{proof} Suppose $x\in\neg A$. By \textsf{Knowability}, there is a $y\in S$ such that for all $z\in R(y)$, we have $z\sqsubseteq x$ and hence $z\in\neg A$ by Lemma \ref{RefinementDef}, so $y\in \Box\neg A$. Since $y\in R(y)$ by \textsf{Reflexivity}, it follows that $y\sqsubseteq x$ and hence $y\between x$, which with $y\in \Box\neg A$ implies $x\not\in \neg\Box\neg A$, so $x\not\in\Diamond A$.\end{proof}

\noindent In essence, the explanation of the badness of Wittgenstein sentences according to our possibility semantics is this: since for any possibility $x$, \textit{there is a possibility where everything settled true by $x$ is known}, then $x$ cannot settle $\neg A\wedge\Diamond A$ as true, for that would imply the existence of a possibility $y$ settling as true the contradictory proposition $\Box(\neg A\cap \Diamond A)\subseteq \Box \neg A\cap \Box\Diamond A\subseteq \Box \neg A\cap \Diamond A=\varnothing$.

\begin{remark} Some (e.g., \citealt{Lassiter:2016}) reject the \textsf{T} axiom for $\Box$. Without \textsf{T}, we can validate Wittgenstein's Law and its generalization in Fact \ref{GeneralizedWittLaw} by adopting a slightly stronger constraint than \textsf{Knowability}:   for all $x\in S$ and $n\in\mathbb{N}$, there is a $y\between x$ such that for all $z\in R^n(y)$, we have $z\sqsubseteq x$, where $R^0(y)=\{y\}$ and $R^{n+1}(y)=R[R^n(y)]$. Clearly this suffices for the proof in Proposition \ref{EpCont} that $\neg A\cap \Diamond A=\varnothing$ and its generalization to $\neg A\cap \Diamond^n A=\varnothing$, so those who reject \textsf{T} can still use our approach to account for the contradictoriness of (generalized) Wittgenstein sentences.\end{remark}

\begin{remark}When we say `there is a possibility where everything settled true by $x$ is known', we are using `known' in a loose way. If we identified epistemic modals with the knowledge of some particular agent (or group) at a time---the speaker, say---then this condition would be somewhat implausible in general, since we can plausibly imagine possibilities where it is settled that not everything settled is known. Instead, we are thinking about `what is known'  as corresponding to something like salient information (possibly) augmented with further facts that may well go beyond what any individual or group knows.
The idea that epistemic modals can track information that goes beyond what is actually known by anyone goes back to \citet{Hacking1967}. Hacking considers

\begin{quote} a salvage crew searching for a ship that sank a long time ago. The mate of the salvage ship works from an old log, makes some mistakes in his calculations, and concludes that the wreck may be in a certain bay. It is possible, he says, that the hulk is in these waters. No one knows anything to the contrary. But in fact, as it turns out later, it simply was not possible for the vessel to be in that bay; more careful examination of the log shows that the boat must have gone down at least thirty miles further south. The mate said something false when he said, `It is possible that we shall find the treasure here,' but the falsehood did not arise from what anyone actually knew at the time. \hfill \citep[p. 148]{Hacking1967}\end{quote} 
(We can, if we want, elaborate stories like this so that the falsehood does not arise from what anyone actually knows at \emph{any} time.)
This idea is  elaborated in different ways in \citealt{Dorr:2013}, \citealt{Mandelkern:2018a}, and \citealt{Kratzer:2017}, and we think it is fundamentally correct. Of course, this idea 
 remains vague and context sensitive; nonetheless, we think that our \textsf{Knowability} constraint is a natural way of making precise the idea that epistemic modalities involve some mix of purely epistemic and factual considerations.  
\end{remark}

To summarize, we impose the following three conditions:
\begin{itemize}
\item $R$-\textsf{regularity}: if $xRy'\between y$, then $\exists x'\between x$: $x'\in\Diamond \mathord{\downarrow}y$\\
(if $x$ can epistemically access a possibility compatible with $y$, then $x$ is compatible with a possibility according to which $y$ \textit{might obtain});
\item \textsf{Reflexivity}: for all $x \in S$, $xRx$ \\ 
(everything known is true);
\item \textsf{Knowability}: for all $x\in S$, there is a $y\in S$ such that for all $z\in R(y)$, we have $z\sqsubseteq x$ \\
(there is a possibility where everything settled true by $x$ is known).
\end{itemize}

\begin{example}\label{EpCompEx} Figure \ref{EpCompFig} shows an epistemic compatibility frame $\mathcal{F}$ based on the compatibility frame of Example \ref{CompEx}. It is a good exercise to check that the three conditions on the $R$ relation are satisfied. The epistemic ortholattice $O(\mathcal{F})$ is isomorphic to the epistemic ortholattice in Figure \ref{Fig1}, as we will see when we discuss a model based on the frame $\mathcal{F}$ in Example \ref{EpModEx}.

\begin{figure}[h]
\begin{center}
\tikzset{every loop/.style={min distance=10mm,looseness=10}}
\begin{tikzpicture}[->,>=stealth',shorten >=1pt,shorten <=1pt, auto,node
distance=2.5cm,semithick]
\tikzstyle{every state}=[fill=gray!20,draw=none,text=black]

\node[circle,draw=black!100, fill=black!100, label=below:$x_1$,inner sep=0pt,minimum size=.175cm] (1) at (0,0) {{}};

\node[circle,draw=black!100, fill=black!100, label=below:$x_2$,inner sep=0pt,minimum size=.175cm] (2) at (1.5,0) {{}};

\node[circle,draw=black!100, fill=black!100, label=below:$x_3$,inner sep=0pt,minimum size=.175cm] (3) at (3,0) {{}};

\node[circle,draw=black!100, fill=black!100, label=below:$x_4$,inner sep=0pt,minimum size=.175cm] (4) at (4.5,0) {{}};
\node[circle,draw=black!100, fill=black!100, label=below:$x_5$,inner sep=0pt,minimum size=.175cm] (5) at (6,0) {{}};

\path (1) edge[-] node {{}} (2);
\path (2) edge[-] node {{}} (3);
\path (3) edge[-] node {{}} (4);
\path (4) edge[-] node {{}} (5);

\path (2) edge[dotted,thick,bend left, MidnightBlue] node {{}} (1);
\path (2) edge[dotted,thick,bend left, MidnightBlue] node {{}} (3);

\path (4) edge[dotted,thick,bend left, MidnightBlue] node {{}} (3);
\path (4) edge[dotted,thick,bend left, MidnightBlue] node {{}} (5);

\end{tikzpicture}

\end{center}\caption{An epistemic compatibility frame with dotted arrows indicating the accessibility relation $R$ (with reflexive loops omitted)}\label{EpCompFig}
\end{figure}
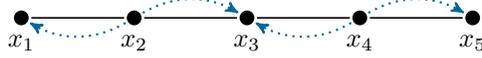
\end{example}

Turning to formal semantics for $\mathcal{EL}$ (recall Definition \ref{ELDef}), Proposition \ref{EpistemicFrameToEpistemicOrtho} leads to the following.

\begin{definition}\label{ModalSemantics} Given a modal compatibility model $\mathcal{M}= \langle S,\between,R,V\rangle$, $x\in S$, and $\varphi\in\mathcal{L}$, we define $\mathcal{M},x\Vdash \varphi$  with the same clauses as in Definition \ref{GoldblattSemantics} plus
\begin{itemize}
\item $\mathcal{M},x\Vdash\Box\varphi$ iff for all $y\in R(x)$, we have $\mathcal{M},y\Vdash \varphi$.
\end{itemize}
As before, we define $\llbracket \varphi\rrbracket^\mathcal{M}=\{x\in S\mid \mathcal{M},x\Vdash\varphi\}$.
\end{definition}
\noindent  Then given our definition of $\Diamond$ as $\neg\Box\neg$, as in (\ref{DiamondEq}) we have:
\begin{itemize}
\item $\mathcal{M},x\Vdash\Diamond\varphi$ iff $\forall x'\between x$ $\exists y'\in R(x')$ $\exists y''\between y'$: $\mathcal{M},y''\Vdash\varphi$.
\end{itemize}

Once again an easy induction shows the following.
\begin{lemma} For any modal compatibility model $\mathcal{M}$ and $\varphi\in\mathcal{EL}$, $\llbracket \varphi\rrbracket^\mathcal{M}$ is a $\between$-regular set.\end{lemma}

\begin{example}\label{EpModEx} Figure \ref{EpCompModFig} shows an epistemic compatibility model based on the epistemic compatibility frame in Figure \ref{EpCompFig} with $V(p)=\{x_1,x_2\}$, so $V(p)$ is a $\between$-regular set and $\llbracket \neg p\rrbracket^\mathcal{M}=\{x_4,x_5\}$. We call this model the Epistemic Scale for $p$, since as we move from $x_1$ to $x_5$ we move as if on a scale from $\Box p$ at $x_1$ to $p$ at $x_2$ to $\Diamond p\wedge\Diamond\neg p$ at $x_3$ to $\neg p$ at $x_4$ to $\Box\neg p$ at $x_5$. Observe the following:
\begin{center}
\begin{minipage}{3in}
\begin{itemize}
\item $\llbracket \Box p\rrbracket^\mathcal{M} =\{x_1\}$;
\item $\llbracket \neg\Box p\rrbracket^\mathcal{M} =\llbracket \Diamond\neg p\rrbracket^\mathcal{M} =\{x_3,x_4,x_5\}$;
\item $\llbracket \Box \neg p\rrbracket^\mathcal{M} =\{x_5\}$;
\end{itemize}
\end{minipage}\begin{minipage}{3in}
\begin{itemize}
\item $\llbracket \neg\Box \neg p\rrbracket^\mathcal{M} =\llbracket \Diamond p\rrbracket^\mathcal{M}=\{x_1,x_2,x_3\}$;
\item $\llbracket \Diamond p\wedge\Diamond\neg p \rrbracket^\mathcal{M}=\{x_3\}$;
\item $\llbracket \Box p\vee \Box\neg p \rrbracket^\mathcal{M}=\{x_1,x_5\}$.
\end{itemize}
\end{minipage}
\end{center}
These calculations match the fact that the associated epistemic ortholattice is isomorphic to that in Figure~\ref{Fig1}. 

Let us also see how our possibility semantics explains the failure of the distributive law of classical logic. Consider  $x_3$, which is a \textit{partial} possibility: it does not settle $p$ as true and it does not settle $\neg p$ as true, as it is compatible with both (since it is compatible with $x_2$ and $x_4$), but like any possibility, it settles $p\vee\neg p$ as true. Now since the information available in $x_3$ is $x_3$ itself, $x_3$ settles that $p$ might be true and that $\neg p$ might be true: $\Diamond p\wedge \Diamond\neg p$. Given that $(p\vee\neg p)\wedge (\Diamond p\wedge \Diamond\neg p)$ is true at $x_3$, the distributive law would require that  $(p\wedge\Diamond \neg p)\vee (\neg p\wedge\Diamond p)$ be true at $x_3$. But we have already seen in our discussion after Proposition \ref{EpCont} why no possibility can settle either of those disjuncts as true. Then since no possibility settles either disjunct as true, $x_3$ does not settle the disjunction as true, which shows that the distributive law is invalid.

\begin{figure}[h]
\begin{center}
\tikzset{every loop/.style={min distance=10mm,looseness=10}}
\begin{tikzpicture}[->,>=stealth',shorten >=1pt,shorten <=1pt, auto,node
distance=2.5cm,semithick]
\tikzstyle{every state}=[fill=gray!20,draw=none,text=black]

\node[circle,draw=black!100, fill=green!100, label=below:$x_1$,inner sep=0pt,minimum size=.175cm] (1) at (0,0) {{}};
\node at (0,-.75) {{$p$}};

\node at (0,-1.25) {{$\Box p$}};

\node at (3,-.75) {{$\Diamond p\wedge\Diamond\neg p$}};

\node at (4.5,-.75) {{$\neg p$}};

\node at (6,-.75) {{$\neg p$}};

\node at (6,-1.25) {{$\Box \neg p$}};

\node[circle,draw=black!100, fill=green!100, label=below:$x_2$,inner sep=0pt,minimum size=.175cm] (2) at (1.5,0) {{}};
\node at (1.5,-.75) {{$p$}};

\node[circle,draw=black!100, fill=black!100, label=below:$x_3$,inner sep=0pt,minimum size=.175cm] (3) at (3,0) {{}};

\node[circle,draw=black!100, fill=black!100, label=below:$x_4$,inner sep=0pt,minimum size=.175cm] (4) at (4.5,0) {{}};
\node[circle,draw=black!100, fill=black!100, label=below:$x_5$,inner sep=0pt,minimum size=.175cm] (5) at (6,0) {{}};

\path (1) edge[-] node {{}} (2);
\path (2) edge[-] node {{}} (3);
\path (3) edge[-] node {{}} (4);
\path (4) edge[-] node {{}} (5);

\path (2) edge[dotted,thick,bend left, MidnightBlue] node {{}} (1);
\path (2) edge[dotted,thick,bend left, MidnightBlue] node {{}} (3);

\path (4) edge[dotted,thick,bend left, MidnightBlue] node {{}} (3);
\path (4) edge[dotted,thick,bend left, MidnightBlue] node {{}} (5);

\end{tikzpicture} 

\end{center}\caption{An epistemic compatibility model, dubbed the Epistemic Scale for $p$, based on the frame of Figure~\ref{EpCompFig}, with $V(p)$ highlighted in green.}\label{EpCompModFig}
\end{figure}
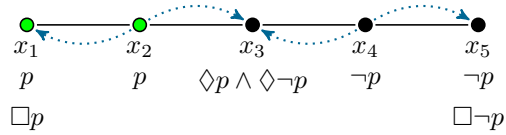
\end{example}

\begin{example} Let us consider an example with two dimensions of epistemic variation in contrast to the one dimension of the Epistemic Scale. Figure \ref{EpistemicGrid} shows what we call an Epistemic Grid for $p$ and $q$. It is obtained from the Epistemic Scale by a product construction: each possibility in the Epistemic Grid is a pair of possibilities $(x,y)$ with $x$ from the Epistemic Scale for $p$ and $y$ from the Epistemic Scale for $q$; $(x,y)$ is compatible with (resp.~accessible from) $(x',y')$ iff $x$ is compatible with (resp.~accessible from) $x'$ and $y$ is compatible with (resp.~accessible from) $y'$;\footnote{This construction of the accessibility relation guarantees that \textsf{$R$-regularity}, \textsf{Reflexivity}, and \textsf{Knowability} are still satisfied.} and a propositional variable is true at $(x,y)$ iff it is true at $x$ or at $y$. We encourage the reader to calculate the truth values of some formulas at possibilities in this Epistemic Grid. For example, $\neg (p\wedge q)$ is true at the possibilities in the bottommost two rows and rightmost two columns; hence $\Box\neg (p\wedge q)$ is true in the bottommost row and rightmost column; and hence $\neg\Box\neg (p\wedge q)$, i.e., $\Diamond (p\wedge q)$, is true at the nine possibilities in the upper-left quadrant, which are also the possibilities where $\Diamond p\wedge\Diamond q$ is true. By contrast, along the centermost column, $\Diamond p\wedge \Diamond\neg p\wedge\neg \Diamond(p\wedge\neg p)$ is true. 

While the epistemic ortholattice arising from an Epistemic Scale has 10 elements, shown on the right of Figure \ref{JIfig}, the epistemic ortholattice arising from an Epistemic Grid has 1,942 elements.\footnote{Though that is small compared to the size of the powerset of the Epistemic Grid at $2^{25}=33,554,432$ elements. See the second of the  notebooks cited in \S~\ref{Intro} for the calculation of the 1,942 figure.} This is another reason why we want a possibility semantics in addition to an algebraic semantics: while a 1,942-element epistemic ortholattice is too large to grasp, we can represent it using a perspicuous 25-element relational  epistemic compatibility frame. One can even visualize an Epistemic Cube for three variables and consider more general $n$-dimensional Epistemic Hypercubes (see the second notebook cited in \S~\ref{Intro}).

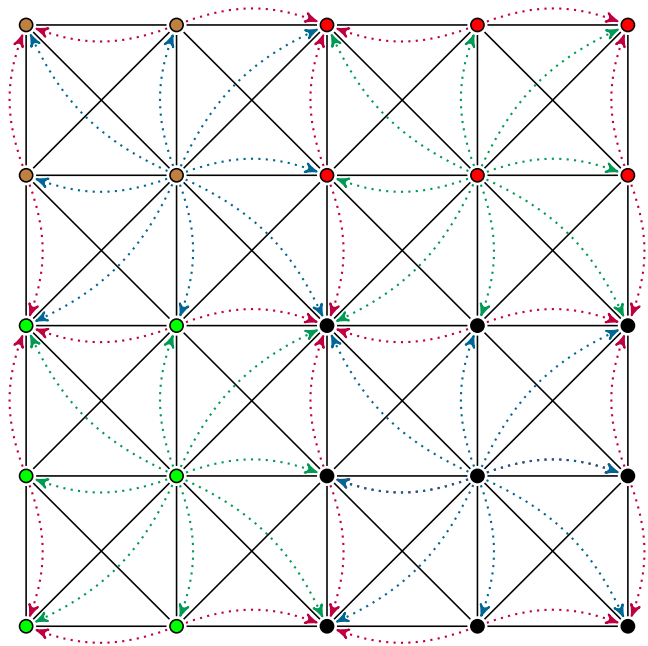
\begin{figure}[h]
\begin{center}
\tikzset{every loop/.style={min distance=10mm,looseness=10}}
\begin{tikzpicture}[->,>=stealth',shorten >=1pt,shorten <=1pt, auto,node
distance=2.5cm,semithick]
\tikzstyle{every state}=[fill=gray!20,draw=none,text=black]

\node[circle,draw=black!100, fill=brown!100, label=below:$$,inner sep=0pt,minimum size=.175cm] (1) at (0,0) {{}};
\node[circle,draw=black!100, fill=brown!100, label=below:$$,inner sep=0pt,minimum size=.175cm] (2) at (2,0) {{}};
\node[circle,draw=black!100, fill=red!100, label=below:$$,inner sep=0pt,minimum size=.175cm] (3) at (4,0) {{}};
\node[circle,draw=black!100, fill=red!100, label=below:$$,inner sep=0pt,minimum size=.175cm] (4) at (6,0) {{}};
\node[circle,draw=black!100, fill=red!100, label=below:$$,inner sep=0pt,minimum size=.175cm] (5) at (8,0) {{}};

\path (1) edge[-] node {{}} (2);
\path (2) edge[-] node {{}} (3);
\path (3) edge[-] node {{}} (4);
\path (4) edge[-] node {{}} (5);

\path (2) edge[dotted,thick,bend left=20, purple] node {{}} (1);
\path (2) edge[dotted,thick,bend left=20, purple] node {{}} (3);

\path (4) edge[dotted,thick,bend left=20, purple] node {{}} (3);
\path (4) edge[dotted,thick,bend left=20, purple] node {{}} (5);

\node[circle,draw=black!100, fill=brown!100, label=below:$$,inner sep=0pt,minimum size=.175cm] (1') at (0,-2) {{}};
\node[circle,draw=black!100, fill=brown!100, label=below:$$,inner sep=0pt,minimum size=.175cm] (2') at (2,-2) {{}};
\node[circle,draw=black!100, fill=red!100, label=below:$$,inner sep=0pt,minimum size=.175cm] (3') at (4,-2) {{}};
\node[circle,draw=black!100, fill=red!100, label=below:$$,inner sep=0pt,minimum size=.175cm] (4') at (6,-2) {{}};
\node[circle,draw=black!100, fill=red!100, label=below:$$,inner sep=0pt,minimum size=.175cm] (5') at (8,-2) {{}};

\path (1') edge[-] node {{}} (2');
\path (2') edge[-] node {{}} (3');
\path (3') edge[-] node {{}} (4');
\path (4') edge[-] node {{}} (5');

\path (1') edge[-] node {{}} (1);
\path (2') edge[-] node {{}} (2);
\path (3') edge[-] node {{}} (3);
\path (4') edge[-] node {{}} (4);
\path (5') edge[-] node {{}} (5);

\path (1') edge[-] node {{}} (2);
\path (2') edge[-] node {{}} (3);
\path (3') edge[-] node {{}} (4);
\path (4') edge[-] node {{}} (5);

\path (2') edge[-] node {{}} (1);
\path (3') edge[-] node {{}} (2);
\path (4') edge[-] node {{}} (3);
\path (5') edge[-] node {{}} (4);

\node[circle,draw=black!100, fill=green!100, label=below:$$,inner sep=0pt,minimum size=.175cm] (1'') at (0,-4) {{}};
\node[circle,draw=black!100, fill=green!100, label=below:$$,inner sep=0pt,minimum size=.175cm] (2'') at (2,-4) {{}};
\node[circle,draw=black!100, fill=black!100, label=below:$$,inner sep=0pt,minimum size=.175cm] (3'') at (4,-4) {{}};
\node[circle,draw=black!100, fill=black!100, label=below:$$,inner sep=0pt,minimum size=.175cm] (4'') at (6,-4) {{}};
\node[circle,draw=black!100, fill=black!100, label=below:$$,inner sep=0pt,minimum size=.175cm] (5'') at (8,-4) {{}};

\path (1'') edge[-] node {{}} (2'');
\path (2'') edge[-] node {{}} (3'');
\path (3'') edge[-] node {{}} (4'');
\path (4'') edge[-] node {{}} (5'');

\path (1'') edge[-] node {{}} (1');
\path (2'') edge[-] node {{}} (2');
\path (3'') edge[-] node {{}} (3');
\path (4'') edge[-] node {{}} (4');
\path (5'') edge[-] node {{}} (5');

\path (1'') edge[-] node {{}} (2');
\path (2'') edge[-] node {{}} (3');
\path (3'') edge[-] node {{}} (4');
\path (4'') edge[-] node {{}} (5');

\path (2'') edge[-] node {{}} (1');
\path (3'') edge[-] node {{}} (2');
\path (4'') edge[-] node {{}} (3');
\path (5'') edge[-] node {{}} (4');

\path (2'') edge[dotted,thick,bend left=20, purple] node {{}} (1'');
\path (2'') edge[dotted,thick,bend left=20, purple] node {{}} (3'');

\path (4'') edge[dotted,thick,bend left=20, purple] node {{}} (3'');
\path (4'') edge[dotted,thick,bend left=20, purple] node {{}} (5'');

\node[circle,draw=black!100, fill=green!100, label=below:$$,inner sep=0pt,minimum size=.175cm] (1''') at (0,-6) {{}};
\node[circle,draw=black!100, fill=green!100, label=below:$$,inner sep=0pt,minimum size=.175cm] (2''') at (2,-6) {{}};
\node[circle,draw=black!100, fill=black!100, label=below:$$,inner sep=0pt,minimum size=.175cm] (3''') at (4,-6) {{}};
\node[circle,draw=black!100, fill=black!100, label=below:$$,inner sep=0pt,minimum size=.175cm] (4''') at (6,-6) {{}};
\node[circle,draw=black!100, fill=black!100, label=below:$$,inner sep=0pt,minimum size=.175cm] (5''') at (8,-6) {{}};

\path (1''') edge[-] node {{}} (2''');
\path (2''') edge[-] node {{}} (3''');
\path (3''') edge[-] node {{}} (4''');
\path (4''') edge[-] node {{}} (5''');

\path (1''') edge[-] node {{}} (1'');
\path (2''') edge[-] node {{}} (2'');
\path (3''') edge[-] node {{}} (3'');
\path (4''') edge[-] node {{}} (4'');
\path (5''') edge[-] node {{}} (5'');

\path (1''') edge[-] node {{}} (2'');
\path (2''') edge[-] node {{}} (3'');
\path (3''') edge[-] node {{}} (4'');
\path (4''') edge[-] node {{}} (5'');

\path (2''') edge[-] node {{}} (1'');
\path (3''') edge[-] node {{}} (2'');
\path (4''') edge[-] node {{}} (3'');
\path (5''') edge[-] node {{}} (4'');

\path (4''') edge[dotted,thick,bend left=20, purple] node {{}} (3''');
\path (4''') edge[dotted,thick,bend left=20, purple] node {{}} (5''');

\node[circle,draw=black!100, fill=green!100, label=below:$$,inner sep=0pt,minimum size=.175cm] (1'''') at (0,-8) {{}};
\node[circle,draw=black!100, fill=green!100, label=below:$$,inner sep=0pt,minimum size=.175cm] (2'''') at (2,-8) {{}};
\node[circle,draw=black!100, fill=black!100, label=below:$$,inner sep=0pt,minimum size=.175cm] (3'''') at (4,-8) {{}};
\node[circle,draw=black!100, fill=black!100, label=below:$$,inner sep=0pt,minimum size=.175cm] (4'''') at (6,-8) {{}};
\node[circle,draw=black!100, fill=black!100, label=below:$$,inner sep=0pt,minimum size=.175cm] (5'''') at (8,-8) {{}};

\path (1'''') edge[-] node {{}} (2'''');
\path (2'''') edge[-] node {{}} (3'''');
\path (3'''') edge[-] node {{}} (4'''');
\path (4'''') edge[-] node {{}} (5'''');

\path (1'''') edge[-] node {{}} (1''');
\path (2'''') edge[-] node {{}} (2''');
\path (3'''') edge[-] node {{}} (3''');
\path (4'''') edge[-] node {{}} (4''');
\path (5'''') edge[-] node {{}} (5''');

\path (1'''') edge[-] node {{}} (2''');
\path (2'''') edge[-] node {{}} (3''');
\path (3'''') edge[-] node {{}} (4''');
\path (4'''') edge[-] node {{}} (5''');

\path (2'''') edge[-] node {{}} (1''');
\path (3'''') edge[-] node {{}} (2''');
\path (4'''') edge[-] node {{}} (3''');
\path (5'''') edge[-] node {{}} (4''');

\path (2'''') edge[dotted,thick,bend left=20, purple] node {{}} (1'''');
\path (2'''') edge[dotted,thick,bend left=20, purple] node {{}} (3'''');

\path (4'''') edge[dotted,thick,bend left=20, purple] node {{}} (3'''');
\path (4'''') edge[dotted,thick,bend left=20, purple] node {{}} (5'''');

\path (1') edge[dotted,thick,bend left=20, purple] node {{}} (1);
\path (1') edge[dotted,thick,bend left=20, purple] node {{}} (1'');

\path (1''') edge[dotted,thick,bend left=20, purple] node {{}} (1'');
\path (1''') edge[dotted,thick,bend left=20, purple] node {{}} (1'''');

\path (3') edge[dotted,thick,bend left=20, purple] node {{}} (3);
\path (3') edge[dotted,thick,bend left=20, purple] node {{}} (3'');

\path (3''') edge[dotted,thick,bend left=20, purple] node {{}} (3'');
\path (3''') edge[dotted,thick,bend left=20, purple] node {{}} (3'''');

\path (5') edge[dotted,thick,bend left=20, purple] node {{}} (5);
\path (5') edge[dotted,thick,bend left=20, purple] node {{}} (5'');

\path (5''') edge[dotted,thick,bend left=20, purple] node {{}} (5'');
\path (5''') edge[dotted,thick,bend left=20, purple] node {{}} (5'''');

\path (2') edge[dotted,thick,bend left=20, MidnightBlue] node {{}} (1');
\path (2') edge[dotted,thick,bend left=20, MidnightBlue] node {{}} (3');
\path (2') edge[dotted,thick,bend left=20, MidnightBlue] node {{}} (2);
\path (2') edge[dotted,thick,bend left=20, MidnightBlue] node {{}} (2'');
\path (2') edge[dotted,thick,bend left=20, MidnightBlue] node {{}} (1);
\path (2') edge[dotted,thick,bend left=20, MidnightBlue] node {{}} (3);
\path (2') edge[dotted,thick,bend left=20, MidnightBlue] node {{}} (3'');
\path (2') edge[dotted,thick,bend left=20, MidnightBlue] node {{}} (1'');

\path (4') edge[dotted,thick,bend left=20, ForestGreen] node {{}} (3');
\path (4') edge[dotted,thick,bend left=20, ForestGreen] node {{}} (5');
\path (4') edge[dotted,thick,bend left=20, ForestGreen] node {{}} (4);
\path (4') edge[dotted,thick,bend left=20, ForestGreen] node {{}} (4'');
\path (4') edge[dotted,thick,bend left=20, ForestGreen] node {{}} (3);
\path (4') edge[dotted,thick,bend left=20, ForestGreen] node {{}} (5);
\path (4') edge[dotted,thick,bend left=20, ForestGreen] node {{}} (3'');
\path (4') edge[dotted,thick,bend left=20, ForestGreen] node {{}} (5'');

\path (2''') edge[dotted,thick,bend left=20, ForestGreen] node {{}} (2'');
\path (2''') edge[dotted,thick,bend left=20, ForestGreen] node {{}} (2'''');
\path (2''') edge[dotted,thick,bend left=20, ForestGreen] node {{}} (1''');
\path (2''') edge[dotted,thick,bend left=20, ForestGreen] node {{}} (3''');
\path (2''') edge[dotted,thick,bend left=20, ForestGreen] node {{}} (1'');
\path (2''') edge[dotted,thick,bend left=20, ForestGreen] node {{}} (3'');
\path (2''') edge[dotted,thick,bend left=20, ForestGreen] node {{}} (1'''');
\path (2''') edge[dotted,thick,bend left=20, ForestGreen] node {{}} (3'''');

\path (4''') edge[dotted,thick,bend left=20, MidnightBlue] node {{}} (4'');
\path (4''') edge[dotted,thick,bend left=20, MidnightBlue] node {{}} (4'''');
\path (4''') edge[dotted,thick,bend left=20, MidnightBlue] node {{}} (3''');
\path (4''') edge[dotted,thick,bend left=20, MidnightBlue] node {{}} (5''');
\path (4''') edge[dotted,thick,bend left=20, MidnightBlue] node {{}} (3'');
\path (4''') edge[dotted,thick,bend left=20, MidnightBlue] node {{}} (5'');
\path (4''') edge[dotted,thick,bend left=20, MidnightBlue] node {{}} (3'''');
\path (4''') edge[dotted,thick,bend left=20, MidnightBlue] node {{}} (5'''');
\end{tikzpicture} 

\end{center}\caption{An Epistemic Grid for $p$ and $q$. Green possibilities make $p$ true; red possibilities make $q$ true; and brown possibilities makes both $p$ and $q$ true. Different colors for different dotted arrows are used only to make the pattern more intelligible; all such arrows represent the accessibility relation. Reflexive dotted loops between each possibility and itself are assumed but omitted from the diagram.}\label{EpistemicGrid}
\end{figure}
\end{example}

\begin{remark}\label{OtherMust} While we interpret the `must' modality $\Box$ using the $R$ relation, there is another operation, which we will denote by $\Box_\between$, available on the ortholattice $O(\mathcal{F})$ of any compatibility frame $\mathcal{F}$: $\Box_\between A={\{x\in S\mid \forall x'\between x\;\, x'\in A\}}$. Note that if $A$ is $\between$-regular, so is $\Box_\between A$: for if $x\not\in \Box_\between A$, then there is a $y\between x$ with $y\not\in A$, which implies that for all $z\between y$, $z\not\in\Box_\between A$.\footnote{Note, by contrast, that the function $\blacklozenge_\between$ defined by $\blacklozenge_\between A=\{x\in S\mid \exists x'\between x: x'\in A\}$ is not an operation on the ortholattice $O(\mathcal{F})$, i.e., it is not guaranteed that if $A$ is $\between$-regular, then so is $\blacklozenge_\between A$. For example, in the compatibility frame in Figure \ref{Fig2}, $\{x_1,x_2,x_3\}$ is $\between$-regular, but $\blacklozenge_\between \{x_1,x_2,x_3\}=\{x_1,x_2,x_3,x_4\}$ is not. A natural response to this problem is to apply the closure operator from Footnote \ref{ClosureOperator} to $ \blacklozenge_\between A$, interpreting `might $A$' as $c_\between \blacklozenge_\between A$. But $c_\between \blacklozenge_\between $ is just the dual of $\Box_\between$, i.e., $c_\between \blacklozenge_\between A = \neg \Box_\between \neg A$, so the reasons for rejecting $\Box_\between$ as `must' in the remainder of Remark \ref{OtherMust} are also reasons for rejecting  $c_\between \blacklozenge_\between$ as `might'.} 

In fact, interpreting `must' as $\Box_\between$ is just a special case of our possibility semantics: the proposal is just to set $R=\,\between$ (note using Lemma \ref{RregEquiv} that $R$-\textsf{regularity} holds). However, while the $\Box_\between$ modality may have useful applications,\footnote{For example, compatibility frames can of course also be regarded as classical possible world frames with $\between$ as an accessibility relation for the modality $\Box_\between$ on the powerset of the set of possibilities. This leads to the translation of orthologic into the classical modal logic \textsf{KTB} (\citealt{Goldblatt1974}, cf.~\citealt{Dishkant1977}), similar in spirit to the translation of intuitionistic logic into classical \textsf{S4} (\citealt{Godel1933b}, \citealt{McKinsey1948}). Likewise, epistemic compatibility frames can be regarded as classical \textit{bimodal} possible world frames, just as possibility frames for classical modal logic are regarded as bimodal possible world frames in \citealt{Benthem2017}. This leads to a translation of our epistemic orthologic into a classical bimodal logic. We omit the details, which can be worked out on the model of \citealt{Benthem2017}.} it is not appropriate for interpreting `must'. For one thing, it validates the B principle (where $\Diamond_\between A=\neg\Box_\between \neg A$): for all $A\in O(\mathcal{F})$, $\Diamond_\between \Box_\between A\subseteq A$, which we do not want to require in general. Even in cases where B is acceptable, $\Box_\between$ may not give the desired results for `must'. For example, in the frame in Figure \ref{EpCompModFig}, where $P=\{x_1,x_2\}$, we have $\Box P=\Box_\between P=\{x_1\}$, $\Box \neg P=\Box_\between \neg P=\{x_5\}$, and $\neg\Box P\wedge\neg\Box\neg P=\neg\Box_\between P\wedge\neg\Box_\between\neg P=\{x_3\}$; however, while $\Box\Box P=\{x_1\}$, $\Box\neg\Box P=\{x_5\}$, and $\Box(\neg \Box P\wedge\neg\Box\neg P)=\{x_3\}$, we have $\Box_\between\Box_\between P=\Box_\between\neg\Box_\between P=\Box_\between (\neg\Box_\between P\wedge \neg\Box_\between \neg P)=\varnothing$. The general problem is that iterating  $\Box_\between$ keeps shrinking the relevant set until one reaches a set of possibilities that are not compatible with anything outside it. Most fatal for the idea of interpreting `must' as $\Box_\between$ is that this approach fails to make Wittgenstein sentences contradictions. Consider a compatibility frame with three possibilities such that $x \between y \between z$, but not $x\between z$. Set $V(p) = \{x\}$, which is a $\between$-regular set. Then $p\wedge\neg\Box_\between p$ is true at $x$. The moral is that we cannot reduce \textit{epistemic} modality entirely to the \textit{alethic} relation of compatibility.
\end{remark}

As before, we define semantic consequence standardly in terms of truth preservation.

\begin{definition}\label{PossCon1} Given a class $\mathbf{C}$ of modal compatibility frames, define the semantic consequence relation $\vDash_\mathbf{C}$, a binary relation on $\mathcal{EL}$, as follows: $\varphi\vDash_\mathbf{C}\psi$ if for every  $\mathcal{F}\in\mathbf{C}$, model $\mathcal{M}$ based on $\mathcal{F}$, and possibility $x$ in $\mathcal{M}$, if $\mathcal{M},x\Vdash\varphi$, then $\mathcal{M},x\Vdash\psi$.
\end{definition}

Our first main result is the completeness of the epistemic orthologic $\textsf{EO}$ in Definition \ref{EODef} with respect to our epistemic possibility semantics. 

\begin{theorem}\label{EOThm} The logic $\textsf{EO}$ is sound and complete with respect to the class $\mathbf{ECF}$ of all epistemic compatibility frames according to the consequence relation of Definition \ref{PossCon1}: for all $\varphi,\psi\in\mathcal{EL}$, we have $\varphi\vdash_\mathsf{EO}\psi$ if and only if $\varphi\vDash_\mathbf{ECF}\psi$.
\end{theorem}

\begin{proof} Soundness is a straightforward check of the principles of the logic. For completeness, as in the proofs of Theorem \ref{AlgComp1} and \ref{AlgComp2}, we consider the Lindenbaum-Tarski algebra $L$ of \textsf{EO}, which is an epistemic ortholattice. From $L$ we define a modal compatibility frame $\mathcal{F}=\langle S,\between,R\rangle$ as follows:
\begin{itemize}
\item $S$ is the set of all proper filters of $L$;\footnote{\label{FilterNote}Recall that a \textit{filter} in a lattice is a nonempty set $F$ of elements such that for any $a,b\in L$, together $a\in F$ and $a\leq b$ imply $b\in F$, and $a,b\in F$ implies $a\wedge b\in F$; a filter is \textit{proper} if it does not contain every element of the lattice.}
\item for $F,G\in S$, $F\between G$ if there is no $a\in F$ with $\neg a\in G$;
\item for $F,G\in S$, $FRG$ iff for all $\Box a\in F$, we have $a\in G$.
\end{itemize}
That $\between$ is reflexive follows from the fact that each $F\in S$ is a \textit{proper} filter, and in $L$, $a\wedge\neg a=0$. That $\between$ is symmetric follows from the fact that if $a\in F$ and $\neg a\in G$, then $\neg a\in G$ and $\neg\neg a\in F$, since $a=\neg\neg a$ in $L$, so there is a $b$ such that $b\in G$ and $\neg b\in F$. That $R$ is reflexive follows from $L$ being a \textsf{T} modal ortholattice.

Below we will use the fact that if $F$ is a proper filter, then so is $i(F)=\{a\in B\mid \Box a\in F\}$. For if not, then there are $\Box a_1,\dots,\Box a_n\in F$ such that $a_1\wedge\dots\wedge a_n=0$, which implies $\Box(a_1\wedge\dots\wedge a_n)=\Box 0$ and hence $\Box a_1\wedge\dots \wedge\Box a_n=\Box0$.  Since $L$ is a \textsf{T} modal ortholattice, $\Box 0 =0$, so we obtain $\Box a_1\wedge\dots \wedge\Box a_n=0$, which contradicts the assumption that $F$ is proper, since $\Box a_1,\dots,\Box a_n\in F$ implies $\Box a_1\wedge\dots \wedge\Box a_n\in F$. Thus, $i(F)$ is proper, and by definition of $R$, we have $FRi(F)$. 

Now let us prove $R$-\textsf{regularity} in its form in Lemma \ref{RregEquiv}:
\begin{itemize}
\item if $FRG'\between G$, then $\exists F'\between F$ $\forall F''\between F'$ $\exists G''$: $F''RG''\between G$.
\end{itemize}
Let $F'$ be the filter generated by $\{\Diamond a\mid a\in G\}$. We claim that $F'$ is proper and that $F'\between F$. If not, then for some $b\in F'$, $\neg b\in F$. Since $b\in F'$, there are $a_1,\dots,a_n\in G$ such that $\Diamond a_1\wedge\dots \wedge\Diamond a_n\leq b$. It follows that $\Diamond (a_1\wedge\dots\wedge a_n)\leq b$, so $\neg b\leq \Box \neg (a_1\wedge\dots\wedge a_n)$. Hence $\Box \neg (a_1\wedge\dots\wedge a_n)\in F$, which with $FRG'$ implies $\neg (a_1\wedge\dots\wedge a_n)\in G'$. But  $a_1,\dots,a_n\in G$ implies $a_1\wedge\dots\wedge a_n\in G$, so $\neg (a_1\wedge\dots\wedge a_n)\in G'$  contradicts $G'\between G$. Thus, we have $F'\between F$. Now consider any $F''\between F'$. We claim that $i(F'')\between G$. If not, then there is some $a\in G$ such that $\neg a\in i(F'')$. But $a\in G$ implies $\Diamond a\in F'$, and  $\neg a\in i(F'')$ implies $\Box\neg a\in F''$, and together $\Diamond a\in F'$ and $\Box\neg a\in F''$ contradict $F''\between F'$. Thus, $ i(F'')\between G$, so setting $G''= i(F'')$ completes the proof of $R$-\textsf{regularity}.

Finally, we show that the frame satisfies \textsf{Knowability}: for all $F\in S$, there is some $G\in S$ such that for all $H\in R(G)$, we have $H\sqsubseteq F$.  Given $F$, let $G$ be the filter generated by $\{\Box a\mid a\in F\}$. We claim that $G$ is proper. If not, then there are $a_1,\dots,a_n\in F$ such that $\Box a_1\wedge\dots\wedge \Box a_n=0$ and hence $\Box (a_1\wedge\dots\wedge a_n)=0$. Hence $\Diamond \neg (a_1\wedge\dots\wedge a_n)=1$, so $\Diamond \neg (a_1\wedge\dots\wedge a_n)\in F$. From  $a_1,\dots,a_n\in F$ we also have $a_1\wedge\dots\wedge a_n\in F$. But since $L$ is an epistemic ortholattice, from $\Diamond \neg (a_1\wedge\dots\wedge a_n)\in F$ and $a_1\wedge\dots\wedge a_n\in F$, we have $0\in F$, contradicting our assumption that $F$ is proper. Hence $G$ is proper. To show $H\sqsubseteq F$, suppose $I\between H$ but not $I\between F$, so there is some $ b\in I$ such that $\neg b \in F$. Since $\neg b\in F$, $\Box\neg b\in G$, which with $GRH$ implies $\neg b\in H$, which contradicts $I\between H$.

Thus, $\mathcal{F}=\langle S,\between,R\rangle \in \mathbf{ECF}$. It is well known that the map sending $a\in L$ to $\widehat{a}=\{F\in S\mid a\in F\}$ is an embedding of the ortholattice reduct of $L$ into the ortholattice of $\between$-regular subsets of  $\langle S,\between\rangle$ (see, e.g., \citealt[Proposition 1]{Goldblatt1975}).\footnote{Goldblatt \citeyearpar{Goldblatt1975} works with the complement of the compatibility relation, called the orthogonality relation, but the proof is easily rephrased in terms of compatibility.} It only remains to observe that this embedding also respects $\Box$: \[\widehat{\Box a} = \{F\in S\mid \Box a\in F\}= \{F\in S\mid a\in i(F)\}= \{F\in S\mid R(F)\subseteq\widehat{a}\}=\Box \widehat{a}.\]
Thus, the map $a\mapsto\widehat{a}$ is a modal ortholattice embedding of $L$ into $O(\mathcal{F})$.  Let $\mathcal{M}$ be the model based on $\mathcal{F}$ with the valuation $V$ defined by $V(p)=\widehat{[p]}$. Then since $a\mapsto\widehat{a}$ is a modal ortholattice embedding, we have $\llbracket \varphi\rrbracket^\mathcal{M}=\widehat{[\varphi]}$ for all $\varphi\in\mathcal{EL}$. Now if $\varphi \nvdash_\mathsf{EO}\psi$, then in $L$ we have $[\varphi]\not\leq [\psi]$. Then since $a\mapsto\widehat{a}$ is a modal ortholattice embedding, it follows that $\widehat{[\varphi]}\not\subseteq \widehat{[\psi]}$ and hence  $\llbracket \varphi\rrbracket^\mathcal{M}\not\subseteq \llbracket \psi \rrbracket^\mathcal{M}$. Therefore, $\varphi\nvDash_\mathbf{ECF} \psi$.\end{proof}

\begin{remark} The proof of Theorem \ref{EOThm} shows that instead of giving semantics for $\Box$ using an accessibility relation $R\subseteq S\times S$, we could give a semantics using an accessibility \textit{function} $i:S\to S$ such that 
\[\Box A=\{x\in S\mid i(x)\in A\}.\]
This semantics simplifies some abstract calculations but has the disadvantage of requiring us to add more possibilities to frames such as the frame in Figure \ref{EpCompFig} in order to do the work of a set of epistemically accessible possibilities using a single coarse-grained possibility (in the case of the frame in Figure \ref{EpCompFig}, the functional approach requires us to add two possibilities to the frame).
\end{remark}

\subsection{Distinguishing non-epistemic propositions}\label{DistinguishSec2}

Just as we distinguished non-epistemic propositions from arbitrary propositions in the context of algebraic semantics in \S~\ref{DistinguishSec}, we can now do the same in the context of possibility semantics.\footnote{Equipping a frame with a distinguished collection of propositions is similar in spirit to introducing \textit{general frames} in modal logic (\citealt[\S~5.5]{Blackburn2001}). Unlike with general frames,  we do not require that the distinguished collection be closed under $\Box$. However,  we will have a different requirement involving $\Box$ in Definition \ref{GroundedEpistemicFrame}.}

\begin{proposition}\label{StratProp} Given a compatibility frame $\mathcal{F}=\langle S,\between,R\rangle$, let  $\mathbb{B}$ be a nonempty set of $\between$-regular subsets of $S$ closed under $\cap$ and the operation $\neg$ from (\ref{NegEq}). Then the following are equivalent:
\begin{enumerate}
\item $\mathbb{B}$ is a Boolean algebra under $\cap$ and $\neg$;
\item\label{StratProp2} for all $A,B\in\mathbb{B}$, if there are $x\in A$ and $y\in B$ with $x\between y$, then $A\cap B\neq\varnothing$.
\end{enumerate}
\end{proposition}
\begin{proof} By Proposition \ref{PseudoToBoole}, $\mathbb{B}$ being Boolean is equivalent to the condition that  $\neg$ restricted to $\mathbb{B}$ is pseudocomplementation, i.e., that for all $A,B\in \mathbb{B}$,  if $A\cap B=\varnothing$, then $A\subseteq\neg B$;  and $A\subseteq \neg B$ is equivalent to there being no $x\in A$ and $y\in B$ with $x\between y$. Contraposing, we obtain the condition in part \ref{StratProp2}.\end{proof}

\noindent Thus, the key condition on $\mathbb{B}$ is that if there are compatible possibilities making $A$ and $B$ true, respectively, then there is a single possibility making both $A$ and $B$ true (cf.~the even stronger classical idea that compatibility implies compossibility in Remark \ref{BooleanCase}).

\begin{definition}\label{GroundedEpistemicFrame} A \textit{grounded modal compatibility frame} is a tuple $\mathcal{F}=\langle S,\between,R,\mathbb{B} \rangle$ where $\langle S,\between,R\rangle$ is a modal compatibility frame and $\mathbb{B}$ is a nonempty collection of $\between$-regular sets closed under $\cap$ and the operation $\neg$ from (\ref{NegEq}). Given a grounded frame $\mathcal{F}=\langle S,\between,R,\mathbb{B}\rangle $, 
\begin{itemize}
\item let $\mathbb{B}_0=\mathbb{B}$, and
\item let  $\mathbb{B}_{n+1}$ be the closure of  $\{\Box B\mid B\in \mathbb{B}_n\}$ under $\cap$ and $\neg$,
\end{itemize}
where $\Box$ is the operation defined from $R$ in (\ref{BoxEq1}). We say that $\mathcal{F}$ is \textit{stratified} if for all $n\in\mathbb{N}$ and  $A,B\in\mathbb{B}_n$, if there are $x\in A$ and $y\in B$ with $x\between y$, then $A\cap B\neq\varnothing$.

A \textit{stratified epistemic compatibility frame} is a stratified modal compatibility frame in which $\langle S,\between,R\rangle$ is an epistemic compatibility frame.\end{definition}

Then the following is immediate from Proposition \ref{StratProp}.

\begin{corollary} In a stratified modal compatibility frame, each $\mathbb{B}_n$ is a Boolean algebra.
\end{corollary}

\begin{example} It is easy to check by hand that the epistemic compatibility frame in Figure \ref{EpCompFig} equipped with $\mathbb{B}=\{\varnothing, \{x_1,x_2\}, \{x_4,x_5\},  S\}$ is stratified. In this case $\mathbb{B}_1$ has eight propositions, and $\mathbb{B}_2=\mathbb{B}_1$, so $\mathbb{B}_n=\mathbb{B}_1$ for all $n\geq 1$; this matches what we observed in Example \ref{LevelwiseBooleEx} for the corresponding epistemic ortho-Boolean lattice. 
\end{example}

Recall from \S~\ref{DistinguishSec} the distinguished set \textsf{Bool} of non-epistemic propositional variables. These are now interpreted as propositions in the distinguished family $\mathbb{B}$.

\begin{definition}\label{GroundedModel} A \textit{grounded modal compatibility model} is a pair $\mathcal{M}=\langle \mathcal{F},V\rangle$ where $\mathcal{F}=\langle S,\between, R, \mathbb{B}\rangle$ is a grounded modal compatibility frame and $V$ assigns to each $p\in\mathsf{Prop}$ a $\between$-regular  $V(p)\subseteq S$ and to each $\mathtt{p}\in\mathsf{Bool}$ a $V(\mathtt{p})\in\mathbb{B}$. We say that $\mathcal{M}$ is \textit{based on} $\mathcal{F}$.\end{definition}

Next comes the definition of consequence, which follows exactly the same pattern as before.

\begin{definition}\label{PossCon} Given a class $\mathbf{C}$ of grounded modal compatibility frames, define the semantic consequence relation $\vDash_\mathbf{C}$, a binary relation on $\mathcal{EL}^+$, as follows: $\varphi\vDash_\mathbf{C}\psi$ if for every  $\mathcal{F}\in\mathbf{C}$, model $\mathcal{M}$ based on $\mathcal{F}$, and possibility $x$ in $\mathcal{M}$, if $\mathcal{M},x\Vdash\varphi$, then $\mathcal{M},x\Vdash\psi$.
\end{definition}

Finally, we can prove that our epistemic orthologic $\textsf{EO}^+$ in Figure \ref{EO+Fig} is sound and complete with respect to possibility semantics using stratified epistemic compatibility frames.

\begin{theorem}\label{EO+Comp} The logic $\textsf{EO}^+$ is sound and complete with respect to the class $\mathbf{ECF}^+$ of all stratified epistemic compatibility frames according to the consequence relation of Definition \ref{PossCon}: for all $\varphi,\psi\in\mathcal{EL}^+$, we have $\varphi\vdash_{\mathsf{EO}^+}\psi$ if and only if $\varphi\vDash_{\mathbf{ECF}^+}\psi$.
\end{theorem}
\begin{proof} Simply add to the proof of Theorem \ref{EOThm} that $\mathbb{B}=\{\widehat{[\beta]}\mid \beta\mbox{ Boolean}\}$. As in the proof of Theorem \ref{EO+comp1}, each $B_n$ is a Boolean subalgebra of the Lindenbaum-Tarski algebra, which implies that $\mathbb{B}_n$ satisfies the key condition in Definition \ref{GroundedEpistemicFrame}: in contrapositive form, if there is no proper filter $F\in \widehat{[\alpha]}\cap \widehat{[\beta]}$, i.e., no proper filter $F$ with $[\alpha]\wedge[\beta]\in F$, so $[\alpha]\wedge[\beta]=0$ in the Lindenbaum-Tarski algebra, then since $[\alpha]$ and $[\beta]$ belong to a Boolean subalgebra, we have $[\alpha]\leq \neg [\beta]$, which implies that there are no proper filters $F\in\widehat{[\alpha]}$ and $G\in \widehat{[\beta]}$ with $F\between G$.\end{proof}

\noindent Thus, our logic \textsf{EO}$^+$   is the logic not only of a certain class of algebras but also of a concrete possibility semantics generalizing possible world semantics. 

\section{Constructing possibilities from worlds}\label{Epistemicization}\label{EpExtSection}

We have now given an algebraic semantics and a possibility semantics that both correspond to our logic \textsf{EO}$^+$. 
From one perspective, our job is now done twice over: for those who might worry that algebraic semantics is suspiciously close to the corresponding syntactic principles, we have given a concrete implementation in terms of truth at possibilities. 
However, some might find it hard to theorize directly in terms of possibilities. Possible worlds are generally more familiar to  semanticists and logicians, and they have been much more thoroughly studied by philosophers. 

In order to counteract any squeamishness about possibilities, we will show in this section that we can {construct} a model for our possibility semantics from any given possible worlds model. This construction should put to rest any worries that readers might have about working with a framework that rests on our non-standard semantic foundation. It shows that our approach is ontologically innocent: if you are comfortable with possible worlds but not possibilities, you can build the whole system on that foundation. (Of course, there may be ontological reasons to prefer possibilities, but we will not make that argument here.) Moreover, by giving a concrete way for those familiar with possible worlds semantics to build 
models in our semantic framework, this construction will enable semanticists comfortable with possible worlds to easily build  models for our possibility semantics.
More broadly, it gives us a  scaffolding for moving from structures defined on  possible worlds  to structures defined on possibilities. In subsequent work, we will show how to use this scaffolding to lift  probability functions and conditional selection functions from worlds to possibilities. More generally, the goal is to develop a technique to nimbly take us from the familiar setting of possible worlds semantics into a possibility semantics framework more suitable for theorizing about epistemic modality, showing how our  possibility semantics can capture both the standard phenomena covered by possible world semantics and the peculiar behavior of epistemic modals.

In more detail, given a possible worlds model, we will show how to transform that model into one of our epistemic compatibility models. As we will see, the resulting models validate a strict strengthening of \textsf{EO}$^+$; in particular,  they commit us to the \textsf{S5} axioms in addition to \textsf{EO}$^+$, so not every epistemic compatibility model can be represented as coming from a possible worlds model. The goal of this construction is to show just one way to obtain a possibility semantics from a non-modal starting point. 

The basic idea behind the construction, starting with a set $W$ of worlds, is to \[\mbox{construct possibilities as \textit{pairs $(A,I)$ of sets of worlds}  where $\varnothing \neq A\subseteq I\subseteq W$.}\]
In fact, we show 
 that our construction applies starting with an arbitrary Boolean algebra $B$, but in the finite case we can assume without loss of generality that $B$ is the powerset of a set of worlds.

Given a Boolean algebra $B$ of non-modal propositions, we wish to ``add modal propositions'' to $B$, resulting in an epistemic ortholattice that we will call the \textit{epistemic extension} of $B$. We will define the epistemic extension of $B$ as the epistemic ortholattice coming from an epistemic compatibility frame constructed from~$B$. 

\begin{definition}\label{EpistemicFrameDef} Let $B$ be a Boolean algebra. The \textit{epistemic frame of $B$} is the tuple $B^\mathsf{e}=(S,\between, R)$ defined as follows:
\begin{enumerate}
\item\label{EpistemicFrameDef1} $S=\{(a,i)\mid a,i\in B, 0\neq a\leq i\}$;
\item\label{EpistemicFrameDef2} $(a,i)\between (a',i')$ iff $a\wedge a'\neq 0$ and $a\leq i'$ and $a'\leq i$;
\item\label{EpistemicFrameDef3} $(a,i)R(a',i')$ iff $a\leq a'$ and $i'\leq i$.
\end{enumerate}
Given a valuation $\theta:\mathsf{Bool}\to B$, we define  $\theta^\mathsf{e}$ by $\theta^\mathsf{e}(\mathtt{p})=\{(a,i)\mid a\leq \theta(\mathtt{p})\}$. The \textit{epistemic model of} $(B,\theta)$ is the pair $(B^\mathsf{e},\theta^\mathsf{e})$.
\end{definition}

Thus, the possibilities in the epistemic frame of $B$ are pairs of propositions. As noted above, if $B$ is the powerset of a set $W$ of  worlds, the possibilities are pairs $(A,I)$ where $\varnothing \neq A\subseteq I\subseteq W$. Pairs of the form $(\{w\},I)$ are familiar from domain semantics for epistemic modals (e.g., \citealt{Yalcin2007}, \citealt{MacFarlane:2011}), where $w\in W$ is a world and $I\subseteq W$ is an information state. 
But unlike in domain semantics,  our $A$ need not be a singleton set, so we have in effect ``partialized'' or ``possibilized'' the worldly component of the familiar picture.
Moreover, unlike in domain semantics, where there is a clear division of labor between the first component, which determines truth-values for Boolean sentences, and the second component, which determines truth-values for modal sentences, in the present framework, \emph{both} components contribute to the evaluation of modals:  the second component $I$ determines what Boolean propositions \emph{must} be the case, while  what Boolean propositions \textit{might} be the case is determined by the first component $A$, which \emph{also} determines what Boolean propositions \textit{are} the case (see Lemmas \ref{SimpleTruth} and \ref{AlgSimpleTruth}). The dual roles of the first component are essential to  how this approach, unlike domain semantics, validates Wittgenstein's Law: if $\beta$ is determined  to be true by $A$, then $\neg\beta$ is not compatible with $A$ and so $\lozenge\neg\beta$ cannot be true. But the second component is also required to determine what must be the case: just because something is settled true by $A$ does not mean that it is \emph{known} to be true. 

More carefully, any Boolean $\beta$ \textit{entailed} by the first component of a possibility (i.e., $A\subseteq \tilde{\theta}(\beta)$, when $B$ is a powerset) \textit{is} true, any Boolean $\beta$ \textit{consistent} with the first component (i.e., $A\cap \tilde{\theta}(\beta)\neq \varnothing$)  \textit{might} be true, and any Boolean $\beta$ \textit{entailed} by the second component (i.e., $I\subseteq \tilde{\theta}(\beta)$) \textit{must} be true. This helps explain the definitions of $\between$ and $R$. The definition of $\between$ says that two possibilities are compatible if (i) how things are and might be according to each possibility is consistent with how things are and might be according to the other ($A\cap A'\neq \varnothing$) and (ii) how things are and might be according to each possibility is a narrowing down of how things must be according to the other ($A\subseteq I'$ and $A'\subseteq I$). While (i) ensures that if one possibility settles that $\beta$ is true, then a compatible possibility does not settle that $\neg \beta$ is true, (ii) ensures that if one possibility settles that $\beta$ \textit{might} be true (so there is a $\beta$-world in $A$), then a compatible possibility does not settle that $\neg \beta$ \textit{must} be true (since the $\beta$-world in $A$ belongs to $I'$). According to the definition of $R$, one possibility has epistemic access to another iff both components of the second possibility lie in the interval $\{X\subseteq W\mid A\subseteq X\subseteq I\}$ between the two components of the first. This ensures that whatever $\beta$ might be true according to the first possibility also might be true according to the second (since $A\subseteq A'$), and whatever $\beta$  must be true according to the first possibility also must be true according to the second (since $I'\subseteq I$).

We can also give a simple characterization of refinement (recall Lemma \ref{RefinementDef}) in the epistemic frame of $B$.

\begin{lemma}\label{SimpleRefinement} Given a Boolean algebra $B$ and possibilities $(a,i)$ and $(a',i')$ in $B^\mathsf{e}$, we have $(a,i)\sqsubseteq (a',i')$ iff $a=a'$ and $i\leq i'$.
\end{lemma}
\begin{proof} For the right-to-left direction,  suppose $a=a'$, $i\leq i'$, and $(a'',i'')\between (a,i)$. Then since $a=a'$ and $i\leq i'$,  we have $(a'',i'')\between (a',i')$. Hence $(a,i)\sqsubseteq (a',i')$. From left to right, if $a\not\leq a'$, then $(a\wedge\neg a',a)\between (a,i)$ but not $(a\wedge\neg a',a)\between (a',i')$, so $(a,i)\not\sqsubseteq (a',i')$. If $a'\not\leq a$, then $(a,a)\between (a,i)$ but not $(a,a)\between (a',i')$, so $(a,i)\not\sqsubseteq (a',i')$. Finally, if $i\not\leq i'$, then $(i,i)\between (a,i)$ but not $(i,i)\between (a',i')$, so  $(a,i)\not\sqsubseteq (a',i')$.\end{proof}

Below we will prove that $B^\mathsf{e}$ is indeed an epistemic compatibility frame (Definition \ref{Knowability}), but first let us consider some concrete examples. Given a set $W$ of worlds, we consider the epistemic frame of $\wp(W)$.

\begin{example}\label{TwoWorlds} Suppose we are about to flip a fair coin, so $W=\{0,1\}$. Figure \ref{FrameFromTwoWorlds} shows the epistemic frame of $\wp(W)$. Note that it is isomorphic to the frame underlying the Epistemic Scale in Figure \ref{EpCompModFig}. As before, the solid  lines represent compatibility (with reflexive loops omitted), and the dotted lines represent epistemic accessibility (with reflexive loops omitted). 
 For example, according to Definition \ref{EpistemicFrameDef}.\ref{EpistemicFrameDef2}, the possibility $x_1=(\{0\}{,}\,\{0\})$ is compatible with $x_2=(\{0\}{,}\,\{0{,}\,1\})$ because their first-coordinates have a non-empty intersection, and the first coordinate of each entails the second coordinate of the other. However, $x_1$ is not compatible with $x_3=(\{0,\,1\}{,}\,\{0{,}\,1\})$, since the first coordinate of $x_3$ does not entail the second coordinate of $x_1$. As for accessibility, according to Definition \ref{EpistemicFrameDef}.\ref{EpistemicFrameDef3}, $x_2$ can access $x_3$ because both coordinates of $x_3$ lie in the interval between the coordinates of $x_2$. However, $x_3$ cannot access $x_2$ because the first coordinate of $x_2$ does not lie in the interval between the coordinates of $x_3$. 
 
\begin{figure}[h]
\begin{center}
\tikzset{every loop/.style={min distance=10mm,looseness=10}}
\begin{tikzpicture}[->,>=stealth',shorten >=1pt,shorten <=1pt, auto,node
distance=2.5cm,semithick]
\tikzstyle{every state}=[fill=gray!20,draw=none,text=black]

\node (1) at (0,0) {{(\{0\}{,}\,\{0\})}};
\node (2) at (3,0) {{$(\{0\}{,}\,\{0{,}\,1\})$}};
\node  (3) at (6,0) {{$(\{0{,}\,1\}{,}\,\{0{,}\,1\})$}};
\node  (4) at (9,0) {{$(\{1\}{,}\,\{0{,}\,1\})$}};
\node (5) at (12,0) {{$(\{1\}{,}\,\{1\})$}};
\node (6) at (0,-.5) {{$x_1$}};
\node (7) at (3,-.5) {{$x_2$}};
\node (8) at (6,-.5) {{$x_3$}};
\node (9) at (9,-.5) {{$x_4$}};
\node (10) at (12,-.5) {{$x_5$}};

\path (1) edge[-] node {{}} (2);
\path (2) edge[-] node {{}} (3);
\path (3) edge[-] node {{}} (4);
\path (4) edge[-] node {{}} (5);

\path (2) edge[dotted,thick,bend right, MidnightBlue] node {{}} (1);
\path (2) edge[dotted,thick,bend left, MidnightBlue] node {{}} (3);

\path (4) edge[dotted,thick,bend right, MidnightBlue] node {{}} (3);
\path (4) edge[dotted,thick,bend left, MidnightBlue] node {{}} (5);

\end{tikzpicture} 
\end{center}
\caption{Epistemic frame constructed from two worlds.}\label{FrameFromTwoWorlds}
\end{figure}
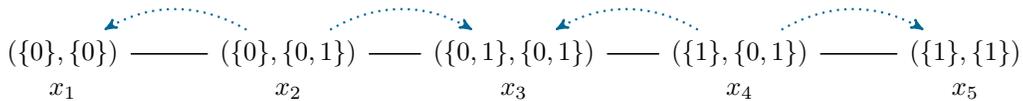
\end{example}

An especially appealing aspect of this construction is how simple it is to check the truth values of formulas of modal depth $\leq 1$ at a possibility $(A,I)$, as shown by the following lemma. Given a possible worlds model $\mathfrak{M}=(W,V)$, where $W$ is a nonempty set and $V:\mathsf{Bool}\to \wp(W)$, and a Boolean formula $\varphi$ (i.e., one free of modals, as in Definition \ref{ELPlus}), define $\mathfrak{M},w\vDash\varphi$ as usual: $\mathfrak{M},w\vDash \mathtt{p}$ iff $w\in V(\mathtt{p})$; $\mathfrak{M},w\vDash \neg\varphi$ iff $\mathfrak{M},w\nvDash\varphi$; and $\mathfrak{M},w\vDash \varphi\wedge\psi$ iff $\mathfrak{M},w\vDash\varphi$ and $\mathfrak{M},w\vDash\psi$.

\begin{lemma}\label{SimpleTruth} Let $\mathfrak{M}=(W,V)$ be a possible worlds model and $\mathcal{M}$ the epistemic model of $(\wp(W),V)$. For any Boolean formula $\varphi$ and $(A,I)$ in $\mathcal{M}$, we have:
\begin{enumerate}
\item\label{SimpleTruth1} $\mathcal{M},(A,I)\Vdash \varphi$ iff for all $w\in A$, we have $\mathfrak{M},w\vDash \varphi$;
\item\label{SimpleTruth2} $\mathcal{M},(A,I)\Vdash \Box\varphi$ iff for all $w\in I$, we have $\mathfrak{M},w\vDash \varphi$;
\item\label{SimpleTruth3} $\mathcal{M},(A,I)\Vdash \Diamond\varphi$ iff for some $w\in A$, we have $\mathfrak{M},w\vDash \varphi$.
\end{enumerate}
\end{lemma}
\begin{proof} We prove part \ref{SimpleTruth1} by induction on $\varphi$. For the base case of a propositional variable $\mathtt{p}$, we have ${\mathcal{M},(A,I)\Vdash \mathtt{p}}$ iff $(A,I)\in V^\mathsf{e}(\mathtt{p})$ iff $A\subseteq V(\mathtt{p})$ iff for all $w\in A$, $\mathfrak{M},w\vDash \mathtt{p}$. The inductive step for $\wedge$ is immediate from the inductive hypothesis. For the inductive step for $\neg$, we have 
\begin{eqnarray*}
\mathcal{M},(A,I)\Vdash\neg\varphi &\mbox{ iff }& \forall (A',I')\between (A,I), \,\mathcal{M},(A',I')\nVdash\varphi\\
&\mbox{ iff }& \forall (A',I')\between (A,I)\; \exists w\in A': \mathfrak{M},w\nvDash\varphi\mbox{ by the inductive hypothesis}\\
&\mbox{ iff }& \forall w\in A,\, \mathfrak{M},w\nvDash\varphi.
\end{eqnarray*}
The implication from the second to third line follows from the fact that for each $w\in A$, we have that $(\{w\},I)\between (A,I)$. The converse implication follows from the fact that $ (A',I')\between (A,I)$ implies $A'\cap A\neq\varnothing$.

For part \ref{SimpleTruth2},  we have
\begin{eqnarray*}
\mathcal{M},(A,I)\Vdash \Box\varphi & \mbox{ iff }&\forall (A',I')\in R((A,I)),\, \mathcal{M},(A',I')\Vdash \varphi \\
&\mbox{ iff }&\forall (A',I')\in R((A,I))\; \forall w\in A',\, \mathfrak{M},w\vDash \varphi \mbox{ by part \ref{SimpleTruth1}}\\
&\mbox{ iff }& \forall w\in I,\, \mathfrak{M},w\vDash \varphi.
\end{eqnarray*}
The implication from the second to third line follows from the fact that $(A,I)R(I,I)$. The converse implication follows from the fact that $(A,I)R(A',I')$ implies $A'\subseteq I'\subseteq I$.

For part \ref{SimpleTruth3}, we have 
\begin{eqnarray*}
\mathcal{M},(A,I)\Vdash \Diamond\varphi & \mbox{ iff }&\forall (A',I')\between (A,I) \;\exists (A'',I'')\in R((A',I'))  \\
&&\qquad\qquad\qquad\quad\;\,\exists (A''',I''')\between (A'',I''): \mathcal{M},(A''',I''')\Vdash \varphi \\
&\mbox{ iff }&\forall (A',I')\between (A,I) \;\exists (A'',I'')\in R((A',I'))  \\
&&\qquad\qquad\qquad\quad\;\,\exists (A''',I''')\between (A'',I'')\;\forall w\in A''', \mathfrak{M},w\vDash \varphi \mbox{ by part \ref{SimpleTruth1}}\\
&\mbox{ iff }& \exists w\in A: \mathfrak{M},w\vDash \varphi.
\end{eqnarray*}
For the implication from the second to third line, setting $(A',I')=(A,A)$, we have $(A',I')\between (A,I)$, and then from $(A',I')R(A'',I'')$ we have $A''\subseteq I''\subseteq I'\subseteq A$. Then from $(A''',I''')\between (A'',I'')$, we have $A'''\cap A''\neq\varnothing$ and hence $A'''\cap A\neq\varnothing$. Then since $\forall w\in A''', \mathfrak{M},w\vDash \varphi $, we have $\exists w\in A: \mathfrak{M},w\vDash \varphi$. For the converse implication, suppose $w\in A$, $\mathfrak{M},w\vDash\varphi$, and $(A',I')\between (A,I)$. It follows that $w\in I'$. Then take $(A'',I'')=(A'\cup\{w\},I')$, so $(A',I')R(A'',I'')$. Finally, taking $(A''',I''')=(\{w\},I')$, we have $(A''',I''')\between (A'',I'')$ and $\forall w\in A'''$, $\mathfrak{M},w\vDash\varphi$, which completes the proof.\end{proof}

To see how Lemma \ref{SimpleTruth} facilitates computations of truth values, let us consider a richer example starting from three worlds.

\begin{example}\label{ThreeWorlds} Suppose we are about to roll a three-sided die, so $W=\{0,1,2\}$. Figure \ref{FrameFromThreeWorlds} shows the epistemic frame of $\wp(W)$. Figure \ref{FrameFromThreeHighlighted} then shows the propositions expressed by several formulas, highlighted in green; the propositional variables $\mathsf{0}$, $\mathsf{1}$, and $\mathsf{2}$ are true in all possibilities $(A,I)$ such that $A=\{0\}$, $A=\{1\}$, and $A=\{2\}$, respectively. Note the ease of computing the modal propositions using Lemma \ref{SimpleTruth}.

\definecolor{Darklavender}{rgb}{0.45, 0.31, 0.59}

\begin{figure}
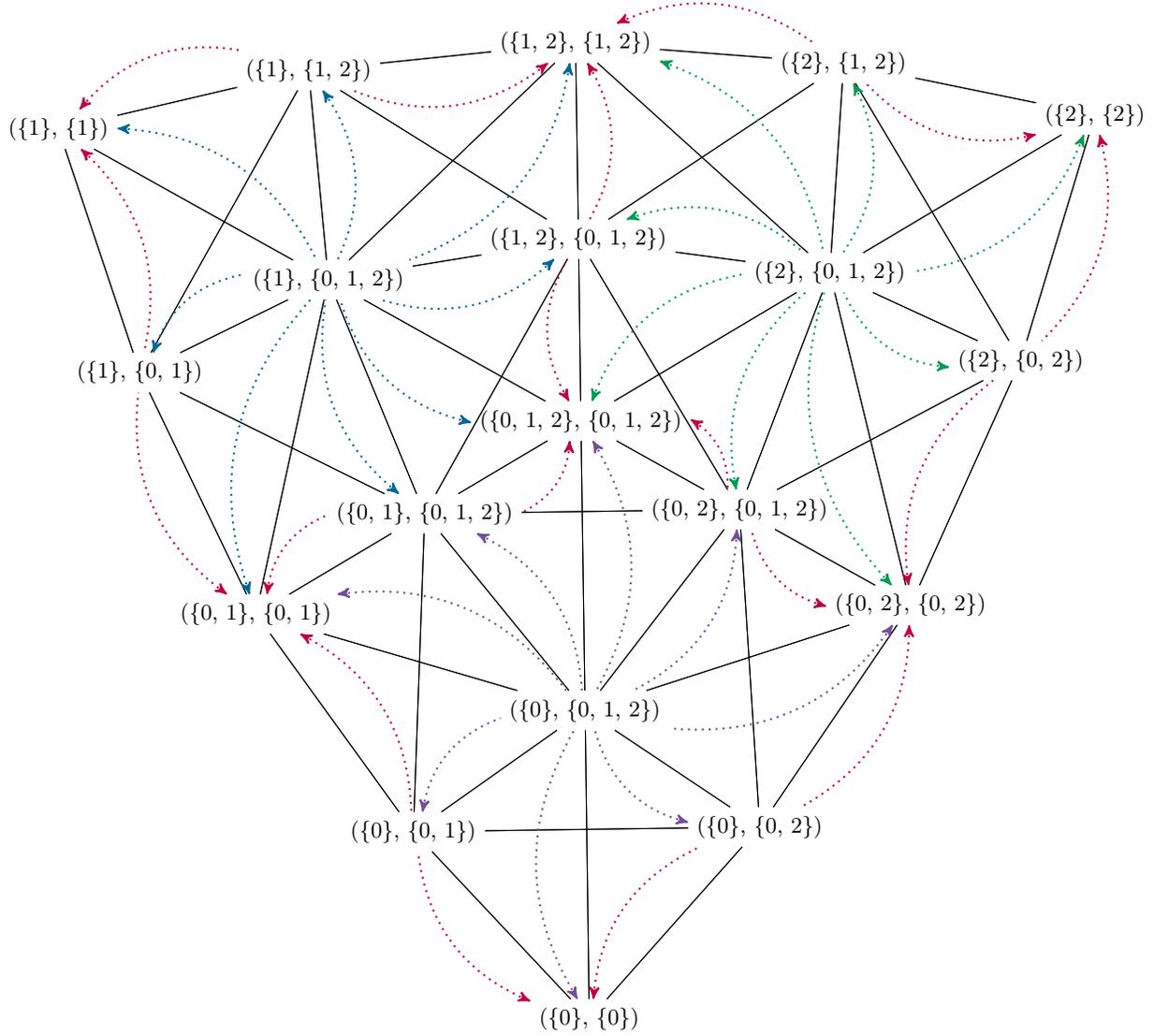


    \caption{Epistemic frame constructed from three worlds}\label{FrameFromThreeWorlds}
\end{figure}

\begin{figure}
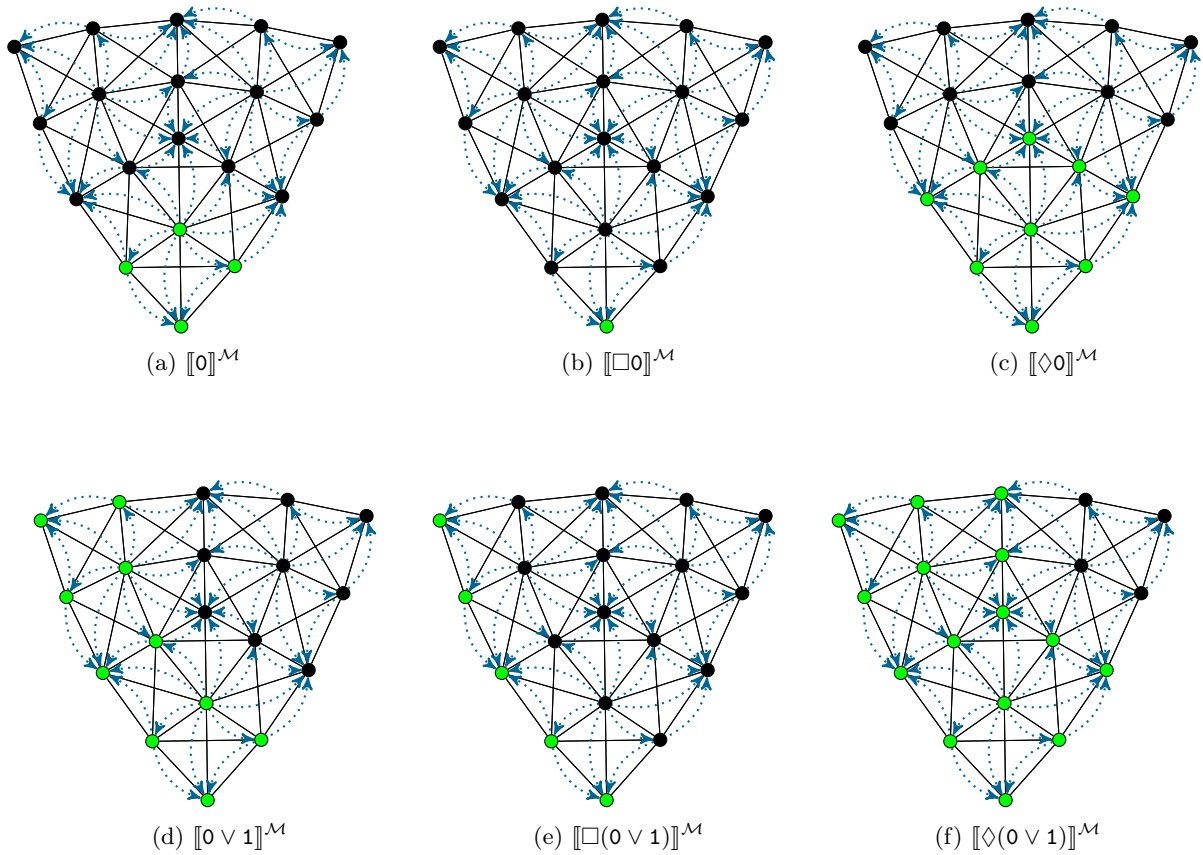
\centering\begin{subfigure}{0.3\textwidth}
  %
    \caption{$\llbracket \mathtt{0}\rrbracket^\mathcal{M}$}
    \end{subfigure}\qquad\begin{subfigure}{0.3\textwidth}
  %
    \caption{$\llbracket \Box\mathtt{0}\rrbracket^\mathcal{M}$}
    \end{subfigure}\qquad\begin{subfigure}{0.3\textwidth}
  %
    \caption{$\llbracket \Diamond\mathtt{0}\rrbracket^\mathcal{M}$}
    \end{subfigure}\vspace{.5in}
    
    \begin{subfigure}{0.3\textwidth}
  %
    \caption{$\llbracket \mathtt{0}\vee\mathtt{1}\rrbracket^\mathcal{M}$}
    \end{subfigure}\quad\begin{subfigure}{0.3\textwidth}
  %
    \caption{$\llbracket \Box(\mathtt{0}\vee\mathtt{1})\rrbracket^\mathcal{M}$}
    \end{subfigure}\quad\begin{subfigure}{0.3\textwidth}
  %
    \caption{$\llbracket \Diamond(\mathtt{0}\vee\mathtt{1})\rrbracket^\mathcal{M}$}
    \end{subfigure}

     \caption{Propositions highlighted in copies of the epistemic model constructed from the possible worlds model with $W=\{0,1,2\}$, $V(\mathtt{0})=\{0\}$, $V(\mathtt{1})=\{1\}$, and $V(\mathtt{2})=\{2\}$}\label{FrameFromThreeHighlighted}
\end{figure}
\end{example}
Finally, let us push one step further to the case of four worlds.

\begin{example} Starting with $W=\{0,1,2,3\}$ is especially interesting, since we can capture all the possible truth value combinations for two Boolean propositional variables $\mathtt{p}$ and $\mathtt{q}$ with $V(\mathtt{p})=\{0,1\}$ and $V(\mathtt{q})=\{0,3\}$. It is harder to cleanly draw the epistemic frame of $\wp(W)$ in this case, but Figure \ref{FrameFromFourWorlds} provides a helpful three-dimensional tetrahedral visualization of the frame.
\begin{figure}[h]
\begin{center}
\includegraphics[scale=.3]{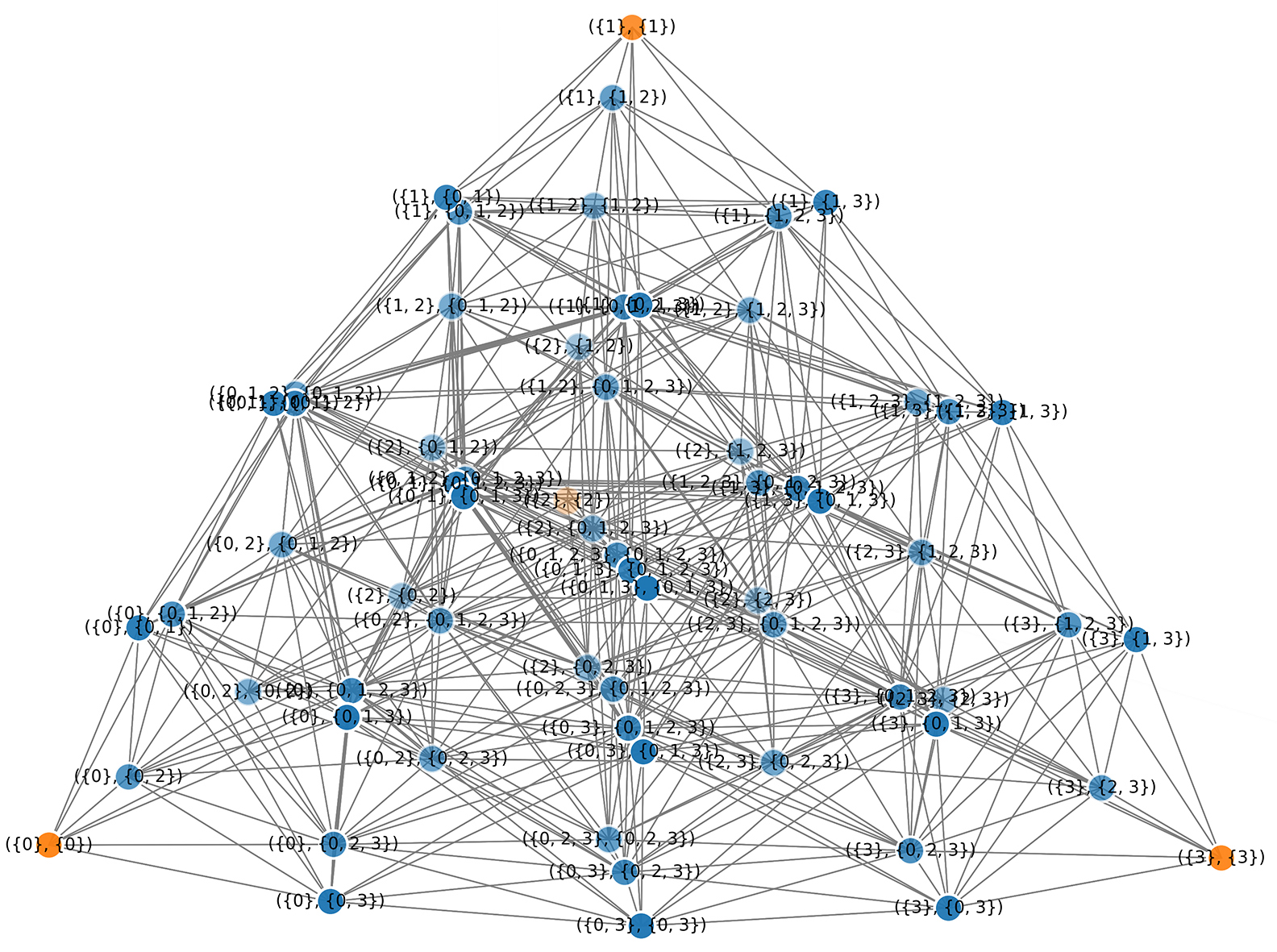}
\end{center}
\caption{Epistemic frame constructed from four worlds with fully determinate possibilities highlighted. To reduce clutter, only compatibility relations are shown.}\label{FrameFromFourWorlds}
\end{figure} 
\end{example}

Table \ref{SizeTable} shows the increase in the size of the epistemic frames as we increase the initial number of worlds. Visualization beyond four worlds becomes difficult, but our computer implementation of the construction can comfortably handle more than four initial worlds (though calculating the size of the ortholattice coming from the epistemic frame becomes prohibitive without an analytic expression for this size).

\begin{table}[h]
\begin{center}
\begin{tabular}{c|c|c|c}
 size of   $B$ & world representation of $B$ & epistemic frame of $B$ & epistemic extension of $B$ \\

\hline
4&2 & 5 & 10 \\
8&3 & 19 & 302 \\
16&4 & 65 & 298,414 \\
32&5 & 211 & ?  \\
64&6 & 665 & ?  \\
\end{tabular}
\end{center}
\caption{Sizes of objects associated with a Boolean algebra $B$.}\label{SizeTable}
\end{table}

Having seen the epistemic frames constructed from three Boolean algebras, we now prove the first main result of this section. We show that the epistemic frame of $B$ is an epistemic compatibility frame (Definition \ref{Knowability}) into whose associated epistemic ortholattice $B$ embeds. The bare fact that $B$ embeds into an (S5) epistemic ortholattice is trivial, since $B$ itself equipped with $\Box$ as the identity function---and hence where $\Diamond$ is the identity function---is an (S5) epistemic ortholattice. But what we want is to embed $B$ into an epistemic ortholattice where $\Diamond$ does not collapse, as in part \ref{NoMightCollapse} of the following.

\begin{theorem}\label{EpistemicFrameExtThm} For any Boolean algebra $B$ with lattice order $\leq$:
\begin{enumerate}
\item\label{EpistemicFrame} $B^\mathsf{e}$ is an epistemic compatibility frame;
\item\label{EpistemicExt} the map $e$ defined by $e_B(a)= \{(b,i)\in S\mid b\leq a \}$ is an embedding of $B$ into the complete epistemic ortholattice $O(B^\mathsf{e})$, which we therefore call the \textit{epistemic extension of}  $B$;
\item\label{EpistemicS5} $O(B^\mathsf{e})$ is an S5 epistemic ortholattice;
\item\label{NoMightCollapse} for all $b\in B$, if $b\not\in \{0,1\}$, then $\Diamond e_B(b) \not\leq e_B(b)$ in $O(B^\mathsf{e})$.
\end{enumerate}
\end{theorem}

\begin{proof} For part \ref{EpistemicFrame}, clearly $\between$ is reflexive and symmetric, so we proceed to the properties of $R$. 

For $R$\textsf{-regularity} in its form in Lemma \ref{RregEquiv}, suppose $(a,i)R(b,k)\between (d,l)$.  We claim that the desired witness for $R$\textsf{-regularity} is $(a',i')=(a\vee d,a\vee d)$. First, we claim $(a,i)\between (a\vee d,a\vee d)$. The first two conditions, namely $a\wedge (a\vee d)\neq 0$ and $a\leq a\vee d$, are immediate. As for $a\vee d\leq i$, from $(a,i)\in S$ we have $a\leq i$, and from $(a,i)R(b,k)\between (d,l)$ we have $i\geq k\geq d$.  Now suppose $(a'',i'')\between (a\vee d,a\vee d)$. Then we claim $(a'',i'')R(a\vee d,a\vee d)\between (d,l)$. That $(a'',i'')R(a\vee d,a\vee d)$ is immediate from $(a'',i'')\between (a\vee d,a\vee d)$. Finally, from $(d,l)\in S$ we have $d\leq l$, and from $(a,i)R(b,k)\between (d,l)$ we have $a\leq b\leq l$, so $a\leq l$. Hence $(a\vee d,a\vee d)\between (d,l)$. This completes the proof of $R$\textsf{-regularity}.

For \textsf{Reflexivity}, it is immediate from the definition of $R$ that $(a,i)R(a,i)$.

For \textsf{Knowability} as in Definition \ref{Knowability}, given $x=(a,i)$, we claim that $y=(a,a)$ is the desired witness. For suppose $z=(a',i')$ and (i) $(a,a)R(a',i')$. We claim that $z$ refines $x$. Supposing (ii) $(a'',i'')\between (a',i')$, we must show that $(a'',i'')\between (a,i)$, i.e., that (a) $a''\wedge a\neq 0$, (b) $a''\leq i$, and (c) $a\leq i''$. By (i), we have $a'\leq i'\leq a$, and by (ii), we have $a''\wedge a'\neq 0$, so (a) follows. By (ii) we also have $a''\leq i'$, which with $i'\leq a$ from (i) and $a\leq i$ implies (b). Finally, by (ii) we have $a'\leq i''$, and by (i) we have $a\leq a'$, so (c) follows.

For part \ref{EpistemicExt}, we first prove that $e_B(a)$ is $\between$-regular. Suppose $(b,i)\not\in e_B(a)$, so $b\not\leq a$. Then where $b'=b\wedge\neg a$, we have $0 \neq b'\leq b\leq i$, so $(b',i)\between (b,i)$. Now consider any $(b'',i'')\between (b',i)$. Then $b''\wedge b' \neq 0$, which implies $b''\not\leq a$, so $(b'',i'')\not\in e_B(a)$. Thus, we have shown that if $(b,i)\not\in e_B(a)$, then there is a $(b',i')\between (b,i)$ such that for all $(b'',i'')\between (b',i')$, we have $(b'',i'')\not\in e_B(a)$. Hence $e_B(a)$ is $\between$-regular.

Next, clearly $e$ preserves all existing meets from $B$:
\[e_B(\underset{a\in A}{\bigwedge}a) = \{(b,i)\in S\mid b\leq  \underset{a\in A}{\bigwedge}a\} = \underset{a\in A}{\bigcap} \{(b,i)\in S\mid b\leq a\} = \underset{a\in A}{\bigwedge} e_B(a).\]
For preservation of joins, we need only show that $e_B(\bigvee A)\subseteq \bigvee \{ e_B(a)\mid a\in A\}$, as the converse follows from order preservation, which follows from meet preservation. Suppose $(b,i)\in e_B(\bigvee A)$, so $b\leq \bigvee A$, and $(b',i')\between (b,i)$, so $b'\wedge b\neq 0$. Since $0\neq b'\wedge b\leq \bigvee A$, by distributivity in $B$ there is some $a\in A$ such that $a\wedge b'\wedge b\neq 0$. Then setting $b''= a\wedge b'\wedge b$, we have $(b'',i')\between (b,i)$ and $(b'',i')\in e_B(a)$. This shows that $(b,i)\in \bigvee \{ e_B(a)\mid a\in A\}$.

Finally, we show that $e_B(\neg a) = \neg e_B(a)$. Suppose $(b,i)\in e_B(\neg a)$, so $b\leq \neg a$. Then for any $(b',i')\between (b,i)$, since $b'\wedge b\neq 0$, it follows that $b'\not\leq a$, so $(b',i')\not\in e_B(a)$, which shows that $(b,i)\in \neg e_B(a)$. Conversely, suppose $(b,i)\not\in e_B(\neg a)$, so $b\not\leq \neg a$. Then  $(b\wedge a,i)\between (b,i)$ and $(b\wedge a,i)\in e_B(a)$, so $(b,i)\not\in \neg e_B(a)$. 

For part \ref{EpistemicS5}, for all $U\in O(B^\mathsf{e})$,  that $\Box U\subseteq U$ is immediate from the reflexivity of the accessibility relation $R$, and that $\Box U\subseteq \Box\Box U$ is immediate from the transitivity of $R$. From here it suffices to show $U\subseteq\Box\Diamond U$ for all $U\in O(B^\mathsf{e})$. Suppose  $(a,i)\in U$ and $(a,i)R(a',i')$, so $a\leq a'$ and $i'\leq i$. To show $(a',i')\in \Diamond U$, suppose $(a'',i'')\between (a',i')$, so $a''\leq i'$ and $a'\leq i''$. From  $a\leq a'\leq i''$ and $a''\leq i''$, we have  $(a'',i'')R(a\vee a'',a\vee a'')$, and from $a\leq i$ and $a''\leq i'\leq i$, we have $(a\vee a'',a\vee a'')\between (a,i)\in U$. This shows that $(a,i)\in\Box\Diamond U$. 

For part \ref{NoMightCollapse}, if $b\neq 0$, then it is easy to see that $(1,1)\in \Diamond e_B(b)$, and if $b\neq 1$, then $(1,1)\not\in e_B(b)$.\end{proof}

Our next goal is to show that $O(B^\mathsf{e})$ equipped with the distinguished Boolean subalgebra $\mathbb{B}=\{e_B(b)\mid b\in B\}$ is a complete epistemic \textit{ortho-Boolean lattice} as in Definition \ref{OrthoBoole}, or equivalently, that $B^\mathsf{e}$ equipped with the distinguished $\mathbb{B}=\{e_B(b)\mid b\in B\}$  is a \textit{stratified} epistemic compatibility frame as in Definition \ref{GroundedEpistemicFrame}. This is equivalent to showing that each $\mathbb{B}_n$ in the hierarchy of Definition \ref{GroundedEpistemicFrame} is a Boolean algebra. To prove this, we will use the following algebraic analogue of Lemma \ref{SimpleTruth}. 

\begin{lemma}\label{AlgSimpleTruth}
 For any Boolean algebra $B$, possibility $(a,i)$ in $ B^\mathsf{e}$, and $b\in B$, we have:
\begin{enumerate}
\item\label{SimpleTruth1} $(a,i) \in e_B(b)$ iff $ a\leq b$;
\item\label{SimpleTruth2}  $(a,i)\in \Box e_B(b)$ iff $i \leq b$; 
\item\label{SimpleTruth3} $(a,i)\in  \Diamond e_B(b)$ iff  $a\wedge b\neq 0$.
\end{enumerate}

\end{lemma}
\begin{proof}
By essentially the same reasoning used in the proof of Lemma \ref{SimpleTruth}, only with propositions instead of formulas and an arbitrary Boolean algebra instead of a powerset. \end{proof}

We will also use the following generalization of parts \ref{SimpleTruth2} and \ref{SimpleTruth3} of  Lemma \ref{AlgSimpleTruth}.
\begin{lemma}\label{KeyLem}
For any Boolean algebra $B$, the following are equivalent for any possibility $(a,i)$ in  $B^\mathsf{e}$ and families $\{b_k,d^1_k,\dots, d^{m_k}_k\}$ of elements of $B$ for $k\in \{1,\dots,n\}$:
\begin{enumerate}
\item\label{KeyLemA} $(a,i)\in \underset{1\leq k\leq n}{\bigvee}\big(\Box e_B(b_k)\wedge \Diamond e_B(d^1_k)\wedge\dots\wedge \Diamond e_B(d^{m_k}_k) \big) $;
\item\label{KeyLemB} for any $c$ such that $a\leq c\leq i$, there is some $k$ such that $c\leq b_k$ and $c\wedge d^\ell_k\neq 0$ for all $\ell\in \{1,\dots,m_k\}$.
\end{enumerate}
\end{lemma}

\begin{proof} 
Suppose part \ref{KeyLemA} holds and $a\leq c\leq i$. Then  $(c,c)\between (a,i)$, so there is an $(a',i')\between (c,c)$ and $k\in \{1,\dots, n\}$ such that $(a',i')\in \Box e_B(b_k)\wedge \Diamond e_B(d^1_k)\wedge\dots\wedge \Diamond e_B(d^{m_k}_k)$.  Then by Lemma \ref{AlgSimpleTruth}, $i'\leq b_k$ and $a'\wedge d^\ell_k\neq 0$ for all $\ell\in \{1,\dots,m_k\}$. 
Then since $(a',i')\between (c,c)$ implies $a'\leq c$ and $c\leq i'$, it follows that $c\leq b_k$ and $c\wedge d^\ell_k\neq 0$ for all $\ell\in \{1,\dots,m_k\}$. Therefore, part \ref{KeyLemB} holds.

Now suppose part \ref{KeyLemB} holds. Suppose $(a',i')\between (a,i)$. Then since $a\leq a\vee a'\leq i$, 
part \ref{KeyLemB} implies there is some $k$ such that $a\vee a'\leq b_k$ and $(a\vee a')\wedge  d^\ell_k\neq 0$ for all $\ell\in \{1,\dots,m_k\}$.  From $(a',i')\between (a,i)$ we also have $a\vee a'\leq i'$. Then where $x=a\vee a'$, 
we have $x\wedge a'\neq 0$, $x\leq i'$, and $a'\leq b_k$, so $(x,b_k)\between (a',i')$,  and by Lemma~\ref{AlgSimpleTruth}, $(x, b_k)\in \Box e_B(b_k)\wedge \Diamond e_B(d^1_k)\wedge\dots\wedge \Diamond e_B(d^{m_k}_k)$.  Thus, we have shown that for every $(a',i')\between (a,i)$, there is an $(a'',i'')\between (a',i')$ and $k\in \{1,\dots,n\}$ such that $(a'',i'')\in \Box e_B(b_k)\wedge \Diamond e_B(d^1_k)\wedge\dots\wedge\Diamond e_B(d^{m_k}_k)$. Therefore, part \ref{KeyLemA} holds. \end{proof}

Using Lemma \ref{KeyLem}, we can show that a general distributive law holds for propositions resulting from applying $\Box$ and $\Diamond$ to Boolean propositions.

\begin{lemma}\label{B1dist} Let $B$ be a Boolean algebra. Let $\{U_{t,j}\mid t\in T, j\in J_t\}$ be a finite family of propositions in $O(B^\mathsf{e})$ such that for each $U_{t,j}$, there is a $b\in B$ such that $U_{t,j}=\Box e_B(b)$ or $U_{t,j}=\Diamond e_B(b)$. Then we have
\[\underset{t\in T}{\bigwedge}\, \underset{j\in J_t}{\bigvee} U_{t,j}\subseteq \underset{f\in F}{\bigvee}\, \underset{t\in T}{\bigwedge}U_{t,f(t)} \]
where $F$ is the set of functions that choose for each $t\in T$ some $f(t)\in J_t$.
\end{lemma}
\begin{proof} Suppose $(a,i)\in \underset{t\in T}{\bigwedge}\, \underset{j\in J_t}{\bigvee} U_{t,j}$. Then by Lemma \ref{KeyLem}, for each $t\in T$ and $c$ such that $a\leq c\leq i$, there is some $j\in J_t$ such that either 

\begin{itemize}
\item[(i)] $U_{t,j}= \Box e_B(b)$ and $c\leq b$ or
\item[(ii)] $U_{t,j}= \Diamond e_B(b)$ and $c\wedge b\neq 0$.
\end{itemize}
Let $f_c$ be a function that chooses such $j\in J_t$ for each $t\in T$. Then let us write $\underset{t\in T}{\bigwedge}U_{t,f_c(t)}$ as 
\begin{equation*}\Box e_B(b^1_c)\wedge \dots \wedge \Box e_B(b_c^{n_c})\wedge  \Diamond e_B(d_c^1)\wedge \dots \wedge \Diamond e_B(d_c^{m_c}) .\label{ModalMeet}\end{equation*}
Where $b_c= b^1_c\wedge\dots\wedge b^{n_c}_c$, this is equivalent to 
\begin{equation*}\Box e_B(b_c)\wedge  \Diamond e_B(d_c^1)\wedge \dots \wedge \Diamond e_B(d_c^{m_c}) .\end{equation*}
Then to show that $(a,i) \in\underset{f\in F}{\bigvee}\, \underset{t\in T}{\bigwedge}U_{t,f(t)}$, it suffices by Lemma \ref{KeyLem} to show that for each $c$ such that $a\leq c\leq i$, we have $c\leq b_c$ and $c\wedge d^s_c\neq 0$ for each $s\in \{1,\dots, m_c\}$. Indeed, we have $c\leq b_c$ by (i) in the construction of $f_c$, and $c\wedge d^s_c\neq 0$ for each $s\in \{1,\dots, m_c\}$ by (ii) in the construction of $f_c$. This completes the proof.\end{proof}

We are now ready to prove that each of the distinguished algebras $\mathbb{B}_n$ is indeed Boolean.

\begin{theorem}\label{StratTheorem} For any Boolean algebra $B$,  $O(B^\mathsf{e})$ with the distinguished Boolean subalgebra $\{e_B(b)\mid b\in B\}$ is a complete epistemic ortho-Boolean lattice. Equivalently,  $(B^\mathsf{e},\{e_B(b)\mid b\in B\})$ is a stratified epistemic compatibility frame.
\end{theorem}
\begin{proof} Where $\mathbb{B}=\{e_B(b)\mid b\in B\}$, it is immediate from Theorem \ref{EpistemicFrameExtThm}.\ref{EpistemicExt} that $\mathbb{B}$ forms a Boolean algebra. It only remains to show that each $\mathbb{B}_n$ for $n\geq 1$ is a Boolean algebra.  First, let $\mathbb{B}'_1$ be the set of elements generated from $\{\Box U\mid U\in\mathbb{B}_0\}\cup \{\Diamond U\mid U\in\mathbb{B}_0\}$ by $\wedge$ and $\vee$. Thanks to De Morgan's laws and involution for negation, $\mathbb{B}'_1=\mathbb{B}_1$.  Now we claim that each $V\in\mathbb{B}'_1$ is equal to $\Box V$. The proof is by induction, given the inductive definition of $\mathbb{B}_1'$. In the case where $V$ is $\Box U$ or $\neg\Box U$ for some $U\in\mathbb{B}_0$, we have that $\Box U =\Box\Box U$ and $\neg\Box U =\Box\neg\Box U$ by the S5 axioms given by Theorem \ref{EpistemicFrameExtThm}.\ref{EpistemicS5}. If $V$ is $V_1\wedge V_2$ for $V_1,V_2\in\mathbb{B}'_1$, then by the inductive hypothesis, $V_1\wedge V_2=\Box V_1\wedge \Box V_2$, which is equal to $\Box (V_1\wedge V_2)$. Finally, if $V$ is $V_1\vee V_2$ for $V_1,V_2\in\mathbb{B}'_1$, then by the inductive hypothesis, $V_1\vee V_2=\Box V_1\vee \Box V_2$, which is equal (using the 4 and T axioms) to  $\Box(\Box V_1\vee \Box V_2)$. Thus, each $V\in \mathbb{B}_1$ is equal to $\Box V$. Hence each generator of $\mathbb{B}_2$, i.e., each $\Box V$ for $V\in \mathbb{B}_1$, already belongs to $\mathbb{B}_1$, which implies $\mathbb{B}_2=\mathbb{B}_1$ and indeed $\mathbb{B}_n=\mathbb{B}_1$ for $n\geq 1$. Thus, to complete the proof we need only show that $\mathbb{B}_1$, or equivalently, $\mathbb{B}_1
'$, is distributive. 
As shown by Kolibiar \citeyearpar[Thm.~2a]{Kolibiar1972}, to show that $\mathbb{B}_1'$ is distributive and hence Boolean, it suffices to show that for any finite set $M$ of finite subsets of the generating set  $\{\Box U\mid U\in\mathbb{B}_0\}\cup \{\Diamond U\mid U\in\mathbb{B}_0\}$, we have
\begin{equation*}\underset{X\in M}{\bigwedge}\,\bigvee X \subseteq \underset{ t\in \Pi M}{\bigvee}\,\underset{X\in M}\bigwedge t(X),\label{DistEq}\end{equation*}
where $\Pi M$ is the set of all functions $t$ on $M$ such that $t(X)\in X$ for $X\in M$. This is precisely Lemma~\ref{B1dist}.
\end{proof}

The moral of Theorems \ref{EpistemicFrameExtThm} and \ref{StratTheorem} is that we can add to any Boolean algebra $B$ of propositions new epistemic modal propositions, by embedding $B$ into a complete epistemic ortho-Boolean lattice that arises from a stratified epistemic compatibility frame built from $B$. The resulting ``epistemic extension'' therefore validates our logic $\mathsf{EO}^+$. Thus, we have an abundant Boolean source of models for our logic.

Now the obvious questions are the following. Which complete epistemic ortho-Boolean lattices arise as epistemic extensions of Boolean algebras? And what is the logic of epistemic frames of Boolean algebras? We answer the first question with an algebraic characterization in Appendix \ref{CharApp}. In Appendix \ref{DecideApp}, we prove the \textit{decidability} of the logic of epistemic frames of Boolean algebras; and in Appendix \ref{NormalApp}, we prove a related semantic normal form theorem. Crucially, the logic is stronger than $\mathsf{EO}^+$, not only because it includes the S5 principles but also because it includes the following principles.

\begin{proposition} \label{weakdist}For any Boolean algebra $B$ and $U,U_j,V_{j,k}\in O(B^\mathsf{e})$ in the image of the embedding~$e_B$:
\begin{enumerate}
\item\label{weakdist1} restricted diamond distributivity: if $V_{j,k}\subseteq U_j$ for  all $1\leq j\leq n$ and $1\leq k\leq m_j$, then \[\big(\underset{1\leq j\leq n}{\bigvee}U_j \big)\wedge \underset{1\leq j\leq n}{\bigwedge} (\Diamond V_{j,1}\wedge\dots\wedge \Diamond V_{j,m_j})\subseteq \underset{1\leq j\leq n}{\bigvee}(U_j\wedge \Diamond V_{j,1}\wedge\dots\wedge \Diamond V_{j,m_j});\]
\item\label{weakdist1.5} box distributivity: $(U_1\vee U_2)\wedge \Box V\subseteq (U_1\wedge\Box V)\vee (U_2\wedge\Box V)$;
\item\label{weakdist2} inheritance: $(U\wedge\Diamond V)\subseteq \Diamond (U\wedge V)$;
\item\label{weakdist3} diamond disjunctive syllogism A: $(U\vee \Diamond V)\wedge \neg \Diamond V\subseteq U$;
\item\label{weakdist4} diamond disjunctive syllogism B: $(U\vee \Diamond V)\wedge \neg U\subseteq \Diamond V$.
\end{enumerate}
\end{proposition}

\begin{proof} For part \ref{weakdist1}, let $b_j=e_B^{-1}(U_j)$ and $c_{j,k}=e_B^{-1}(V_{j,k})$. Then since $V_{j,k}\subseteq U_j$, we have $c_{j,k}\leq b_j$. Suppose $(a,i)$ is in the left-hand side of the inclusion. Then $a\leq b_1\vee \dots\vee b_n$ and $a\wedge c_{j,k}\neq 0$ for $1\leq j\leq n$ and $1\leq k\leq m_j$ by Lemma \ref{AlgSimpleTruth}. To show $(a,i)$ is in the right-hand side, consider any $(a',i')\between (a,i)$. Since $a'\wedge a\neq 0$ and $a\leq b_1\vee \dots \vee b_n$, it follows that $a'\wedge a\wedge b_j\neq 0$ for some $j$. Let $x=(a'\wedge a\wedge b_j)\vee (a\wedge c_{j,1})\vee \dots\vee (a\wedge c_{j,m_j}) $. Since $a'\leq i'$ and $a\leq i'$, we have $x\leq i'$. Hence $(x,i')\between (a',i')$. Moreover, since $x\leq b_j$ and $x\wedge c_{j,k}\neq 0$ for $1\leq k\leq m_j$, we have $(x,i')\in U_j\wedge\Diamond V_{j,1}\wedge\dots\wedge \Diamond V_{j,m_j}$ by Lemma~\ref{AlgSimpleTruth}. Thus, for every $(a',i')\between (a,i)$, there is an $(a'',i'')\in (U_1\wedge\Diamond V_{1,1}\wedge\dots \wedge \Diamond V_{1,m_1})\cup \dots \cup (U_n\wedge\Diamond V_{n,1}\wedge\dots \wedge \Diamond V_{n,m_n})$, which shows that $(a,i)$ is in the right-hand side.

For part \ref{weakdist1.5}, let $b_1=e_B^{-1}(U_1)$, $b_2=e_B^{-1}(U_2)$, and $c=e_B^{-1}(V)$. Suppose $(a,i)$ is in the left-hand side, so $i\leq c$ by Lemma \ref{AlgSimpleTruth}. Further suppose that $(a',i')\between (a,i)$, so $a'\leq i\leq c$. Then since $(a,i)\in U_1\vee U_2$, there is an $(a'',i'')\between (a',i')$ such that $(a'',i'')\in U_i$ for some $i\in\{1,2\}$, so $a''\leq b_i$. Now $(a'',i'')\between (a',i')$ implies $(a''\wedge a',a')\between (a',i')$, and since $a''\leq b_i$ and $a'\leq c$, we have $(a''\wedge a',a')\in U_i\wedge \Box V$ by Lemma \ref{AlgSimpleTruth}. Thus, for every $(a',i')\between (a,i)$, there is an $(a'',i'')\in (U_1\wedge \Box V)\cup (U_2\wedge \Box V)$, so $(a,i)$ is in the right-hand side.

For part \ref{weakdist2}, let $b=e_B^{-1}(U)$ and $c=e_B^{-1}(V)$. If $(a,i)$ is in the left-hand side of the inclusion, then $a\leq b$ and $a\wedge c\neq 0$ by Lemma \ref{AlgSimpleTruth}, which implies $a\wedge b\wedge c\neq 0$, so $(a,i)$ is in the right-hand side by Lemma \ref{AlgSimpleTruth}.

For parts \ref{weakdist3} and \ref{weakdist4},  let $b=e_B^{-1}(U)$ and $c=e_B^{-1}(V)$. If $(a,i)\in U\vee \Diamond V$, then for any $(a',i')\between (a,i)$ there is an $(a'',i'')\between (a',i')$ such that $(a'',i'')\in U$ or $(a'',i'')\in\Diamond V$, so $a''\leq b$ or $a''\wedge c\neq 0$ by Lemma~\ref{AlgSimpleTruth}. 

For part \ref{weakdist3}, if $(a,i)\not\in U$, so $a\not\leq b$, then $(a\wedge\neg b,i)\between (a,i)$, so by the previous paragraph, there is an $(a'',i'')\between (a\wedge\neg b,i)$ such that $a''\leq b$ or $a''\wedge c\neq 0$; but $a''\leq b$ contradicts $a''\wedge a\wedge\neg b\neq 0$ from $(a'',i'')\between (a\wedge\neg b,i)$, so we conclude that $a''\wedge c\neq 0$, which implies $(a'',i''\vee a)\in \Diamond V$ by Lemma \ref{AlgSimpleTruth}. Finally, $(a'',i'')\between (a\wedge\neg b,i)$ implies $(a'',i''\vee a)\between (a,i)$, so $(a'',i''\vee a)\in \Diamond V$ implies $(a,i)\not\in \neg\Diamond V$.

For part \ref{weakdist4}, if $(a,i)\not\in \Diamond V$, so $a\wedge c=0$ by Lemma \ref{AlgSimpleTruth}, then let $(a',i')=(a,a)$. Hence $(a',i')\between (a,i)$, so as above there is an  $(a'',i'')\between (a',i')$ such that $a''\leq b$ or $a''\wedge c\neq 0$. But  $(a'',i'')\between (a',i')$ implies $a''\leq i'=a$, so $a\wedge c=0$  implies $a''\wedge c=0$ and hence $a''\leq b$, which implies $(a'',i'')\in U$. Since $(a'',i'')\between (a',i')=(a,a)$, it follows that $(a'',i'')\between (a,i)$, so $(a,i)\not\in\neg U$.\end{proof}

It is not difficult to construct epistemic ortho-Boolean lattices validating $\mathsf{EO}^+$ that invalidate the principles in Proposition \ref{weakdist}. These principles may, however, by desirable (though the plausibility of part \ref{weakdist2} may depend on the S5 assumption built into the construction of possibilities from a set of worlds). To see the plausibility of the first principle, consider examples like \ref{mightfrance}, suggested to us by 
Seth Yalcin (p.c.): 

\ex. \label{mightfrance} \a. \label{mightfrancea} Either the window is open or the window is closed, and it might be raining in France. \b. \label{mightfranceb} Therefore, either the window is open and it might be raining in France, or the window is closed and it might be raining in France.

The inference from \ref{mightfrancea} to \ref{mightfranceb} is intuitively reasonable. After all, the possible weather in France has nothing to do with the state of my window, and so nothing, intuitively, prevents us from distributing the disjunction over the conjunction in this case. Given the enthymematic premise that \textit{it might be that both the window is open and it's raining in France}, and \textit{it might be that both the window is closed and it's raining in France}, the weakening of distributivity in Proposition \ref{weakdist}.\ref{weakdist1} would explain the felt validity of this reasoning. The other principles also strike us as reasonable, though we will not discuss them further here; see \citealt{Goldstein:2018,Hawke:2020} for discussion of inheritance in particular.

The principles in Proposition \ref{weakdist} highlight the interest of providing an axiomatization of the logic of epistemic frames of Boolean algebras, extending $\mathsf{EO}^+$, in the style of Definitions \ref{EODef} and \ref{EOplus}. However, we leave this task for future work.

\section{Comparisons}\label{Comparisons}

In this section, we compare our approach to epistemic modals to some others in the literature. We cannot cover them all, since there are far too many. We will, however, briefly situate our approach in the existing literature, with emphasis on comparing our theories to the most similar ones in the literature.\footnote{Among those we will not discuss are the salience-based approaches of \citealt{Dorr:2013} and \citealt{Stojnic:2016} and the probabilistic approaches of \citealt{Moss2014} and \citealt{Swanson2016}. See \citealt{Pseudodynamics} for critical discussion of the former and \citealt{Mandelkern:2018a} and \citealt{Hawke:2020} for critical discussion of the latter.}

Among non-\textsf{E}-logics, several variants are worth mentioning. First, there are logics in which Wittgenstein sentences are contradictions but cannot always be replaced by contradictions \emph{salva veritate}.
In other words, these logics have half of the profile of \textsf{E}-logics. A prominent example of a logic like this comes from the domain semantics for epistemic modals (\citealt{Yalcin2007}, \citealt{MacFarlane:2011}, \citealt{Bledin:2014}) together with Yalcin \citeyearpar{Yalcin2007} and Bledin's \citeyearpar{Bledin:2014} notion of informational consequence or Kolodny and MacFarlane's \citeyearpar{Kolodny:2010} notion of quasi-validity.  But the fact that Wittgenstein sentences cannot always be substituted for contradictions ends up depriving approaches like this of a great deal of empirical coverage, since we cannot conclude that $\varphi$ and $\varphi'$ are always equivalent when $\varphi$ embeds a sentence $\psi\wedge\lozenge \neg \psi$ and $\varphi'$ replaces that sentence with $\psi\wedge\neg \psi$. More concretely,  in the examples from \S~\ref{EpContra} of modalized, disjoined, and quantified Wittgenstein sentences (\ref{drekm}--\ref{indefiniteec}), the cited approaches fail to predict any infelicity (see \citealt{Mandelkern:2018a} for extensive discussion). Another way to put this is in terms of logic. For instance, in the logics of the cited approaches, $\lozenge p\wedge\lozenge \neg p$ entails $(\lozenge p\wedge\neg p)\vee (\lozenge\neg p \wedge p)$. Likewise, as 
Justin Helms (p.c.) pointed out to us, $\lozenge p$ entails $ p\vee (\neg p\wedge\lozenge p)$, which is intuitively equivalent to $p$ (see \citealt{Dorr:2013} for a related observation). 
This brings out the importance of both halves of an \textsf{E}-logic---a logic in which Wittgenstein sentences are contradictions \emph{and} can always be substituted for contradictions. 

A  different non-\textsf{E}-logic with the same profile is developed in a multilateral proof-theoretic setting in \citealt{Incurvati:2020} and \citealt{Aloni:2022}. We are methodologically sympathetic with that approach, insofar as it focuses on directly characterizing the  logic of epistemic modals. However, the fact that their logic is not an \textsf{E}-logic leads to problems like those that face the domain semantics: disjoined and modalized Wittgenstein sentences are predicted to be consistent. \citet{Incurvati:2020} propose to explain the incoherence of those cases by adopting a pragmatic rule that says that a disjunction can only be asserted if its disjuncts are supposable, and likewise that $\lozenge\varphi$ can only be asserted if $\varphi$ is supposable. Since Wittgenstein sentences are not supposable in their logic, this suffices to account for the incoherence of asserting a disjunction of Wittgenstein sentences or a Wittgenstein sentence embedded under `might'. (A defender of domain semantics could appeal to a similar rule.) 

However, this approach  fall shorts when it comes to disjoined or modalized Wittgenstein sentences that are themselves embedded under attitude predicates. There is clearly no rule that says the prejacent of an attitude predicate needs to be supposable in order for the attitude ascription to be assertable. For instance, if Sue has inconsistent beliefs about the weather, we can report them as in \ref{incon}:

\ex. \label{incon} Sue believes that it is raining and that it isn't raining. Her beliefs are inconsistent!

But now consider a disjoined Wittgenstein sentence under `believes':

\ex. Sue believes that either it's raining and might not be, or it isn't raining and might be.\label{decb}

Our judgment about \ref{decb} is that, like \ref{incon}, it ascribes incoherent beliefs to Sue. But we cannot see how an account like the one under discussion would predict this, since (i) the prejacent of `believes' is consistent on this account; and (ii) as \ref{incon} illustrates, there is no pragmatic rule that prevents us from ascribing unsupposable beliefs to others. Similar points apply to sentences like \ref{decm}:

\ex. Sue believes that it might be that it's raining and might not be raining.\label{decm}

In our view, these points provide further evidence for  an \textsf{E}-logic and against a pragmatic account of these~data.

Another kind of non-\textsf{E}-logic comes from the dynamic semantic tradition (\citealt{Veltman:1996}, \citealt{GSV:1996}). There are many different systems under this umbrella; for concreteness, we will focus on the semantics plus consequence relation of  \citealt{GSV:1996} for the moment and then say something about how this generalizes. In the system of \citealt{GSV:1996},  sentences with the form $\varphi\wedge\lozenge \neg \varphi$ are contradictions provided that $\varphi$ is Boolean. However, when $\varphi$ is non-Boolean, $\varphi\wedge\lozenge\neg \varphi$ can fail to be a contradiction. And order variants like $\lozenge \varphi\wedge\neg \varphi$ are generally not contradictions. As many have observed, however, there is no evidence for the kinds of stark order asymmetries thus predicted (see, e.g., \citealt{Yalcin:2013}, \citealt{Pseudodynamics}). While this approach has a ready explanation of the unassertability of Wittgenstein sentences of both orders, it has a harder time making sense of various embedding data involving the left-modal variant. For instance, the standard dynamic approach to the conditional predicts that $(\lozenge \varphi\wedge\neg \varphi)\to \psi$ is coherent and indeed equivalent to $\lozenge\varphi\wedge(\neg\varphi\to\psi)$ in  contexts consistent with $\varphi$. By contrast, sentences like this sound as bad as in the reverse order, as illustrated by the contrast between \ref{rainmight1} and \ref{rainmight2}:

\ex. \a. \# If it might be raining but it isn't, we don't need rainjackets.\label{rainmight1}
\b. It might be raining, but if it isn't, we don't need rainjackets.\label{rainmight2}

 Finally, this system invalidates many classical inferences---including the law of non-contradiction ($\varphi\wedge\neg\varphi\vdash\bot$) and excluded middle---which are intuitively valid for epistemic language (\citealt{Mandelkern:2018}). 

There are many other approaches to epistemic modals within a broadly dynamic framework (see, e.g.,  \citealt{Does:1997}, \citealt{Aloni:2000}, \citealt{YalcinDeRe}, \citealt{Goldstein:2018}, \citealt{Gillies:2018}), which have respectively different logics. For instance, if we combined the semantic system of \citealt{GSV:1996} with the \emph{test-to-test} consequence relation instead of the update-to-test relation of \citealt{GSV:1996}, we would have a logic where all Wittgenstein sentences are contradictions, as are sentences with the form $\varphi\wedge\neg\varphi$. However, this change in perspective does not yield a system that is descriptively adequate, given the desiderata laid out at the beginning, because in this system contradictions do not always \emph{embed} as we would expect contradictions to. So, for instance, while $\varphi\wedge\neg\varphi$ is a contradiction in this system, $\lozenge(\varphi\wedge\neg\varphi)$ is not, and neither is $\lozenge(\lozenge\varphi\wedge\neg\varphi)$. That is, (epistemic) contradictions remain epistemically possible. Likewise, (epistemic) contradictions can be coherently embedded in the antecedents of conditionals in  this system: $(\varphi\wedge\neg\varphi)\to\psi$ and $(\lozenge\varphi\wedge\neg\varphi)\to\psi$ can both be coherent and non-trivially true and false. Indeed, $(\lozenge\varphi\wedge\neg\varphi)\to\psi$ is still predicted to be equivalent to $\lozenge\varphi\wedge(\neg\varphi\to\psi)$ in contexts consistent with $\varphi$. While there are many different pairs of dynamic semantic systems and dynamic notions of consequence, with diverse logical profiles, we do not know of any that can capture the full range of embedding data sketched in \S~\ref{des}.
(See \citealt{Pseudodynamics} for more extensive critical discussion of various dynamic systems and more on the surrounding dialectic.)

We turn now to \textsf{E}-logics: that is, logics that predict that Wittgenstein sentences are contradictions and are always substitutable for contradictions (in some but not all of these systems, this holds because of a more general fact that logically equivalent sentences are always substitutable \emph{salva veritate}).

One \textsf{E}-logic in the literature is the bounded theory of \citealt{Mandelkern:2018a}. That system is given in a trivalent setting, which yields (at least) two natural logics: the \emph{truth-preserving} logic (the set of inferences that preserve truth) and the \emph{Strawson logic} (the set of inferences that preserve truth when the premises and conclusion are all true or false). Both logics are \textsf{E}-logics.\footnote{In the weaker of the two systems Mandelkern considers, this holds only if we restrict attention to \emph{diagonal consequence}.} But neither has exactly the properties that we have argued a logic for epistemic modals should have. In the truth-preserving logic, many classical inferences, like conjunction introduction, fail; from this perspective, the logic is too non-classical. In the Strawson logic, by contrast, all classical inferences are valid. But from this perspective, the logic is too classical: inferences like distributivity and disjunctive syllogism, which intuitively fail, come out valid. From a logical perspective, neither of these profiles is exactly the right fit for epistemic modals. The combination of the two perspectives might still suffice to make sense of the data, but the present approach is logically much more perspicuous.  

 By identifying a logic that is sound and complete with respect to the algebraic semantics and possibility semantics we have studied, consisting of intuitively valid principles, we have strong evidence that our semantics does not overgenerate in its predictions of the logical inferences that will strike speakers as valid. Since capturing felt (in)validity is a central part of the task of natural language semantics, having such a characterization is an essential part of any adequate theory of a fragment of natural language. We emphasize this point not just because it brings out a contrast with the approach of \citealt{Mandelkern:2018a} but  because in general axiomatization is often left out of contemporary semantic theory, and we think it is important methodologically to explicitly axiomatize the set of inferences predicted by a system to strike speakers as valid (see \citealt{HollidayIcard2018} for more discussion of the role of axiomatization in semantic theory).

A different class of (nearly) \textsf{E}-logics is given in the state-based frameworks of \citealt{Veltman85}, \citealt{Hawke:2020}, \citealt{Aloni:2016}, and \citealt{Flocke:2020}, where sentences are evaluated for truth relative to information states. 
Formally speaking, we can (if we want) model an information state as a set of possible worlds, which in turn is a special case of a possibility. So the indices in our semantics  are closely related to those in state-based semantic frameworks, which makes these approaches a direct precedent for ours. We will discuss two of these approaches in a bit more detail. 
First consider Veltman's \emph{data semantics}.\footnote{\citet{Gauker:2020} develops an interesting system that is, as far as our concerns go, relevantly like Veltman's.} In data semantics, $\lozenge \varphi$ is true at an information state just in case it has some refinement where $\varphi$ is true. When $\beta$ is Boolean, it is persistent: if $ \beta$ is true in an information state, it is true in every refinement of that information state. So $\beta\wedge\lozenge\neg\beta$ can never be true (conjunction has standard truth-conditions) and is everywhere substitutable for $\bot$. However, this system is not quite an \textsf{E}-logic because (just as in dynamic semantics) $\varphi\wedge\lozenge\neg \varphi$ is not contradictory when $\varphi$ itself is epistemic. For instance, where $\varphi=\lozenge p$, $\lozenge p$ may be true in a state but false in a refinement of that state (namely, a refinement that makes $\neg p$ true), so $\lozenge p\wedge \lozenge \neg\lozenge p$ will be consistent. 
And Veltman's logic in general is in some ways less classical and in other ways more classical than ours---and in both directions, we think our system is an improvement. In the first direction, 
  data semantics invalidates excluded middle (\citealt[p.~172]{Veltman85}). In this respect, it is even more extreme than  dynamic semantics, since excluded middle fails even for Boolean sentences (disjunction has its standard truth-conditions, so in a state that does not settle $p$ or $\neg p$, $p\vee\neg p$ will not be true, because neither $p$ nor $\neg p$ is true). In fact, as Veltman points out, there are no valid Boolean sentences at all. Of course, we may be able to give some other explanation of why  $p\vee\neg p$ has the feeling of a validity, but it is not immediately clear what that would be.
In the second direction, Veltman validates distributivity, because he validates proof by cases with side assumptions (from $\xi\wedge \varphi\vdash \chi$ and $\xi\wedge \psi\vdash \chi$ conclude $\xi\wedge(\varphi\vee\psi)\vdash\chi$), and disjunction introduction, which  together entail distributivity, as noted in \S~\ref{Log1}.  Thus, assuming disjunction introduction, we have an argument against proof by cases with side assumptions. For a more direct argument, note that the following are all intuitive (and hold in our logic):
\[\Diamond p\wedge\Diamond\neg p\vdash (\Diamond p\wedge\Diamond\neg p)\wedge(p\vee\neg p),\; (\Diamond p\wedge\Diamond \neg p)\wedge p\vdash\bot,\mbox{ and } (\Diamond p\wedge\Diamond \neg p)\wedge \neg p\vdash\bot .\]
But then by proof by cases with side assumptions and transitivity of $\vdash$, we would have $\Diamond p\wedge\Diamond\neg p\vdash \bot$. 

In fact, since data semantics validates distributivity and $\beta\wedge\lozenge\neg\beta$ is always substitutable for a contradiction, it has an unwanted prediction:  it rightly predicts that $(\beta\wedge\lozenge \neg \beta )\vee (\neg \beta \wedge\lozenge\beta)$ is a contradiction, but then since it validates distributivity, it predicts that $(\beta\vee\neg\beta)\wedge(\lozenge\neg\beta\wedge \lozenge \beta)$ is a contradiction. This seems implausible and indeed worse than the \emph{lack} of a prediction that $p\vee\neg p$ is valid: while the latter lack might be removed via some other theory of felt validity, it is hard to see how this approach could be made compatible with the feeling that $(\beta\vee\neg\beta)\wedge(\lozenge\neg\beta\wedge \lozenge \beta)$ is consistent and indeed equivalent to $\lozenge\neg\beta\wedge\lozenge\beta$.

 Turn next to the system of \citealt{Hawke:2020} (henceforth HS-T), which is designed specifically to account for the epistemic modal data above and indeed does a beautiful job  of capturing them---in particular, invalidating distributivity, unlike Veltman. 
HS-T's system, like ours, is classical over the Boolean fragment, but over the full language, it invalidates De Morgan's laws as well as disjunction introduction. Of course, we might have found failures of these laws, and HS-T indeed argue, on the basis of the free choice inference, that disjunction introduction should fail. We will not try to adjudicate that complicated topic here.\footnote{HS-T also argue for the principle they call \emph{Inheritance}, which says that $p\wedge\lozenge q\vDash \lozenge(p\wedge q)$. We showed in Proposition \ref{weakdist}.\ref{weakdist2} that this principle is valid (for $p,q$ Boolean) in epistemic frames of Boolean algebras. On the other hand, if we drop the S5 assumption built into the epistemic frame construction, then one may have doubts. For example, perhaps you can have high credence that Mary will come to the party and  that Aliyah might come to the party, while having very low credence that it might be that they both come (e.g., they just broke up). Certainly Inheritance cannot be accepted for arbitrary non-Boolean $p, q$, since $  \Diamond p\wedge\Diamond \neg p$ should not entail $\Diamond (\Diamond p\wedge \neg p)$ (i.e., $\Diamond \bot$ and hence $\bot$). Compare these points to Example \ref{DiamondBoxDiamond}.} 
  We think the failures of De Morgan's laws are more obviously problematic.   For instance, $\lozenge p\wedge\lozenge\neg p$ is not equivalent in their system to $\neg (\Box \neg p\vee \Box  p)$: while the first is true, as one would expect, in any state containing both $p$-worlds and $\neg p$-worlds, the second is not. In fact, for Boolean $p$, $\neg (\Box \neg p\vee \Box  p)$ is never true. For $\neg (\Box \neg p\vee \Box  p)$ is true at a state iff $\Box \neg p\vee \Box  p$ is false at a state. But $\Box \neg p\vee \Box  p$ is true at every state: for it is true at a state iff that state has a (possibly empty) part where $\Box \neg p$ is true  and a (possibly empty) part where $\Box p$ is true. But that will always hold: we can find those parts simply by partitioning the state into the $p$- and $\neg p$-worlds. This prediction seems wrong. On the one hand, the sentences in \ref{hstpair} seem equivalent (as De Morgan's laws, and our account, predict, but contrary to HS-T); and on the other, the disjunction in \ref{hst2} does not sound like a logical truth, contrary to HS-T (it feels more committal than `Either it isn't raining or it is').

\ex. \label{hstpair}\a. It might be raining and it might not be raining.
\b. It's not the case that it must not be raining or must  be raining. 

\ex. \label{hst2} It must not be raining or it must be raining. 

So we think it is a substantial advantage of our view that we validate the De Morgan equivalences and do not validate $\Box \neg p\vee\Box p$.

A second point concerns probabilities. We can naturally form judgments about the probabilities of sentences containing modals, which leads to interesting puzzles. In a companion paper, we show how to extend a probability measure on the powerset of a set of possible worlds to a measure on the epistemic extension of that powerset, yielding a natural model of those judgments. 
 While we leave the details to that paper, here a simple observation suffices: since possibilities represent \textit{how things are}, \textit{partially}, rather than attitudes towards a part of the world, there is no fundamental problem assigning probabilities to sets of possibilities. 
 By contrast, state-based theories face fundamental problems here. The first is conceptual. In state-based theories, at least as understood by HS-T, the indices relative to which such sentences are evaluated as true or false are information states. Those information states  represent mental states, and it is unclear how we could define reasonable measures over sets of such states. (Sloganistically, the probability of a sentence is the probability that it is \emph{true}, not the probability that it is \emph{accepted} by some relevant agents; for a similar point, compare \citet[p. 5]{KhooBook}: ``What would it mean to assign a probability to a property of your cognitive state?'') 
The second, closely related problem arises from the fact that in state-based theories $p$ entails $\Box p$. This holds not just in the informational logic, but also in the logic of truth-preservation, since truth is defined relative to mental states. 
 Reflecting just on what it takes for a mental state to \emph{accept} something---the primitive notion in state-based semantics---this equivalence might look defensible. But as soon as we consider probabilities,  it does not; the probability of $\Box p$ can be less than 
the probability of $p$, as we noted in \S~\ref{des}.
While of course it is open to state-based approaches to deny that probabilities must track entailments, it is hard to see, practically speaking, how to construct a probability measure over their points that avoids this conclusion, since, again, $\Box p$ is true at any state where $p$ is. Note that this  second problem will persist even if we drop HS-T's conceptualization of states as representing mental states, since it arises simply from the logic of the system.  The state-based theory thus faces a fundamental problem in meshing with a reasonable account of credences defined on epistemic modal propositions.

It is worth noting that similar points apply to a range of theories that characterize sentential contents in non-standard terms, including dynamic semantics and various expressivist and relativist theories.  In those systems, sentential contents in general are not the kinds of things ordinarily measured by probabilities, raising at least a prima facie worry about how such a system can make sense of exchanges like \ref{doubleh}:

\ex. \label{doubleh} \a. The coin must have landed heads.\label{doubleh}
\b. That has a .5 probability of being true.

Exchanges like this arguably make sense in some cases (for instance, in a case where there is a .5 chance that the coin was double-headed). 
But in systems where the content of a sentence like \ref{doubleh} is a kind of thing that probability measures are not defined on, exchanges like this, on the face of it, will be nonsense. This is particularly problematic since many of these theories predict that while an exchange like \ref{doubleh} should make no sense, the corresponding sentence in \ref{doubleh2} is perfectly coherent:

\ex. \label{doubleh2} There is .5 probability that the coin must have landed heads.

To give one example of a recent system with a profile like this, in the probabilistic semantics of \citealt{Moss2014}, sentential contents are sets of probability measures, which are not, of course, themselves defined on sets of probability measures.  See \citealt{KhooBook}, \citealt{Goldstein:2021}, and \citealt{Charlow:2019} for discussion of this problem and some attempted solutions.

A final point is slightly more complicated, but it brings out an important difference between our system and the state-based ones cited above, as well as standard implementations of domain and dynamic semantics. HS-T, again, cast their system as an attempt to cash out the core of the \emph{expressivist} approach of \citealt{Yalcin2007}: namely, that to assert $\lozenge p$ is to express a certain state of mind, namely, one that is compatible with $p$, and thus to believe $\lozenge p$ is simply for $p$ to be compatible with one's beliefs---hence their definition of truth at a state, understood as a mental state. This idea is tempting, but while this approach seems reasonable enough for non-factive mental states,  it does not work for factive mental states (\citealt{YalcinContext}, \citealt{Dorr:2013}). It is not plausible that to know $\lozenge p$ is for $p$ to be compatible with your knowledge, since any truth is compatible with your knowledge, but it does not follow from $p$ that you know $\lozenge p$. That is, the natural extension of this approach to factives will validate the inference $p\vDash K_a\lozenge p$,  but that inference is plainly invalid:

\ex. The treasure is buried in Alaska, but Blackbeard doesn't even know it might be there---he's sure it's at sea. 

 While HS-T do not give a semantics for `knows', they will have to contend with this problem (as will any approach built on Yalcin's central idea). While there are various responses available that complicate the semantics of attitude predicates to deal with these cases (e.g., in \citealt{YalcinContext}, \citealt{Rothschild:2011a}, \citealt{Willer:2013}, \citealt{Beddor:2021}), we by contrast will simply not face this problem in the first place. Given an accessibility relation $R_{K_a}$ between possibilities, let $K_a$ be the box operation defined using $R_{K_a}$, just as our $\Box$ is defined using $R$. Then if we consider, e.g., the Epistemic Scale from Example \ref{EpModEx} and set $R_{K_a}$ to be the universal relation between possibilities, representing full ignorance for agent $a$, then  $p\wedge\neg K_a\Diamond p$ will be true at the possibilities $x_1$ and $x_2$ that make $p$ true, so  $p\nvDash K_a\lozenge p$. 
  
 A closely connected point concerns the relation between knowledge and belief. If, as on the expressivist way of thinking, $K_a\lozenge p$ is true whenever $p$ is compatible with $a$'s knowledge state, and $B_a\lozenge p$ is true whenever $p$ is compatible with $a$'s beliefs, then $K_a\lozenge p$ will not entail $B_a\lozenge p$. Instead, $B_a\lozenge p$ will entail $K_a\lozenge p$. For whenever $p$ is compatible with a belief state, it is compatible with the corresponding knowledge state, but not vice versa. By contrast, we can capture the correct entailment relation by simply requiring that $R_{B_a}(x)\subseteq R_{K_a}(x)$, where $R_{B_a}$ is the doxastic accessibility relation for agent $a$, as in standard doxastic-epistemic logic. Then $K_a\lozenge p$ entails $B_a\lozenge p$ but not vice versa.  Together with the  points about logic and probability above, we take the possibility of a smooth extension of our system to factive attitude verbs to speak in favor of it over HS-T.
 
Finally, \citet{Kratzer:2017} develops a theory of epistemic language in a situation-based framework. Her notion of situation is very different from our notion of a possibility (see \citealt{KratzerSituations}), but it is an important precedent for our approach insofar as partiality plays a central role in that theory, as in ours. Having said that, the logic of the theory is classical, with apparent failures of distributivity to be explained pragmatically, so the overall approach is very different.

\section{Conclusion}\label{Conclusion}

Epistemic modals have strange properties. While this has been well appreciated in the recent literature, that literature has tended to take the corresponding logical failures seriously but not literally. Our goal has been to develop a logic and corresponding semantics that directly capture the apparent logical peculiarity of epistemic modals, while retaining as much of classical logic as we can. One central feature of the resulting system is that it is an \textsf{E}-logic: Wittgenstein sentences are contradictions and can always be substituted for contradictions \emph{salva veritate}. That enables our system to account for the wide range of embedding data for epistemic modals, as well as the plausibly closely related data for indicative conditionals. Another  feature of our system is that it invalidates distributivity, whose failure is (to a degree that has not been sufficiently appreciated) central to characterizing the logic of epistemic modals. Not only does distributivity intuitively fail, but invalidating distributivity is also equivalent (relative to a background orthologic) to rejecting the interpretation of negation as pseudocomplementation. That, in turn,  allows us to block the inference from $p\wedge\lozenge\neg p\vDash \bot$ to $\lozenge \neg p\vDash\neg p$, an inference that is valid in a classical setting and that obviously would prevent us from making $p\wedge\lozenge\neg p$ a contradiction. A set of related principles, like disjunctive syllogism and orthomodularity, are also invalid in our system.

Apart from this, however, our system retains much of classical logic. The non-modal fragment is fully classical, as is each fragment consisting of sentences at the same ``epistemic level''; and a wide variety of classical laws, which we have highlighted throughout, remain valid for the entire language. Our goal is to find a minimal variant on classical logic that invalidates just those laws that should fail for epistemic modals. Having said that, we have, again, no proof that we have succeeded; for there may be arguments for logics intermediate in strength between the one we have developed and classical logic. For instance, the principles we identified in Proposition \ref{weakdist} that are valid in the logic of epistemic extensions, but not in \textsf{EO}$^+$, seem plausible to us. 
 We hope that this paper will spur exploration of such logics. 
 
One way to taxonomize the literature on epistemic modality is in terms of where different proposals locate the central revision required by epistemic modals. The fundamental player in the dynamic treatment of Wittgenstein sentences is its non-classical  treatment of conjunction (and, correspondingly, the quantifiers). The domain treatment of Wittgenstein sentences instead has a classical approach to the connectives and holds that what is responsible for the badness of Wittgenstein sentences is fundamentally intensional: the impossibility of an information state \emph{accepting} such a sentence. The state-based approaches start from a similar intuition. The bounded theory, like the dynamic theory, locates the action again in the connectives---in particular in the treatment of conjunction and quantifiers. Our approach, by contrast, highlights a heretofore neglected option: namely, focusing on negation. That is, while we agree with the dynamic and bounded theories that the central action is in the interaction of connectives with epistemic modals, in our theory, negation plays a starring role.  We think this is well motivated, for recall that we can formulate the fundamental problem without even using conjunction: we want $\{\lozenge\neg p,p\}$ to be inconsistent, but we do not want to be able to conclude $\lozenge\neg p\vDash\neg p$, as classical logic would allow us to do. By treating negation as orthocomplementation rather than pseudocomplementation, we block this inference. In fact, conjunction remains ``classical'' in our theory in the sense that it is just set intersection (semantically speaking) and greatest lower bound (algebraically speaking). Of course, conjunction has distinctively non-classical properties in its interaction with disjunction---most prominently, the failure of distributivity. But this is also related to our treatment of negation, since disjunction and conjunction are related by negation via the De Morgan equivalences. 

Our treatment of epistemic modals also brings out the utility of working with possibilities rather than possible worlds as the building blocks of semantics. A central commitment of possible world semantics is that compatibility just is compossibility: for $\varphi$ and $\psi$ to be compatible just is for them to be true somewhere together. While this is also a commitment of possibility semantics for classical logic (recall Remark \ref{BooleanCase}), once we take the step from possible worlds to partial possibilities, we can take the further step of pulling apart compatibility from compossibility. Indeed, pulling apart compatibility from compossibility is \emph{the} central difference between possibility semantics for orthologic vs.~classical logic. And on reflection, it is not clear that compatibility entails compossibility. Romeo and Juliet were compatible but not compossible;  $p$ and $\lozenge \neg p$ are the same. This is why they are jointly contradictory (they are not compossible), but neither entails the negation of the other (they are compatible). These are positions we can hold only if compatibility does not entail compossibility---that is, only in the possibility framework.

We have left many open questions, and there are many points for further development. 
As noted, we need to axiomatize the logic of epistemic extensions of Boolean algebras. We need to  develop a model of credences in epistemic modal propositions, showing how to generalize the notion of a probability measure to a space of possibilities. We need a theory of conditionals in possibility semantics, in particular one that makes sense of the striking commonalities between indicative conditionals and epistemic modals. We also need a semantics for binary quantifiers, where (as noted briefly in \S~\ref{dissubsection}) many interesting issues arise from the interaction with epistemic modals.

A substantial further question to explore is what pragmatic system naturally goes with the possibility semantics we have developed, and, correspondingly, what the companion logic of rational acceptance (``informational logic'') comes to. One natural pragmatic model would go as follows. First,  associate each context with a context set, as in \citealt{Stalnaker:1974}, but treated now as a set of possibilities rather than worlds: namely, the set of possibilities where everything commonly accepted in the conversation is true. Unlike Stalnaker, we cannot model updates intersectively, since conversants may learn $\lozenge p$ and subsequently learn $\neg p$, without this presenting any fundamental informational conflict or tension. However, an apparatus required to model conditionals yields an elegant alternate model of updates. Given a conditional selection function $g$ which takes a proposition $A$ and a possibility $x$ to the set of possibilities closest to $x$ where $A$ is true, we can model conversational updates  as transitions from the context set to its image under $g$, relative to the newly accepted content. That is, if $A$ is asserted and accepted in a conversation whose context set is $C$ and whose selection function is $g$, the updated context set is $C+A=\bigcup\{g(x,A)\mid x\in C \}$. This yields a natural corresponding informational logic:  $\varphi$ informationally entails $\psi$ iff for any context set $C$, we have $x\in \llbracket \psi\rrbracket$ for all $x\in C+\llbracket \varphi\rrbracket$.  Interestingly, if we adopt a theory of conditionals on which $p\to\Box p$ is a theorem, this approach  has an attractive feature: it predicts that  $\varphi$ informationally entails $\Box\varphi$, since any context updated with $\varphi$ by the procedure just defined will entail $\Box\varphi$.  This is in line with the informational logic proposed by \citealt{Yalcin2007} and following, where the inference from $\varphi$ to $\Box\varphi$ is taken to be a central desideratum of an informational logic. Of course, this all depends on developing a theory of conditionals in the possibility framework, a project we hope to take up in future work.
 
A different approach to pragmatics would be to treat the context set  in the first instance as a set of worlds, which is updated in a standard way, and treat the corresponding possibility model as being derived from the context set using the method of epistemic extensions. This approach would need to be further developed to explain how we update with purely modal information. 

 A further set of questions concerns how to think about truth, correctness, agreement, and retraction in a possibility semantic framework. These are related to, but not necessarily determined by, questions about updating. These are difficult problems in the context of modal discourse; among other things, it is not particularly clear what the empirical generalizations are (see \citealt{Khoo14}, \citealt{Khoo:2019}, \citealt{Beddor:2018}, and \citealt{Phillips:2019}, which together suggest that the empirical picture is quite complicated). As far as we can tell, our approach is compatible with a wide range of answers to this question, including contextualist, relativist, and expressivist ones, although the details of a pragmatic theory in any of these veins will need to be developed further.

Whether our system, or further developments of it, ultimately provides the correct account, we hope that future work will take on board our goal of directly characterizing the logical peculiarities of epistemic language, aiming to capture exactly where classical logic goes wrong and where it doesn't.

\subsection*{Acknowledgements}

For helpful feedback, we thank the two reviewers for this journal and our editor, Ivano Ciardelli, as well as Maria Aloni, Alexandru Baltag, Justin Bledin, Scott Blomgren, Kyle Blumberg, Sam Carter, Lucas Champollion, Ahmee Christensen,  Yifeng Ding, J. Dmitri Gallow,  Justin Helms, Alexander Kocurek, John MacFarlane, Guillaume Massas,  Gareth Norman, 
Eric Pacuit, Frank Veltman,  Seth Yalcin, and participants at the LIRa seminar at the University of Amsterdam in February 2022,  the NASSLLI 2022 course on Possibility Semantics, Advances in Modal Logic 2022,  the 2023 MITing of the Minds conference, the  seminar on Sources of Nonclassical Logic at UC Berkeley in Spring 2023, the Tsinghua Logic Salon in May 2023, and the ESSLLI 2023 course on Possibility Semantics. 
This paper extends a proceedings paper (\citealt{HM2022}) presented at the 2022 Amsterdam Colloquium, and we thank our audience there and our two reviewers. 

\appendix

\section{Appendix}

\subsection{Characterization of epistemic extensions of Boolean algebras}\label{CharApp}

In \S~\ref{EpExtSection}, we proved that the epistemic extension of a Boolean algebra is a complete epistemic ortho-Boolean lattice. Now we characterize exactly which complete epistemic ortho-Boolean lattices arise in this way.

\begin{definition}\label{worldly} A complete epistemic ortho-Boolean lattice $L$ with distinguished Boolean algebra $B$ is \textit{extensive} if the following conditions hold:
\begin{enumerate}
\item $L$ is an S5 modal ortholattice;
\item\label{worldly1} for all $x,y\in L$, if $x\not\leq y$, then there are $a,i\in B$ with $0\neq a\leq i$ such that 
\[a\wedge \underset{b\in B,\,0\neq b\leq a}{\bigwedge}\Diamond b \wedge \Box i\leq x\quad\mbox{but}\quad a\wedge \underset{b\in B,\,0\neq b\leq a}{\bigwedge}\Diamond b \wedge \Box i\not\leq y;\]
\item\label{NegLemma} for all $a,i,a',i'\in B$ with $0\neq a\leq i$ and $0\neq a'\leq i$, the following are equivalent:
\begin{enumerate}
\item $a\wedge a'\neq 0$, $a\leq i'$, $a'\leq i$; 
\item $a'\wedge \underset{b\in B,\,0\neq b\leq a'}{\bigwedge}\Diamond b \wedge \Box i' \not\leq \neg (a\wedge \underset{b\in B,\,0\neq b\leq a}{\bigwedge}\Diamond b \wedge \Box i)$;
\end{enumerate}
\item\label{S5lem3} for all $a,i,a',i'\in B$ with $0\neq a\leq i$ and $0\neq a'\leq i$, if $a'\wedge \underset{b\in B,\,0\neq b\leq a'}{\bigwedge}\Diamond b \wedge \Box i'\leq \underset{b\in B,\,0\neq b\leq a}{\bigwedge}\Diamond b \wedge \Box i$, then $a\leq a'$ and $i'\leq i$;
\item\label{worldly5} for all $a,b_k,i\in B$, if $ b_k\leq a$, then $\Diamond b_1\wedge\dots\wedge\Diamond b_n\wedge\Box i\leq \Diamond (a\wedge \Diamond b_1\wedge\dots\wedge\Diamond b_n\wedge\Box i)$.
\end{enumerate}
\end{definition}

Our first task is to show that epistemic extensions of Boolean algebras are indeed extensive.

\begin{proposition} For any Boolean algebra $B$, its epistemic extension $O(B^\mathsf{e})$ is extensive.
\end{proposition}

\begin{proof} Verifying all the parts of Definition \ref{worldly} is straightforward with the exception of part \ref{worldly5}. Let us show that for all $U,V_k,W\in \mathbb{B}_0$,
 if $V_k\subseteq U $, then 
\begin{equation}\Diamond V_1\wedge\dots\wedge\Diamond V_n\wedge\Box W\subseteq \Diamond (U\wedge \Diamond V_1\wedge\dots\wedge\Diamond V_n\wedge\Box W).\label{WideDiamondEq}\end{equation} Suppose  $(a,i)$ is a possibility, so $a\leq i$, which belongs to the left-hand side. Then where $e$ is the embedding of $B$ into $O(B^\mathsf{e})$ from  Theorem \ref{EpistemicFrameExtThm}, we have  $a\wedge e^{-1}( V_k)\neq 0$ for $1\leq k\leq n$ and $i\leq e^{-1}(W)$ by Lemma~\ref{AlgSimpleTruth}.  Now consider any $(a',i')\between (a,i)$, so $a'\leq i'$, $a\wedge a'\neq 0$, $a\leq i'$, and $a'\leq i$. Then where 
\[x= \underset{1\leq i\leq n}{\bigvee}(a\wedge e^{-1}( V_k)),\]
we have $x\leq a\leq  i'$, so $a'\vee x\leq i'$.  From $a'\leq i\leq e^{-1}(W )$, we have $a'\leq e^{-1}(W )$, and from  $a\wedge e^{-1}( V_k)\leq i\leq e^{-1}(W )$, we have $x\leq e^{-1}(W )$, so $a'\vee x\leq e^{-1}(W )$. It follows that $(a',i')R(a'\vee x,i')\between (x,e^{-1}(W ))$ and \[(x,e^{-1}(W ))\in U\wedge \Diamond V_1\wedge\dots\wedge\Diamond V_n\wedge\Box W \] using Lemma \ref{AlgSimpleTruth}, which completes the proof that $(a,i)$ is in the right-hand side of (\ref{WideDiamondEq}).\end{proof}

Conversely, we will show that any extensive $L$ is isomorphic to the epistemic extension of its distinguished Boolean algebra.  We begin with a useful lemma.

\begin{lemma}\label{WorldlyLem}  In any extensive $L$ as in Definition \ref{worldly}, we have the following for all  $a,i\in B$ with $0\neq a\leq i$ and $x\in L$:
\begin{enumerate}
\item \label{S5lem} if $a\wedge \underset{b\in B,\,0\neq b\leq a}{\bigwedge}\Diamond b \wedge\Box i\leq \Box x$, then $\underset{b\in B,\,0\neq b\leq a}{\bigwedge}\Diamond b\wedge\Box i\leq \Box x$;
\item \label{S5lem2} if $\underset{b\in B,\,0\neq b\leq a}{\bigwedge}\Diamond b\wedge\Box i\leq x$, then $\underset{b\in B,\,0\neq b\leq a}{\bigwedge}\Diamond b\wedge\Box i\leq \Box x$.
\end{enumerate}
\end{lemma}
\begin{proof} For part \ref{S5lem}, we have
\begin{eqnarray*}
a\wedge  \underset{b\in B,\,0\neq b\leq a}{\bigwedge}\Diamond b  \wedge\Box i\leq \Box x & \Rightarrow &  \Diamond (a\wedge  \underset{b\in B,\,0\neq b\leq a}{\bigwedge}\Diamond b  \wedge\Box i)\leq \Diamond \Box x\quad\mbox{by monotonicity of }\Diamond \\
& \Rightarrow &   \underset{b\in B,\,0\neq b\leq a}{\bigwedge}\Diamond b  \wedge\Box i\leq \Diamond \Box x \leq \Box x\quad\mbox{by  Definition \ref{worldly}.\ref{worldly5} and S5}.
\end{eqnarray*}
For part \ref{S5lem2},  we have 
\begin{eqnarray*}
 \underset{b\in B,\,0\neq b\leq a}{\bigwedge}\Diamond b \wedge\Box i\leq x &  \Rightarrow & \Box( \underset{b\in B,\,0\neq b\leq a}{\bigwedge}\Diamond b \wedge\Box i)\leq \Box x \quad\mbox{by monotonicity of }\Box \\
&  \Rightarrow & \underset{b\in K}{\bigwedge}\Box\Diamond b\wedge \Box\Box i\leq \Box x \quad\mbox{by multiplicativity of }\Box \\
&  \Rightarrow &  \underset{b\in B,\,0\neq b\leq a}{\bigwedge}\Diamond b \wedge \Box i\leq \Box x \quad\mbox{by S5}.\qedhere
\end{eqnarray*}
\end{proof}

Now we prove the promised result.

\begin{theorem} If $L$ is extensive with distinguished Boolean algebra $B$, then $L$ is isomorphic to  $O(B^\mathsf{e})$. 
\end{theorem}
\begin{proof} Let $B^\mathsf{e}=(S,\between,R)$. Define a map $f: L\to O(B^\mathsf{e})$ by
\[f(x)=\{(a,i)\in S\mid a\wedge \underset{b\in B,\,0\neq b\leq a}{\bigwedge}\Diamond b \wedge \Box i\leq x \}.\]
That $x\not\leq y$ implies $f(x)\not\subseteq f(y)$ follows from Definition \ref{worldly}.\ref{worldly1}, so $f$ is injective.

We begin by showing that $f$ preserves $\neg$. First,  we show that $f(\neg x)\subseteq \neg f(x)$. Suppose $(a,i)\in f(\neg x)$, so 
\[a\wedge \underset{b\in B,\,0\neq b\leq a}{\bigwedge}\Diamond b \wedge \Box i\leq \neg x\]
and hence
\begin{equation}x\leq \neg (a\wedge \underset{b\in B,\,0\neq b\leq a}{\bigwedge}\Diamond b \wedge \Box i).\label{xunder}\end{equation}
Now suppose $(a,i)\between (a',i')$, so $a\wedge a'\neq 0$, $a\leq i'$, and $a'\leq i$.  Then by the (a) to (b) direction of Definition \ref{worldly}.\ref{NegLemma}, we have
\[a'\wedge \underset{b\in B,\,0\neq b\leq a'}{\bigwedge}\Diamond b \wedge \Box i'\not\leq \neg (a\wedge \underset{b\in B,\,0\neq b\leq a}{\bigwedge}\Diamond b \wedge \Box i),\]
which with (\ref{xunder}) implies
\[a'\wedge \underset{b\in B,\,0\neq b\leq a'}{\bigwedge}\Diamond b \wedge \Box i'\not\leq  x,\]
so $(a',i')\not\in f(x)$. It follows that $(a,i)\in \neg f(x)$.

Next we show that $\neg f(x)\subseteq f(\neg x)$. Suppose $(a,i)\not\in f(\neg x)$, so 
\[a\wedge \underset{b\in B,\,0\neq b\leq a}{\bigwedge}\Diamond b \wedge \Box i\not\leq \neg x,\]
which implies 
\[x\not\leq \neg (a\wedge \underset{b\in B,\,0\neq b\leq a}{\bigwedge}\Diamond b \wedge \Box i).\]
Then by part \ref{worldly1} of Definition \ref{worldly}, there are $a',i'\in B$ such that $0\neq a'\leq i'$ and
\[a'\wedge \underset{b\in B,\,0\neq b\leq a'}{\bigwedge}\Diamond b \wedge \Box i'\leq x\]
but 
\[a'\wedge \underset{b\in B,\,0\neq b\leq a'}{\bigwedge}\Diamond b \wedge \Box i'\not\leq \neg (a\wedge \underset{b\in B,\,0\neq b\leq a}{\bigwedge}\Diamond b \wedge \Box i).\]
Then by the (b) to (a) direction of Definition \ref{worldly}.\ref{NegLemma}, we have $(a',i')\between (a,i)$, which with $(a',i')\in f(x)$ implies $(a,i)\not\in \neg f(x)$.

Using the fact that $f$ preserves $\neg$, we claim that $f(x)$ is a $\between$-regular set. Suppose $(a,i)\not\in f(x)$, so
\[a\wedge \underset{b\in B,\,0\neq b\leq a}{\bigwedge}\Diamond b \wedge \Box i\not\leq x\]
and hence
\[\neg x\not\leq \neg (a\wedge \underset{b\in B,\,0\neq b\leq a}{\bigwedge}\Diamond b \wedge \Box i).\]
Then by part \ref{worldly1} of Definition \ref{worldly}, there are $a',i'\in B$ with $0\neq a'\leq i'$ such that 
\[a'\wedge \underset{b\in B,\,0\neq b\leq a'}{\bigwedge}\Diamond b \wedge \Box i'\leq\neg x\]
but 
\[a'\wedge \underset{b\in B,\,0\neq b\leq a'}{\bigwedge}\Diamond b \wedge \Box i'\not\leq \neg (a\wedge \underset{b\in B,\,0\neq b\leq a}{\bigwedge}\Diamond b \wedge \Box i).\]
Then $(a',i')\in  f(\neg x)=\neg f(x)$, and by Definition \ref{worldly}.\ref{NegLemma}, $(a',i')\between (a,i)$. This shows that $f(x)$ is $\between$-regular.

Next  we observe that $f(\underset{x\in X}{\bigwedge}x)= \underset{x\in X}{\bigwedge}f(x)$:
\begin{eqnarray*}
(a,i)\in f(\underset{x\in X}{\bigwedge}x) & \Leftrightarrow & a\wedge \underset{b\in B,\,0\neq b\leq a}{\bigwedge}\Diamond b \wedge \Box i\leq  \underset{x\in X}{\bigwedge}x \\
& \Leftrightarrow & a\wedge \underset{b\in B,\,0\neq b\leq a}{\bigwedge}\Diamond b \wedge \Box i\leq x\mbox{ for all }x\in X \\
& \Leftrightarrow & (a,i)\in f(x)\mbox{ for all }x\in X \\
& \Leftrightarrow & (a,i)\in \underset{x\in X}{\bigwedge} f(x).
\end{eqnarray*}

Next we show that $f(\underset{x\in X}{\bigvee}x)\leq \underset{x\in X}{\bigvee}f(x)$. The converse inclusion follows from order preservation, which follows from meet preservation. Suppose $(a,i)\in  f(\underset{x\in X}{\bigvee}x)$, so
\begin{equation*}a\wedge \underset{b\in B,\,0\neq b\leq a}{\bigwedge}\Diamond b \wedge \Box i\leq \underset{x\in X}{\bigvee}x ,\end{equation*}
which implies 
\begin{equation*}\neg ( \underset{x\in X}{\bigvee}x)\leq \neg (a\wedge \underset{b\in B,\,0\neq b\leq a}{\bigwedge}\Diamond b \wedge \Box i ).\end{equation*}
Further suppose that $(a',i')\between (a,i)$, so $a\wedge a'\neq 0$, $a\leq i'$, and $a'\leq i$. Then by Definition \ref{worldly}.\ref{NegLemma}, we have 
\[a'\wedge \underset{b\in B,\,0\neq b\leq a'}{\bigwedge}\Diamond b \wedge \Box i' \not\leq \neg (a\wedge \underset{b\in B,\,0\neq b\leq a}{\bigwedge}\Diamond b \wedge \Box i)\]
and hence
\[a'\wedge \underset{b\in B,\,0\neq b\leq a'}{\bigwedge}\Diamond b \wedge \Box i' \not\leq \neg ( \underset{x\in X}{\bigvee}x) = \underset{x\in X}{\bigwedge} \neg x,\]
which implies that for some $x\in X$, 
\[a'\wedge \underset{b\in B,\,0\neq b\leq a'}{\bigwedge}\Diamond b \wedge \Box i' \not\leq \neg x.\]
Thus, $(a',i')\not\in f(\neg x)=\neg f(x)$, so there is an $(a'',i'')\between (a',i')$ such that $(a'',i'')\in f(x)$. This shows that $(a,i)\in \underset{x\in X}{\bigvee}f(x)$.

Next we show that $f(\Box x)\subseteq \Box f(x)$. Suppose $(a,i)\in f(\Box x)$, so 
\[a\wedge \underset{b\in B,\,0\neq b\leq a}{\bigwedge}\Diamond b \wedge \Box i\leq \Box x.\]
Then by Lemma \ref{WorldlyLem}.\ref{S5lem}, we have
\[\underset{b\in B,\,0\neq b\leq a}{\bigwedge}\Diamond b \wedge \Box i\leq \Box x.\]
Now if $(a,i)R(a',i')$, so $a\leq a'$ and $i'\leq i$, then $a$ is among the nonzero elements of $B$ below $a'$, and $\Box i'\leq \Box i$, so we have
\[a'\wedge \underset{b\in B,\,0\neq b\leq a'}{\bigwedge}\Diamond b \wedge \Box i'\leq  \underset{b\in B,\,0\neq b\leq a}{\bigwedge}\Diamond b \wedge \Box i \leq \Box x \leq x,\]
so $(a',i')\in f(x)$. This shows that $(a,i)\in \Box f(x)$.

Now to show that $\Box f(x)\subseteq f(\Box x)$, suppose $(a,i)\not\in f(\Box x)$, so 
\[a\wedge \underset{b\in B,\,0\neq b\leq a}{\bigwedge}\Diamond b \wedge \Box i\not\leq \Box x,\]
which by Lemma \ref{WorldlyLem}.\ref{S5lem2} implies 
\[\underset{b\in B,\,0\neq b\leq a}{\bigwedge}\Diamond b \wedge \Box i \not\leq x.\]
Now by part \ref{worldly1} of Definition \ref{worldly}, there are $a',i'\in B$ such that $0\neq a'\leq i'$ and
\begin{equation}a'\wedge \underset{b\in B,\,0\neq b\leq a'}{\bigwedge}\Diamond b \wedge \Box i'\leq \underset{b\in B,\,0\neq b\leq a}{\bigwedge}\Diamond b \wedge \Box i \label{AccessStep}\end{equation}
but 
\[a'\wedge \underset{b\in B,\,0\neq b\leq a'}{\bigwedge}\Diamond b \wedge \Box i'\nleq x,\]
so $(a',i')\not\in f(x)$. By (\ref{AccessStep}) and Definition \ref{worldly}.\ref{S5lem3}, we have $(a,i)R(a',i')$, so we conclude that $(a,i)\not\in \Box f(x)$.

Finally, we show that $f$ is surjective. For $U\in O(B^\mathsf{e})$, where \[x_U = \underset{(a,i)\in U}{\bigvee}\big(a\wedge \underset{b\in B,\,0\neq b\leq a}{\bigwedge}\Diamond b \wedge \Box i\big),\] we claim that $f(x_U)=U$. By join preservation, proved above, we have
\[f(x_U)= \underset{(a,i)\in U}{\bigvee} f\big(a\wedge \underset{b\in B,\,0\neq b\leq a}{\bigwedge}\Diamond b \wedge \Box i\big).\]
Obviously $U\subseteq f(x_U)$, since $(a,i)\in f\big(a\wedge \underset{b\in B,\,0\neq b\leq a}{\bigwedge}\Diamond b \wedge \Box i\big)$ by definition of $f$.   Conversely, suppose $(a^*,i^*)\in f(x_U)$ and $(a',i')\between  (a^*,i^*)$. Then since $(a^*,i^*)\in f(x_U)$, by the previous equation there is an $(a'',i'')\between (a',i')$ such that for some $(a,i)\in U$, \[(a'',i'')\in f\big(a\wedge \underset{b\in B,\,0\neq b\leq a}{\bigwedge}\Diamond b \wedge \Box i\big),\]  which means
\[a''\wedge \underset{b\in B,\,0\neq b\leq a''}{\bigwedge}\Diamond b \wedge \Box i''\leq a\wedge \underset{b\in B,\,0\neq b\leq a}{\bigwedge}\Diamond b \wedge \Box i\]
and therefore 
\[\neg (a\wedge \underset{b\in B,\,0\neq b\leq a}{\bigwedge}\Diamond b \wedge \Box i)\leq \neg (a''\wedge \underset{b\in B,\,0\neq b\leq a''}{\bigwedge}\Diamond b \wedge \Box i'').\]
Hence any element not under the right side is not under the left, which with Definition \ref{worldly}.\ref{NegLemma} implies that any possibility compatible with $(a'',i'')$ is also compatible with $(a,i)$, so $(a,i)\in U$ implies $(a'',i'')\in U$ by Lemma~\ref{RefinementDef}.  Thus, for any $(a',i')\between (a^*,i^*)$, there is an $(a'',i'')\between (a',i')$ with $(a'',i'')\in U$, which with $U\in O(B^\mathsf{e})$ implies $(a^*,i^*)\in U$, which completes the proof that $f(x_U)\subseteq U$.\end{proof}

In the proof, we use the completeness of the ortholattice in only two places: first, so that the meets $\bigwedge \{\Diamond b\mid b\in B,\,0\neq b\leq a\}$ exist in $L$ for each $a\in B$, and second, in the proof that the map $f$ is surjective. Thus, even if we drop the assumption that $L$ is complete, we can prove the following.
 
 \begin{theorem} Let $L$ be an epistemic ortho-Boolean lattice with distinguished Boolean subalgebra $B$, satisfying all the conditions of Definition \ref{worldly} except possibly completeness. If the meet $\bigwedge \{\Diamond b\mid b\in B,\,0\neq b\leq a\}$ exist in $L$ for each $a\in B$, then there is a complete embedding\footnote{I.e., an embedding preserving all existing meets and joins.} of $L$ into $O(B^\mathsf{e})$.
\end{theorem}
Of course if $B$ is finite, then all the relevant meets exist in $L$.

\begin{corollary}\label{FinCor} Let $L$ be an epistemic ortho-Boolean lattice with distinguished \textit{finite} Boolean subalgebra $B$, satisfying all the conditions of Definition \ref{worldly} except possibly completeness. Then there is a complete embedding  of $L$ into $O(B^\mathsf{e})$.\end{corollary}

\subsection{Decidability of the logic of epistemic frames of Boolean algebras}\label{DecideApp}

In this section, we show that consequence over epistemic frames of Boolean algebras is decidable. The key device is a canonical finite model.

\begin{definition} For any finite $P\subseteq\mathsf{Bool}$, a \textit{$P$-state description} is a conjunction of exactly one of $\mathtt{p}$ or $\neg \mathtt{p}$ for each $\mathtt{p}\in P$. Let $W_P$ be the set of all $P$-state descriptions, $V_P$ the valuation sending each $\mathtt{p}\in P$  to the set of $P$-state descriptions in which $\mathtt{p}$ occurs positively, and $\mathcal{M}_P$ the epistemic model of $(\wp(W_P),V_P)$.
\end{definition}

 For $P\subseteq\mathsf{Bool}$, let $\mathcal{EL}^*(P)$ be the fragment of $\mathcal{EL}^+$ whose only proposition letters belong to $P$. Let $\mathbf{EB}$ be the class of epistemic frames of Boolean algebras, and define $\vDash_\mathbf{EB}$ as in Definition \ref{PossCon}.

\begin{proposition}\label{FMP} For $\varphi,\psi\in\mathcal{EL}^*(P)$, if $\varphi\nvDash_\mathbf{EB}\psi$, then there is a possibility in $\mathcal{M}_P$ satisfying $\varphi$ but not~$\psi$.
\end{proposition}

\begin{proof} We outline the main steps, leaving some details to the reader. Suppose $\varphi\nvDash_\mathbf{EB}\psi$, so there is some Boolean algebra $B$ with valuation $\theta$ such that in the epistemic model $\mathcal{M}$ of $(B,\theta)$, we have $\llbracket \varphi\rrbracket^\mathcal{M}\not\subseteq \llbracket \psi\rrbracket^\mathcal{M}$. Let $W_0$ be the set of ultrafilters of $B$, $V_0$ the valuation for $P$ such that $V(\mathtt{p})=\{u\in W_0\mid \theta(\mathtt{p})\in u\}$, and $\mathcal{M}_0$ the epistemic model of $(\wp(W_0),V_0)$. Then since $B$ embeds into $\wp(W_0)$,  $O(B^\mathsf{e})$ embeds into $O(\wp(W_0)^\mathsf{e})$, from which it is easy to show that $\llbracket \varphi\rrbracket^\mathcal{M}\not\subseteq \llbracket \psi\rrbracket^\mathcal{M}$ implies $\llbracket \varphi\rrbracket^{\mathcal{M}_0}\not\subseteq \llbracket \psi\rrbracket^{\mathcal{M}_0}$. Now let $\sim$ be an equivalence relation on $W_0$ defined by $u\sim u'$ if $u\cap \{\theta(\mathtt{p})\mid \mathtt{p}\in P\} = u'\cap \{\theta(\mathtt{p})\mid \mathtt{p}\in P\}$,\footnote{Cf.~the method of \textit{filtration} in modal logic.} and let $[u]$ be the equivalence class of~$u$. Let $W_1$ be the set of equivalence classes of $\sim$, $V_1$ the valuation for $P$ such that $V_0(\mathtt{p})=\{[u]\in W_1\mid \theta(\mathtt{p})\in u\}$, and $\mathcal{M}_1$ the epistemic model of $(\wp(W_1),V_1)$. Then one can prove by induction that for any $\chi\in \mathcal{EL}^*(P)$, a possibility $(A,I)$ in $\mathcal{M}_0$ satisfies $\chi$ iff the possibility $(\{[u]\mid u\in A\}, \{[u]\mid u\in I\})$ in $\mathcal{M}_1$ satisfies $\chi$. Hence $\llbracket \varphi\rrbracket^{\mathcal{M}_0}\not\subseteq \llbracket \psi\rrbracket^{\mathcal{M}_0}$ implies $\llbracket \varphi\rrbracket^{\mathcal{M}_1}\not\subseteq \llbracket \psi\rrbracket^{\mathcal{M}_1}$. Finally, $W_1$ can be identified with a subset of $W_P$: send an equivalence class $[u]$ of ultrafilters to the $P$-state description $s$ such that $\mathtt{p}$ occurs positively in $s$ iff $\theta(\mathtt{p})\in u$. Then $\wp(W_1)$ embeds into $\wp(W_P)$, so $O(\wp(W_1)^\mathsf{e})$ embeds into $O(\wp(W_P)^\mathsf{e})$, from which it is easy to show that   $\llbracket \varphi\rrbracket^{\mathcal{M}_1}\not\subseteq \llbracket \psi\rrbracket^{\mathcal{M}_1}$ implies $\llbracket \varphi\rrbracket^{\mathcal{M}_P}\not\subseteq \llbracket \psi\rrbracket^{\mathcal{M}_P}$.\end{proof}

For the following corollary, let $\mathcal{EL}^*=\mathcal{EL}^*(\mathsf{Bool})$.

\begin{corollary} For  $\varphi,\psi\in\mathcal{EL}^*$, it is decidable whether $\varphi\vDash_\mathbf{EB}\psi$.
\end{corollary}
\begin{proof} Where $P$ is the set of proposition letters appearing in $\varphi$ or $\psi$, by Proposition \ref{FMP} it suffices to check every possibility in the finite model $\mathcal{M}_P$ to see if any satisfies $\varphi$ but not $\psi$.
\end{proof}

\subsection{Normal forms}\label{NormalApp}

Closely related to the previous section is a semantic normal form theorem. Given a Boolean formula $\alpha\in\mathcal{EL}^*(P)$, let $st_P(\alpha)$ be the set of all $P$-state descriptions that entail $\alpha$ in classical logic.

\begin{definition} For any finite $P\subseteq\mathsf{Bool}$ and $\varphi\in \mathcal{EL}^*(P)$, we say that $\varphi$ is in \textit{$P$-canonical normal form} if it is a disjunction of conjunctions of the form
\[\alpha \wedge \underset{\beta\in st_P(\alpha)}{\bigwedge}\Diamond\beta \wedge\Box \gamma,\]
where $\alpha$ and $\gamma$ are disjunctions of $P$-state descriptions.
\end{definition}

\begin{proposition}\label{NormalForm} For any finite $P\subseteq\mathsf{Bool}$ and $\varphi\in \mathcal{EL}^*(P)$, $\varphi$ is equivalent over epistemic frames of Boolean algebras to a formula in $P$-canonical normal form, which can be effectively computed from $\varphi$.
\end{proposition}
\begin{proof} Given a possibility $(A,I)$ in $\mathcal{M}_P$, let $\alpha_A$ be the disjunction of $P$-state descriptions in $A$ and $\gamma_I$ the disjunction of $P$-state descriptions in $I$, respectively. Then by Lemmas \ref{SimpleTruth} and \ref{SimpleRefinement},  $\alpha_A\wedge \underset{\beta\in st_P(\alpha_A)}{\bigwedge}\Diamond\beta \wedge\Box\gamma_I$ is true at all and only the refinements of $(A,I)$ in $\mathcal{M}_P$. It follows that 
\begin{eqnarray*}
\llbracket \varphi\rrbracket^{\mathcal{M}_P}&=& \underset{(A,I)\in \llbracket \varphi\rrbracket^{\mathcal{M}_P}}{\bigcup }\llbracket \alpha_A\wedge \underset{\beta\in st_P(\alpha_A)}{\bigwedge}\Diamond\beta \wedge\Box\gamma_I\rrbracket^{\mathcal{M}_P} \\
&=& \underset{(A,I)\in \llbracket \varphi\rrbracket^{\mathcal{M}_P}}{\bigvee }\llbracket \alpha_A\wedge \underset{\beta\in st_P(\alpha_A)}{\bigwedge}\Diamond\beta \wedge\Box\gamma_I\rrbracket^{\mathcal{M}_P} \quad\mbox{since }\llbracket \varphi\rrbracket^{\mathcal{M}_P}\mbox{ is $\between$-regular}\\
 &=& \llbracket \underset{(A,I)\in \llbracket \varphi\rrbracket^{\mathcal{M}_P}}{\bigvee} (\alpha_A\wedge \underset{\beta\in st_P(\alpha_A)}{\bigwedge}\Diamond\beta \wedge\Box\gamma_I) \rrbracket^{\mathcal{M}_P} .
 \end{eqnarray*}
Then we claim $\varphi$ is equivalent to the formula in normal form in the last line. For if not, then by Lemma~\ref{FMP}, there is a possibility in $\mathcal{M}_P$ in which only one of the formulas is true, contradicting the equation above. Clearly the normal form of $\varphi$ can be effectively computed via the construction of the finite model $\mathcal{M}_P$.
\end{proof}

From Proposition \ref{NormalForm}, we obtain a somewhat unilluminating but nonetheless complete logic with respect to $\vDash_\mathbf{EB}$. The key principles beyond orthologic are $\varphi\vdash NF_P(\varphi)$ and $NF_P(\varphi)\vdash\varphi$, where $P$ contains all proposition letters in $\varphi$ and $NF_P(\varphi)$ is the $P$-canonical normal form of $\varphi$ computed as in the proof of Proposition \ref{NormalForm} (fixing some ordering of $P$ and $W_P$ to get uniqueness of $NF_P(\varphi)$). Then if $\varphi\vDash_\mathbf{EB}\psi$, the proof of Proposition \ref{NormalForm} shows that $NF_P(\psi)$ can be obtained from $NF_P(\varphi)$ by disjunction introduction, so $\varphi\vdash  NF_P(\varphi)\vdash NF_P(\psi)\vdash\psi$. We leave it to future work to give a more illuminating set of logical principles that allow us to derive the equivalence between $\varphi$ and $NF_P(\varphi)$.

\bibliographystyle{plainnat}
\bibliography{modals}

\providecommand{\noopsort}[1]{}\newcommand{\SortNoop}[1]{}
\begin{thebibliography}{98}
\providecommand{\natexlab}[1]{#1}
\providecommand{\url}[1]{\texttt{#1}}
\expandafter\ifx\csname urlstyle\endcsname\relax
  \providecommand{\doi}[1]{doi: #1}\else
  \providecommand{\doi}{doi: \begingroup \urlstyle{rm}\Url}\fi

\bibitem[Aloni(2000)]{Aloni:2000}
Maria Aloni.
\newblock Conceptual covers in dynamic semantics.
\newblock In Lawrence Cavedon, Patrick Blackburn, Nick Braisby, and Atsushi
  Shimojima, editors, \emph{Logic, Language and Computation}, volume III. CSLI,
  2000.

\bibitem[Aloni(2018)]{Aloni:2016}
Maria Aloni.
\newblock {FC} disjunction in state-based semantics.
\newblock Slides for Logical Aspects of Computational Linguistics
  (\emph{LACL}), Nancy, France, 2018.

\bibitem[Aloni et~al.(2022)Aloni, Incurvati, and Schl\"oder]{Aloni:2022}
Maria Aloni, Luca Incurvati, and Julian~J. Schl\"oder.
\newblock Epistemic modals in hypothetical reasoning.
\newblock \emph{Erkenntnis}, 2022.
\newblock \doi{0.1007/s10670-022-00517-x}.

\bibitem[Beddor and Egan(2018)]{Beddor:2018}
Bob Beddor and Andy Egan.
\newblock Might do better: Flexible relativism and the {QUD}.
\newblock \emph{Semantics \& Pragmatics}, 11\penalty0 (7), 2018.
\newblock \doi{10.3765/sp.11.7}.

\bibitem[Beddor and Goldstein(2021)]{Beddor:2021}
Bob Beddor and Simon Goldstein.
\newblock Mighty knowledge.
\newblock \emph{Journal of Philosophy}, 118\penalty0 (5):\penalty0 229--269,
  2021.
\newblock \doi{10.5840/jphil2021118518}.

\bibitem[{\SortNoop{Benthem}}van~Benthem(1996)]{Benthem:1996}
J.~{\SortNoop{Benthem}}van~Benthem.
\newblock \emph{Exploring Logical Dynamics}.
\newblock Center for the Study of Language and Information, Stanford, CA, 1996.

\bibitem[{\noopsort{Benthem}}{van Benthem} et~al.(2017){\noopsort{Benthem}}{van
  Benthem}, Bezhanishvili, and Holliday]{Benthem2017}
Johan {\noopsort{Benthem}}{van Benthem}, Nick Bezhanishvili, and Wesley~H.
  Holliday.
\newblock A bimodal perspective on possibility semantics.
\newblock \emph{Journal of Logic and Computation}, 27\penalty0 (5):\penalty0
  1353--1389, 2017.
\newblock \doi{10.1093/logcom/exw024}.

\bibitem[Bimb\'{o}(2007)]{Bimbo2007}
Katalin Bimb\'{o}.
\newblock Functorial duality for ortholattices and {De Morgan} lattices.
\newblock \emph{Logica Universalis}, 1:\penalty0 311--333, 2007.
\newblock \doi{10.1007/s11787-007-0016-9}.

\bibitem[Birkhoff(1940)]{Birkhoff1940}
G.~Birkhoff.
\newblock \emph{Lattice Theory}.
\newblock American Mathematical Society, New York, 1940.

\bibitem[Birkhoff(1967)]{Birkhoff1967}
G.~Birkhoff.
\newblock \emph{Lattice Theory}.
\newblock American Mathematical Society, Providence, RI, 3rd edition, 1967.

\bibitem[Birkhoff and von Neumann(1936)]{Birkhoff1936}
Garrett Birkhoff and John von Neumann.
\newblock The logic of quantum mechanics.
\newblock \emph{Annals of Mathematics}, 37\penalty0 (4):\penalty0 823--843,
  1936.

\bibitem[Blackburn et~al.(2001)Blackburn, de~Rijke, and Venema]{Blackburn2001}
P.~Blackburn, M.~de~Rijke, and Y.~Venema.
\newblock \emph{Modal Logic}.
\newblock Cambridge University Press, New York, 2001.

\bibitem[Bledin(2014)]{Bledin:2014}
Justin Bledin.
\newblock Logic informed.
\newblock \emph{Mind}, 123\penalty0 (490):\penalty0 277--316, 2014.
\newblock \doi{10.1093/mind/fzu073}.

\bibitem[Bledin(2015)]{Bledin:2015}
Justin Bledin.
\newblock {Modus Ponens} defended.
\newblock \emph{The Journal of Philosophy}, CXII\penalty0 (2):\penalty0 57--83,
  2015.
\newblock \doi{10.5840/jphil201511225}.

\bibitem[Burgess(2003)]{Burgess2003}
John~P. Burgess.
\newblock Which modal models are the right ones (for logical necessity)?
\newblock \emph{Theoria}, 18\penalty0 (47):\penalty0 145--158, 2003.
\newblock \doi{10.1387/theoria.418}.

\bibitem[Chagrov and Zakharyaschev(1997)]{Chagrov1997}
A.~Chagrov and M.~Zakharyaschev.
\newblock \emph{Modal logic}, volume~35 of \emph{Oxford Logic Guides}.
\newblock The Clarendon Press, New York, 1997.

\bibitem[Charlow(2019)]{Charlow:2019}
Nate Charlow.
\newblock Grading modal judgement.
\newblock \emph{Mind}, 129\penalty0 (515):\penalty0 769--807, 2019.
\newblock \doi{10.1093/mind/fzz028}.

\bibitem[Chiara and Giuntini(2002)]{Chiara2002}
Maria Luisa~Dalla Chiara and Roberto Giuntini.
\newblock Quantum logics.
\newblock In Dov Gabbay and Franz Guenthner, editors, \emph{Handbook of
  Philosophical Logic}, pages 129--228. Springer, 2002.

\bibitem[Dishkant(1972)]{Dishkant1972}
H.~Dishkant.
\newblock Semantics of the minimal logic of quantum mechanics.
\newblock \emph{Studia Logica}, 30\penalty0 (4):\penalty0 23--30, 1972.
\newblock \doi{10.1007/BF02120818}.

\bibitem[Dishkant(1977)]{Dishkant1977}
H.~Dishkant.
\newblock Imbedding of the quantum logic in the modal system of {B}rouwer.
\newblock \emph{The Journal of Symbolic Logic}, 42\penalty0 (3):\penalty0
  321--328, 1977.
\newblock \doi{10.2307/2272861}.

\bibitem[Dmitrieva(2021)]{Dmitrieva2021}
Anna Dmitrieva.
\newblock Positive modal logic beyond distributivity: duality, preservation and
  completeness.
\newblock Master's thesis, University of Amsterdam, 2021.

\bibitem[{\SortNoop{Does}}van~der Does et~al.(1997){\SortNoop{Does}}van~der
  Does, Groeneveld, and Veltman]{Does:1997}
Jaap {\SortNoop{Does}}van~der Does, Willem Groeneveld, and Frank Veltman.
\newblock An update on ``might''.
\newblock \emph{Journal of Logic, Language, and Information}, 6:\penalty0
  361--381, 1997.

\bibitem[Dorr and Hawthorne(2013)]{Dorr:2013}
Cian Dorr and John Hawthorne.
\newblock Embedding epistemic modals.
\newblock \emph{Mind}, 122\penalty0 (488):\penalty0 867--913, 2013.
\newblock \doi{10.1093/mind/fzt091}.

\bibitem[Dummett and Lemmon(1959)]{Dummett1959}
M.~A.~E. Dummett and E.~J. Lemmon.
\newblock Modal logics between {S}4 and {S}5.
\newblock \emph{Zeitschrift f\"ur Mathematische Logik und Grundlagen der
  Mathematik}, 5:\penalty0 250--264, 1959.
\newblock \doi{10.1002/malq.19590051405}.

\bibitem[{\SortNoop{Fintel}}von~Fintel and Gillies(2010)]{MustStayStrong}
Kai {\SortNoop{Fintel}}von~Fintel and Anthony Gillies.
\newblock Must...stay...strong!
\newblock \emph{Natural Language Semantics}, 18\penalty0 (4):\penalty0
  351--383, 2010.
\newblock \doi{10.1007/s11050-010-9058-2}.

\bibitem[{\SortNoop{Fintel}}von~Fintel and Gillies(2021)]{Fintel:2016}
Kai {\SortNoop{Fintel}}von~Fintel and Anthony Gillies.
\newblock Still going strong.
\newblock \emph{Natural Language Semantics}, 29:\penalty0 91--113, 2021.
\newblock \doi{10.1007/s11050-020-09171-x}.

\bibitem[Fitch(1952)]{Fitch1952}
Frederic~B. Fitch.
\newblock \emph{Symbolic Logic: An Introduction}.
\newblock The Ronald Press Company, New York, 1952.

\bibitem[Fitch(1966)]{Fitch1966}
Frederic~B. Fitch.
\newblock Natural deduction rules for obligation.
\newblock \emph{American Philosophical Quarterly}, 3\penalty0 (1):\penalty0
  27--38, 1966.

\bibitem[Flocke(2020)]{Flocke:2020}
Vera Flocke.
\newblock Nested epistemic modals.
\newblock Manuscript, Indiana University, Bloomington, July 2020.

\bibitem[Gauker(2020)]{Gauker:2020}
Christopher Gauker.
\newblock A strictly stronger relative \emph{must}.
\newblock \emph{Thought: A Journal of Philosophy}, 10:\penalty0 82--89, 2020.
\newblock \doi{10.1002/tht3.482}.

\bibitem[Gillies(2020)]{Gillies:2018}
Anthony~S. Gillies.
\newblock Updating data semantics.
\newblock \emph{Mind}, 129\penalty0 (513):\penalty0 1--41, 2020.
\newblock \doi{10.1093/mind/fzy008}.

\bibitem[G\"{o}del(1933)]{Godel1933b}
Kurt G\"{o}del.
\newblock Eine {I}nterpretation des intuitionistischen {A}ussagenkalk\"{u}ls.
\newblock \emph{Ergebnisse eines Mathematischen Kolloquiums}, 4:\penalty0
  39--40, 1933.

\bibitem[Goldblatt(1974)]{Goldblatt1974}
Robert~I. Goldblatt.
\newblock Semantic analysis of orthologic.
\newblock \emph{Journal of Philosophical Logic}, 3\penalty0 (1):\penalty0
  19--35, 1974.
\newblock \doi{10.1007/BF00652069}.

\bibitem[Goldblatt(1975)]{Goldblatt1975}
Robert~I. Goldblatt.
\newblock The {S}tone space of an ortholattice.
\newblock \emph{Bulletin of the London Mathematical Society}, 7\penalty0
  (1):\penalty0 45--48, 1975.
\newblock \doi{10.1112/blms/7.1.45}.

\bibitem[Goldstein(2019)]{Goldstein:2018}
Simon Goldstein.
\newblock Generalized update semantics.
\newblock \emph{Mind}, 128\penalty0 (511):\penalty0 795--835, 2019.
\newblock \doi{10.1093/mind/fzy076}.

\bibitem[Goldstein(2021)]{Goldstein:2021}
Simon Goldstein.
\newblock Epistemic modal credence.
\newblock \emph{Philosophers' Imprint}, 26, 2021.

\bibitem[Groenendijk et~al.(1996)Groenendijk, Stokhof, and Veltman]{GSV:1996}
Jeroen Groenendijk, Martin Stokhof, and Frank Veltman.
\newblock Coreference and modality.
\newblock In Shalom Lappin, editor, \emph{Handbook of Contemporary Semantic
  Theory}, pages 179--216. Oxford, Blackwell, 1996.

\bibitem[Grzegorczyk(1964)]{Grzegorczyk1964}
A.~Grzegorczyk.
\newblock A philosophically plausible formal interpretation of intuitionistic
  logic.
\newblock \emph{Indagationes Mathematicae}, 26:\penalty0 596--601, 1964.

\bibitem[Hacking(1967)]{Hacking1967}
Ian Hacking.
\newblock Possibility.
\newblock \emph{Philosophical Review}, 76\penalty0 (2):\penalty0 143--168,
  1967.

\bibitem[Harding(1988)]{Harding1988}
John Harding.
\newblock Varieties of ortholattices containing the orthomodular lattices.
\newblock Master's thesis, McMaster University, 1988.

\bibitem[Hawke and Steinert-Threlkeld(2021)]{Hawke:2020}
Peter Hawke and Shane Steinert-Threlkeld.
\newblock Semantic expressivism for epistemic modals.
\newblock \emph{Linguistics and Philosophy}, 44:\penalty0 475--511, 2021.
\newblock \doi{10.1007/s10988-020-09295-7}.

\bibitem[Hintikka(1962)]{Hintikka2005}
Jaakko Hintikka.
\newblock \emph{Knowledge and Belief: An Introduction to the Logic of the Two
  Notions}.
\newblock Cornell University Press, Ithaca, NY, 1962.

\bibitem[Holliday(2014)]{Holliday2014}
Wesley~H. Holliday.
\newblock Partiality and adjointness in modal logic.
\newblock In R.~Gor\'{e}, B.~Kooi, and A.~Kurucz, editors, \emph{Advances in
  Modal Logic}, volume~10, pages 313--332. College Publications, London, 2014.

\bibitem[Holliday(2021{\natexlab{a}})]{Holliday2021}
Wesley~H. Holliday.
\newblock Three roads to complete lattices: {O}rders, compatibility, polarity.
\newblock \emph{Algebra Universalis}, 82\penalty0 (26), 2021{\natexlab{a}}.
\newblock \doi{10.1007/s00012-021-00711-y}.

\bibitem[Holliday(2021{\natexlab{b}})]{HollidayForthB}
Wesley~H. Holliday.
\newblock Possibility semantics.
\newblock In Melvin Fitting, editor, \emph{Selected Topics from Contemporary
  Logics}, Landscapes in Logic, pages 363--476. College Publications,
  2021{\natexlab{b}}.
\newblock URL \url{https://escholarship.org/uc/item/9ts1b228}.

\bibitem[Holliday(2022)]{Holliday2022}
Wesley~H. Holliday.
\newblock Compatibility and accessibility: lattice representations for
  semantics of non-classical and modal logics.
\newblock In David~Fern\'{a}ndez Duque and Alessandra Palmigiano, editors,
  \emph{Advances in Modal Logic, Vol.~14}, pages 507--529. London, 2022.
\newblock \href{https://arxiv.org/abs/2201.07098}{arXiv:2201.07098 [math.LO]}.

\bibitem[Holliday(2023)]{Holliday2023}
Wesley~H. Holliday.
\newblock A fundamental non-classical logic.
\newblock \emph{Logics}, 1:\penalty0 36--79, 2023.

\bibitem[Holliday(Forthcoming)]{HollidayForthA}
Wesley~H. Holliday.
\newblock Possibility frames and forcing for modal logic.
\newblock \emph{The Australasian Journal of Logic}, Forthcoming.
\newblock URL \url{https://escholarship.org/uc/item/0tm6b30q}.

\bibitem[Holliday and Icard(2010)]{Holliday2010}
Wesley~H. Holliday and Thomas~F. Icard.
\newblock {Moorean Phenomena in Epistemic Logic}.
\newblock In Lev Beklemishev, Valentin Goranko, and Valentin Shehtman, editors,
  \emph{Advances in Modal Logic}, volume~8, pages 178--199. College
  Publications, London, 2010.

\bibitem[Holliday and Icard(2018)]{HollidayIcard2018}
Wesley~H. Holliday and Thomas~F. Icard.
\newblock Axiomatization in the meaning sciences.
\newblock In Derek Ball and Brian Rabern, editors, \emph{The Science of
  Meaning: Essays on the Metatheory of Natural Language Semantics}, pages
  73--97, Oxford, 2018. Oxford University Press.
\newblock \doi{10.1093/oso/9780198739548.003.0002}.

\bibitem[Holliday and Litak(2019)]{Litak2019}
Wesley~H. Holliday and Tadeusz Litak.
\newblock Complete additivity and modal incompleteness.
\newblock \emph{The Review of Symbolic Logic}, 12\penalty0 (3):\penalty0
  487--535, 2019.
\newblock \doi{10.1017/S1755020317000259}.

\bibitem[Holliday and Mandelkern(2022)]{HM2022}
Wesley~H. Holliday and Matthew Mandelkern.
\newblock Compatibility, compossibility, and epistemic modality.
\newblock In Marco Degano, Tom Roberts, Giorgio Sbardolini, and Marieke
  Schouwstra, editors, \emph{Proceedings of the 23rd Amsterdam Colloquium},
  pages 120--126, 2022.

\bibitem[Holliday et~al.(2012)Holliday, Hoshi, and Icard]{HHI2012}
Wesley~H. Holliday, Tomohiro Hoshi, and Thomas~F. Icard.
\newblock A uniform logic of information dynamics.
\newblock In Thomas Bolander, Torben Bra\"{u}ner, Silvio Ghilardi, and Lawrence
  Moss, editors, \emph{Advances in Modal Logic}, volume~9, pages 348--367.
  College Publications, London, 2012.

\bibitem[Holliday et~al.(2013)Holliday, Hoshi, and Icard]{HHI2013}
Wesley~H. Holliday, Tomohiro Hoshi, and Thomas~F. Icard.
\newblock Information dynamics and uniform substitution.
\newblock \emph{Synthese}, 190\penalty0 (1):\penalty0 31--55, 2013.
\newblock \doi{10.1007/s11229-013-0278-0}.

\bibitem[Humberstone(1981)]{Humberstone1981}
L.~Humberstone.
\newblock From worlds to possibilities.
\newblock \emph{Journal of Philosophical Logic}, 10\penalty0 (3):\penalty0
  313--339, 1981.
\newblock \doi{10.1007/BF00293423}.

\bibitem[Incurvati and Schl\"oder(2020)]{Incurvati:2020}
Luca Incurvati and Julian~J. Schl\"oder.
\newblock Epistemic multilateral logic.
\newblock \emph{Review of Symbolic Logic}, 2020.
\newblock \doi{10.1017/S1755020320000313}.

\bibitem[Kaplan(1995)]{Kaplan:1995}
David Kaplan.
\newblock A problem in possible worlds semantics.
\newblock In Walter Sinnott-Armstrong, Diana Raffman, and Nicholas Asher,
  editors, \emph{Modality, Morality, and Belief. Essays in Honour of Ruth
  Barcan Marcus}, pages 41--52. Oxford University Press, 1995.

\bibitem[Khoo(2015)]{Khoo14}
Justin Khoo.
\newblock Modal disagreements.
\newblock \emph{Inquiry}, 58\penalty0 (5):\penalty0 511--534, 2015.

\bibitem[Khoo(2022)]{KhooBook}
Justin Khoo.
\newblock \emph{The Meaning of \emph{If}}.
\newblock Oxford University Press, 2022.

\bibitem[Khoo and Phillips(2019)]{Khoo:2019}
Justin Khoo and Jonathan Phillips.
\newblock New horizons for a theory of epistemic modals.
\newblock \emph{Australasian Journal of Philosophy}, 2019.

\bibitem[Klinedinst and Rothschild(2012)]{KlinedinstRothschildConnectives}
Nathan Klinedinst and Daniel Rothschild.
\newblock Connectives without truth-tables.
\newblock \emph{Natural Language Semantics}, 20:\penalty0 137--175, 2012.
\newblock \doi{10.1007/s11050-011-9079-5}.

\bibitem[Kolibiar(1972)]{Kolibiar1972}
M.~Kolibiar.
\newblock Distributive sublattices of a lattice.
\newblock \emph{Proceedings of the American Mathematical Society}, 34\penalty0
  (2):\penalty0 359--364, 1972.

\bibitem[Kolodny and MacFarlane(2010)]{Kolodny:2010}
Niko Kolodny and John MacFarlane.
\newblock Ifs and oughts.
\newblock \emph{Journal of Philosophy}, 107\penalty0 (3):\penalty0 115--143,
  2010.
\newblock \doi{10.5840/jphil2010107310}.

\bibitem[Kratzer(1977)]{Kratzer1977}
Angelika Kratzer.
\newblock What `must' and `can' must and can mean.
\newblock \emph{Linguistics and Philosophy}, 1\penalty0 (3):\penalty0 337--355,
  1977.
\newblock \doi{10.1007/BF00353453}.

\bibitem[Kratzer(1981)]{Kratzer:1981}
Angelika Kratzer.
\newblock The notional category of modality.
\newblock In H.~Eikmeyer and H.~Rieser, editors, \emph{Words, Worlds, and
  Contexts: New Approaches in Word Semantics}, pages 38--74. de Gruyter, 1981.

\bibitem[Kratzer(1989)]{KratzerSituations}
Angelika Kratzer.
\newblock An investigation of the lumps of thought.
\newblock \emph{Linguistics and Philosophy}, 12(5):\penalty0 607--653, 1989.
\newblock \doi{10.1007/BF00627775}.

\bibitem[Kratzer(2020)]{Kratzer:2017}
Angelika Kratzer.
\newblock What's an epistemic modal anyway?
\newblock Manuscript, University of Massachusetts at Amherst, July 2020.

\bibitem[Kripke(1965)]{Kripke1965}
Saul~A. Kripke.
\newblock Semantical analysis of intuitionistic logic {I}.
\newblock In J.~N. Crossley and M.~A.~E. Dummett, editors, \emph{Formal Systems
  and Recursive Functions}, pages 92--130. North-Holland, Amsterdam, 1965.

\bibitem[Lassiter(2016)]{Lassiter:2016}
Daniel Lassiter.
\newblock {\emph{Must}}, knowledge, and (in)directness.
\newblock \emph{Natural Language Semantics}, 24\penalty0 (2):\penalty0
  117--163, 2016.

\bibitem[MacFarlane(2011)]{MacFarlane:2011}
John MacFarlane.
\newblock Epistemic modals are assessment sensitive.
\newblock In Andy Egan and Brian Weatherson, editors, \emph{Epistemic
  Modality}, pages 144--177. Oxford University Press, 2011.

\bibitem[MacLaren(1964)]{MacLaren1964}
M.~Donald MacLaren.
\newblock Atomic orthocomplemented lattices.
\newblock \emph{Pacific Journal of Mathematics}, 14\penalty0 (2):\penalty0
  597--612, 1964.

\bibitem[Mandelkern(2019)]{Mandelkern:2018a}
Matthew Mandelkern.
\newblock Bounded modality.
\newblock \emph{The Philosophical Review}, 128\penalty0 (1):\penalty0 1--61,
  2019.
\newblock \doi{10.1215/00318108-7213001}.

\bibitem[Mandelkern(2020)]{Mandelkern:2018}
Matthew Mandelkern.
\newblock Dynamic non-classicality.
\newblock \emph{Australasian Journal of Philosophy}, 98\penalty0 (2):\penalty0
  382--392, 2020.
\newblock \doi{10.1080/00048402.2019.1623826}.

\bibitem[Mandelkern(2023)]{Pseudodynamics}
Matthew Mandelkern.
\newblock Bounds: Meaning and the limits of interpretation.
\newblock Manuscript, NYU, January 2023.

\bibitem[McDonald and Yamamoto(2021)]{McDonald2021}
Joseph McDonald and Kentar\^{o} Yamamoto.
\newblock Choice-free duality for orthocomplemented lattices.
\newblock \href{https://arxiv.org/abs/2010.06763}{arXiv:2010.06763 [math.LO]},
  2021.

\bibitem[McGee(1985)]{McGee:1985}
Vann McGee.
\newblock A counterexample to modus ponens.
\newblock \emph{The Journal of Philosophy}, 82\penalty0 (9):\penalty0 462--471,
  1985.
\newblock \doi{10.2307/2026276}.

\bibitem[McKinsey and Tarski(1948)]{McKinsey1948}
J.~C.~C. McKinsey and Alfred Tarski.
\newblock Some theorems about the sentential calculi of {L}ewis and {H}eyting.
\newblock \emph{The Journal of Symbolic Logic}, 13:\penalty0 1--15, 1948.
\newblock \doi{10.2307/2268135}.

\bibitem[Moss(2015)]{Moss2014}
Sarah Moss.
\newblock On the semantics and pragmatics of epistemic vocabulary.
\newblock \emph{Semantics and Pragmatics}, 8\penalty0 (5):\penalty0 1--81,
  2015.
\newblock \doi{10.3765/sp.8.5}.

\bibitem[Ninan(2018)]{Ninan:2017}
Dilip Ninan.
\newblock Quantification and epistemic modality.
\newblock \emph{The Philosophical Review}, 127\penalty0 (2):\penalty0 433--485,
  2018.
\newblock \doi{10.1215/00318108-6973010}.

\bibitem[Norlin(2020)]{Norlin:2019}
Kurt Norlin.
\newblock In the logic of certainty, $\supset$ is stronger than $\rightarrow$.
\newblock \emph{Thought}, 9\penalty0 (1):\penalty0 58--63, 2020.

\bibitem[Norlin(2022)]{Norlin:2020}
Kurt Norlin.
\newblock Acceptance and certainty, doxastic modals, and indicative
  conditionals.
\newblock \emph{Journal of Philosophical Logic}, 2022.
\newblock \doi{10.1007/s10992-022-09656-6}.

\bibitem[Phillips and Mandelkern(2020)]{Phillips:2019}
Jonathan Phillips and Matthew Mandelkern.
\newblock Eavesdropping: What is it good for?
\newblock \emph{Semantics \& Pragmatics}, 13, 2020.
\newblock \doi{10.3765/sp.13.19}.

\bibitem[Rothschild(2011)]{Rothschild:2011a}
Daniel Rothschild.
\newblock Expressing credences.
\newblock \emph{Proceedings of the Aristotelian Society}, 112\penalty0
  (1):\penalty0 99--114, 2011.
\newblock \doi{10.1111/j.1467-9264.2012.00327.x}.

\bibitem[Santorio(2021)]{Santorio:2018a}
Paolo Santorio.
\newblock Trivializing informational consequence.
\newblock \emph{Philosophy and Phenomenological Research}, 2021.
\newblock \doi{10.1111/phpr.12745}.

\bibitem[Stalnaker(1974)]{Stalnaker:1974}
Robert Stalnaker.
\newblock Pragmatic presuppositions.
\newblock In Milton~K. Munitz and Peter Unger, editors, \emph{Semantics and
  Philosophy}, pages 197--213. New York University Press, New York, 1974.
\newblock \doi{10.1093/0198237073.003.0003}.

\bibitem[Stalnaker and Thomason(1970)]{Stalnaker:1970}
Robert~C. Stalnaker and Richmond~H. Thomason.
\newblock A semantic analysis of conditional logic.
\newblock \emph{Theoria}, 36\penalty0 (1):\penalty0 23--42, 1970.
\newblock \doi{10.1111/j.1755-2567.1970.tb00408.x}.

\bibitem[Stojni\'{c}(2016)]{Stojnic:2016}
Una Stojni\'{c}.
\newblock One's \emph{Modus Ponens}: Modality, coherence, and logic.
\newblock \emph{Philosophy and Phenomenological Research}, 93\penalty0
  (3):\penalty0 1--48, 2016.
\newblock \doi{10.1111/phpr.12307}.

\bibitem[Swanson(2016)]{Swanson2016}
Eric Swanson.
\newblock The application of constraint semantics to the language of subjective
  uncertainty.
\newblock \emph{Journal of Philosophical Logic}, 45\penalty0 (2):\penalty0
  121--146, 2016.
\newblock \doi{10.1007/s10992-015-9367-5}.

\bibitem[Tarski(1935)]{Tarski1935}
Alfred Tarski.
\newblock Zur grundlegung der {Boole'schen} algebra. {I}.
\newblock \emph{Fundamenta Mathematicae}, 24:\penalty0 177--198, 1935.

\bibitem[van Wijnbergen-Huitink(2020)]{Wijnbergen:2020}
Janneke van Wijnbergen-Huitink.
\newblock Modal concord.
\newblock In Daniel Gutzmann, Lisa Matthewson, C\'ecile Meier, Hotze Rullman,
  and Thomas~Ede Zimmermann, editors, \emph{The Wiley Blackwell Companion to
  Semantics}. Wiley Online Library, 2020.
\newblock \doi{10.1002/9781118788516.sem129}.

\bibitem[Veltman(1985)]{Veltman85}
Frank Veltman.
\newblock \emph{Logics for Conditionals}.
\newblock PhD thesis, University of Amsterdam, 1985.

\bibitem[Veltman(1996)]{Veltman:1996}
Frank Veltman.
\newblock Defaults in update semantics.
\newblock \emph{Journal of Philosophical Logic}, 25\penalty0 (3):\penalty0
  221--261, 1996.
\newblock \doi{10.1007/BF00248150}.

\bibitem[Willer(2013)]{Willer:2013}
Malte Willer.
\newblock Dynamics of epistemic modality.
\newblock \emph{Philosophical Review}, 122\penalty0 (1):\penalty0 45--92, 2013.
\newblock \doi{10.1215/00318108-1728714}.

\bibitem[Wittgenstein(1953)]{Wittgenstein:1953}
Ludwig Wittgenstein.
\newblock \emph{Philosophical Investigations}.
\newblock Wiley-Blackwell, 2001, 3rd edition, 1953.

\bibitem[Yalcin(2007)]{Yalcin2007}
Seth Yalcin.
\newblock Epistemic modals.
\newblock \emph{Mind}, 116\penalty0 (464):\penalty0 983--1026, 2007.
\newblock \doi{10.1093/mind/fzm983}.

\bibitem[Yalcin(2012{\natexlab{a}})]{Yalcin:2013}
Seth Yalcin.
\newblock Dynamic semantics.
\newblock In Gillian Russell and Delia~Graff Fara, editors, \emph{Routledge
  Companion to Philosophy of Language}, pages 253--279. Routledge, New York,
  2012{\natexlab{a}}.

\bibitem[Yalcin(2012{\natexlab{b}})]{YalcinContext}
Seth Yalcin.
\newblock Context probabilism.
\newblock In Maria Aloni, Vadim Kimmelman, Floris Roelofsen, Galit~W. Sasson,
  Katrin Schulz, and Matthijs Westera, editors, \emph{Logic, Language, and
  Meaning. 18th Amsterdam Colloquium}, volume 7218 of \emph{Lecture Notes in
  Computer Science}, pages 12--21, 2012{\natexlab{b}}.

\bibitem[Yalcin(2015)]{YalcinDeRe}
Seth Yalcin.
\newblock Epistemic modality \textit{de re}.
\newblock \emph{Ergo}, 2\penalty0 (19):\penalty0 475--527, 2015.
\newblock \doi{10.3998/ergo.12405314.0002.019}.

\end{thebibliography}

\end{document}